\documentclass[12pt]{article}
\usepackage[utf8]{inputenc}
\usepackage[english]{babel}
\usepackage{amsmath}
\usepackage{mathrsfs}
\usepackage{amssymb}
\usepackage{amsthm}
\usepackage{amsfonts}
\usepackage{xspace}
\usepackage[normalem]{ulem}
\usepackage{graphicx}
\usepackage{multirow}
\usepackage{booktabs}
\usepackage{appendix}
\usepackage{enumerate}
\usepackage{url} 
\usepackage{authblk}
\usepackage{xcolor}
\usepackage{subcaption}
\usepackage{xr}
\usepackage{adjustbox}
\usepackage{tabularx}
\usepackage{tikz}
\usepackage[colorlinks=true,linkcolor=red,citecolor=blue]{hyperref}
\usepackage{enumitem}
\usepackage{natbib}
\usepackage{iftex}
\usepackage{mathtools}
\usepackage{etoolbox}
\usepackage{algorithm}
\usepackage{algpseudocode}
\usetikzlibrary{shapes, shapes.geometric, arrows.meta, positioning, calc, decorations.pathreplacing, backgrounds, fit}

\definecolor{myblue}{RGB}{70, 130, 180}
\definecolor{myorange}{RGB}{255, 140, 0}
\definecolor{lightbg}{RGB}{248, 250, 255}
\newcommand{\blind}{1}

\makeatletter
\newcommand*{\addFileDependency}[1]{
\typeout{(#1)}
\@addtofilelist{#1}
\IfFileExists{#1}{}{\typeout{No file #1.}}
}\makeatother

\newcommand{\mP}{\mathbb{P}}

\newcommand{\bbeta}{\boldsymbol{\beta}}

\newcommand{\bGamma}{\boldsymbol{\Gamma}}
\newcommand{\bX}{\mathbf{X}}
\newcommand{\bC}{\mathbf{C}}

\newcommand{\bc}{\mathbf{c}}
\newcommand{\bx}{\mathbf{x}}
\newcommand{\bV}{\mathbf{V}}
\newcommand{\bY}{\mathbf{Y}}
\newcommand{\by}{\mathbf{y}}

\newcommand{\calT}{\mathcal{T}}
\newcommand{\calU}{\mathcal{U}}

\newcommand{\ba}{\mathbf{a}}
\newcommand{\mO}{\mathcal{O}}
\newcommand{\bmO}{\boldsymbol{\mathcal{O}}}
\newcommand{\mK}{\mathcal{K}}
\newcommand{\mJ}{\mathcal{J}}

\newcommand{\bM}{\mathbf{M}}
\newcommand{\bm}{\mathbf{m}}
\newcommand{\bv}{\mathbf{v}}
\newcommand{\buj}{\bullet j}
\newcommand{\bumj}{\bullet, -j}

\newcommand{\lb}{\left\{}
\newcommand{\rb}{\right\}}
\newcommand{\nabep}{\nabla_{\epsilon=0}}

\newtheorem{assumption}{Assumption}
\newtheorem{remark}{Remark}
\newtheorem{theorem}{Theorem}
\newtheorem{prop}{Proposition}
\newtheorem{lemma}{Lemma}

\makeatletter
\begin{document}

\def\spacingset#1{\renewcommand{\baselinestretch}%
{#1}\small\normalsize} \spacingset{1}


\if1\blind
{
  \title{\Large \bf Causal mediation in cluster-randomized trials with multiple mediators: spillover-aware decomposition, identification, and semiparametric efficient inference}
\author{
  Jiaqi Tong$^{1}$,
  Chao Cheng$^{2}$,
  and Fan Li$^{1,*}$
  \vspace{0.2cm} 
  \\
  $^{1}$Department of Biostatistics, Yale School of Public Health, New\\ Haven, CT, USA \\
  $^{2}$Department of Statistics and Data Science, Washington University\\ in St. Louis, St. Louis, MO, USA \\
  $^{*}$\emph{email}: fan.f.li@yale.edu
}
  \maketitle
} \fi

\if0\blind
{
  \bigskip
  \bigskip
  \bigskip
  \begin{center}
    {\LARGE\bf Causal mediation in cluster-randomized trials with multiple mediators: spillover-aware decomposition, identification, and semiparametric efficient inference}
\end{center}
  \medskip
} \fi

\bigskip
\begin{abstract}
Causal mediation analysis in cluster-randomized trials (CRTs) is complicated by the presence of multiple mediators, intracluster correlation, and within-cluster interference. 
Existing mediation methods often fall short in accommodating these features simultaneously, and semiparametric efficient estimators that fully address them remain unavailable. We develop a unified framework that defines a class of mediation effect estimands, including exit indirect effects, exit spillover mediation effects, and their interaction effects, to investigate causal mechanisms in CRTs with an arbitrary number of mediators under an unknown causal structure. We introduce a set of interpretable causal assumptions for point identification of each estimand. For optimal inference, we first derive the efficient influence functions for the proposed estimands and construct corresponding one-step and debiased machine learning estimators. In particular, to flexibly model the joint mediator density, we employ an elliptical copula marginal regression model that combines a nonparametric marginal regression with an interpretable association structure. We assess the finite-sample performance of the proposed estimators through simulation studies and illustrate the methodology by reanalyzing the PPACT CRT data with three causally unordered mediators.
\end{abstract}

\noindent%
{\it Keywords:} Multiple mediators, multiple robustness, interaction effects, spillover mediation effects, semiparametric efficiency, elliptical copula marginal regression
\vfill

\newpage
\spacingset{1.8} 

\section{Introduction}
Causal mediation analysis is an effective framework in health and social sciences to explore causal pathways from treatment to outcome through post-treatment intermediate variables, referred to as \emph{mediators}. With a single mediator, the primary step of mediation analysis is to decompose the total treatment effect into a natural indirect effect through the mediator and a natural direct effect not through the mediator, thereby disentangling the role of mediator in explaining the treatment-outcome relationship \citep{imai2010identification}. Although there is a rich literature on causal mediation analysis with independent and identically distributed data \citep{vanderweele2015explanation}, causal mediation methods have only recently been extended to correlated data settings, including multilevel data with individual-level assignment \citep{liu2025estimating} and cluster-randomized trials (CRTs) with cluster-level assignment \citep{vanderweele2013mediation,cheng2024semiparametric}.

Of these settings, causal mediation analysis in CRTs presents distinct methodological challenges and has become an active area of research. That is, investigation of effect pathways is typically complicated by hierarchical data structures that induce both intracluster correlation and within-cluster interference. Intracluster correlation results in the dependence among individuals within the same cluster arising from shared environments or unobserved common factors, whereas interference is reflected by spillover effects wherein the intermediate variable of one individual can potentially influence the outcomes of others in the same cluster. Incorporating within-cluster interference to the context of causal mediation, \cite{vanderweele2010direct} and \cite{vanderweele2013mediation} further decomposed the natural indirect effect into a \emph{spillover mediation effect} to capture the pathway from the treatment to outcome through mediators of other individuals within the same cluster and an \emph{individual mediation effect} to capture the pathway through the individual’s own mediator. Leveraging this finer decomposition, \cite{cheng2024semiparametric} developed semiparametric efficiency theory and proposed robust and efficient estimators for the corresponding causal mediation estimands allowing for spillover. Furthermore, in a setting of two mediators, \cite{ohnishi2024bayesian} decomposed the exit indirect effect estimands \citep{xia2022decomposition} into the \emph{exit spillover mediation effects} and \emph{exit individual mediation effects}, and pursued a Bayesian nonparametric approach for estimation and inference. 

It is not uncommon in CRTs to collect multiple intermediate outcomes, possibly more than two, that may serve as potential mediators of the treatment-outcome relationship. For example, the Pain Program for Active Coping and Training (PPACT) CRT evaluated the impact of a primary care-based cognitive behavioral therapy (CBT) intervention on chronic pain impact, measured by the PEGS scale, among long-term opioid users \citep{debar2022primary}. These effects are hypothesized to operate through three mediators, including the RMDQ score, patient satisfaction with primary care services, and average daily opioid dose, whose causal structure is unknown.
In this context, the Bayesian nonparametric approach of \cite{ohnishi2024bayesian} is not applicable to three mediators, and several key gaps remain. At the conceptual level, with an arbitrary number of $K$ mediators with an unknown causal structure, a suitable set of causal mediation estimands that can address spillover has not been previously characterized in CRTs. At the inferential level, semiparametric efficiency theory and corresponding optimal, multiply robust estimators for such estimands have not been previously established. Last but not the least, at the computational level, the joint density of multiple mediators across individuals within a cluster, a key nuisance function for assessing complex mediation effects in CRTs, also poses unique estimation challenges.

In this paper, we develop a semiparametric framework for mediation analysis in CRTs involving $K\geq 2$ mediators with an unknown causal structure. Our novel contribution is threefold. First, we explicitly characterize a class of causal mediation estimands that accommodate spillover effects in CRTs, and discuss the algebraic architecture of the complex effect decomposition in this context. Our estimands characterization generalizes the special case studied in \cite{ohnishi2024bayesian} from $K=2$ mediators to an arbitrary number of $K\geq 2$ mediators. 
Specifically, we introduce: (i) symmetric interaction indirect mediation effects, which capture interaction among any nonempty subsets of mediators; and (ii) spillover and individual interaction mediation effects, which quantify the contributions from neighborhood spillover and individual pathways, respectively.
Second, we propose a set of structural assumptions for nonparametric point identification of each mediation estimand and discuss their interpretation. Third, for estimation and inference, we derive the efficient influence function and construct semiparametric efficient one-step estimators. In particular, the proposed one-step estimators can be combined with: (i) nuisance estimators based on parametric working models, which provide robustness to certain misspecifications (see Section \ref{Sec:semi-estimation} for full details); and (ii) data-adaptive machine learning algorithms for nuisance estimation, yielding debiased estimators that achieve $\sqrt{I}$-consistency ($I$ is the number of clusters) and semiparametric efficiency under mild regularity conditions. To balance model flexibility and estimation stability, we propose a semiparametric elliptical copula marginal regression approach for modeling the joint distribution of multiple mediators in the presence of clustering, {which also constitutes a useful standalone contribution to the copula literature on multivariate density estimation}. This joint density is a nuisance function that is critical for constructing the semiparametric efficient estimators for our mediation estimands.


The remainder of this paper is organized as follows. Section \ref{ss:notation} introduces the notation and data structure. Section \ref{sec:two-mediators} reviews and refines the causal estimands and the nonparametric identification results in the two-mediator setting. Section \ref{sec:K-mediators} extends these results to $K \geq 2$ mediators. Section \ref{ss:semi-copula} describes the elliptical copula marginal regression model. Section \ref{Sec:semi-estimation} derives efficient influence functions, constructs semiparametric one-step estimators for inference, and establishes their asymptotic properties. Section \ref{Sec:simulation} reports a simulation study. Section \ref{sec:PPACT-application} applies the proposed methods to the PPACT CRT.
Section \ref{ss:concluson} concludes.

\section{Notations and data structure}
\label{ss:notation}
We consider a two-arm CRT with $I$ clusters, where each cluster $i \in[I]$ contains $N_i$ individuals, and $[n]=\{1,\ldots,n\}$ denotes the collection of positive integers up to $n$. Let $A_i$ be the cluster-level treatment, where $A_i = 1$ represents the treated condition and $A_i = 0$ represents the control condition. For each individual $j\in[N_i]$ in cluster $i$, we observe a vector of individual-level baseline covariates $\bX_{ij}\in\mathcal{X}\subseteq \mathbb{R}^{d_X}$ and a vector of cluster-level baseline covariates $\bV_{i}\in\mathcal{V}\subseteq \mathbb{R}^{d_V}$. We also observe an individual-level outcome  $Y_{ij}\in\mathcal{Y}\subseteq\mathbb{R}$ and $K$ individual-level mediators $M^k_{ij}$ for $k \in \mK\equiv[K]$, which are observed post-randomization but prior to the measurement of the outcome. We define $\bX_i = (\bX_{i1}, \ldots, \bX_{iN_i})^\top$ as the design matrix of individual-level covariates for cluster $i$, and $\bC_i = (\bX_i, \mathbf{1}_{N_i}\otimes\bV_i^\top)$ as the design matrix containing all baseline covariates, where $\otimes$ denotes the Kronecker product and $\mathbf{1}_{N_i} = (1, \ldots, 1)^\top$ is a vector of ones of dimension $N_i$. We further define $\bY_i = (Y_{i1}, \ldots, Y_{iN_i})^\top$ as the vector of outcomes for cluster $i$, and $\bM^k_i = (M^k_{i1}, \ldots, M^k_{iN_i})^\top\in\mathcal{M}^k\subseteq\mathbb{R}^{N_i}$ as the vector of mediators. We also denote $\bM_i = ({\bM^1_i}^\top, \ldots,{\bM^K_i}^\top)^\top\in\mathcal{M}\subseteq\mathbb{R}^{K N_i}$ as the complete set of $K$ mediators. In particular, we denote $\bM^k_{i,-j}$ as the sub-vector of $\bM^k_i$ excluding the entry $M^k_{ij}$, and $\bM^{-k}_i$ as the collection of the $K-1$ mediators excluding $\bM^k_i$. We define the vector of observed data for individual $j$ in cluster $i$ as $\mO_{ij} = ( \bX_{ij}^\top, M^1_{ij}, \ldots,M^K_{ij}, Y_{ij} )^\top$, and the observed data vector for cluster $i$ as $\bmO_i = ( A_i, N_i, \bV_i^\top, \mO_{i1}^\top,\ldots,\mO_{iN_i}^\top )^\top$. In a CRT, an analyst observes $I$ independent and identically distributed data vectors $\bmO_i$.

In our motivating PPACT CRT, the causal ordering among the three mediators, the RMDQ score, satisfaction with primary care services, and average daily opioid dose, is unknown because all three are measured at 6 months. In this setting, let $M^k_{ij}(a)$, $\bM^k_i(a)$, $\bM_i(a)$, and $\bM^k_{i,-j}(a)$ denote the counterfactual versions of their observed counterparts under treatment $a$.
Further, we define $Y_{ij}(a,\bm)$ as the counterfactual outcome under treatment  $a$ and mediators $\bm$. 
Following \cite{imai2010identification}, we adopt the composition assumption and define $Y_{ij}(a) \equiv Y_{ij}(a, \bM(a))$. This posits that the counterfactual outcome under treatment $a$ is equivalent to the outcome obtained when the treatment is set to $a$ and all mediators assume the natural values they would take under that treatment assignment. Let $\mathcal{W}_{ij}=(\bX_{ij}^\top, \{M^k_{ij}(a)\}_{a\in\{0,1\},k\in\mK}, \{Y_{ij}(a,\bm)\}_{a\in\{0,1\},\bm\in\mathcal{M}})^\top$ and $\boldsymbol{\mathcal{W}}_i=( A_i, N_i, \bV_i^\top, \mathcal{W}_{i1}^\top,\ldots,\mathcal{W}_{iN_i}^\top )^\top$ denote the full data vectors for individual $j$ within cluster $i$ and for cluster $i$, respectively. Without loss of generality, we often omit subscript $i$ when referring to generic cluster membership; e.g., use $Y_{\buj}$ to denote $Y_{ij}$ for an individual $j$ from a given cluster. 
Finally, let $f(\bullet | A, \bC, N)$ denote the conditional density function of its argument (e.g., the mediators or their subsets) given $\{A,\bC,N\}$. Throughout, the integrals are taken over their corresponding valid supports and are omitted for simplicity.

\section{Estimands and nonparametric identification with $K=2$ mediators: A review and refinement}\label{sec:two-mediators}
{To begin, we review causal mediation estimands with $K=2$ mediators and propose a set of interpretable structural assumptions with several new ingredients tailored to the generalization to $K \geq 2$ mediators and to the ratio estimands presented later in Section \ref{sec:K-mediators}.} We consider a class of weighted average treatment effect $\text{TE}=L(\mu_w(1),\mu_w(0))$,
where for $a\in\{0,1\}$, $\mu_w(a)=E\{ {w(\bV,N)}/{N}\sum_{j=1}^N Y_{\buj}(a)\}/E\{w(\bV,N)\},$
$w \equiv w(\bV,N)$ is a pre-specified weight depending on the cluster-level covariates and cluster sizes, and $L$ is a pre-specified smooth link function. Following \cite{li2025model}, setting $w = N$ yields the \emph{individual-average treatment effect}, while setting $w = 1$ yields the \emph{cluster-average treatment effect}; choosing $L(x, y) = x - y$, $L(x, y) = y^{-1}x$, and $L(x, y) = y^{-1}(1-x)^{-1}x(1-y)$ corresponds to the causal risk difference, risk ratio, and odds ratio, respectively. For ease of exposition, we focus on the cluster-average treatment effect with $w=1$ to demonstrate the mediation effect estimands of interest; the Supplementary Material provides proofs that accommodate general weights $w$. Following \cite{imai2010identification}, we define the \emph{natural indirect effect} (NIE) and the \emph{natural direct effect} (NDE), respectively, as $\text{NDE}=\theta_1(1,0,0)-\theta_1(0,0,0)$ and $\text{NIE}=\theta_1(1,1,1)-\theta_1(1,0,0)$,
where $\theta_1(a_1,a_2,a_3)=E\{{N}^{-1}\sum_{j=1}^NY_{\buj}(a_1,\bM^1(a_2),\bM^2(a_3))\}$ for $a_l\in\{0,1\},l\in\{1,2,3\}$. Here, the NIE quantifies the indirect effect transmitted through the mediators, whereas the NDE measures the direct effect that bypasses them. These components satisfy the decomposition $\text{TE} = \text{NDE} + \text{NIE}$.
To disentangle the role of each mediator, \cite{xia2022decomposition} further decomposed the NIE into three components: two \emph{exit indirect effects} (EIEs), which quantify the indirect effects transmitted through each mediator individually, and an \emph{interaction effect} (INT); for $k\in\{1,2\}$, we have $\text{EIE}_k=\theta_1(1,1,1)-\theta_1(1,k-1,2-k)$ and $\text{INT}=\theta_1(1,1,1)-\theta_1(1,1,0)-\theta_1(1,0,1)+\theta_1(1,0,0)$.
These components satisfy the decomposition $\text{NIE}=\text{EIE}_1+\text{EIE}_2-\text{INT}$. Here, $\text{EIE}_k$ represents the indirect mediation effect transmitted through the $k$-th mediator, while INT captures the interaction effect, reflecting the extent to which the indirect effect through one mediator is modified by the treatment status of the other. 

To account for the neighborhood spillover, \citet{ohnishi2024bayesian} further decomposed EIE into the sum of the \emph{exit spillover mediation effect} (ESME) and the \emph{exit individual mediation effect} (EIME); that is, for $k\in\{1,2\}$, we have $\text{EIE}_k=\text{ESME}_k+\text{EIME}_k$ with $\text{ESME}_k=\theta_1(1,1,1)-\theta_2(k-1,2-k)$ and $\text{EIME}_k=\theta_2(k-1,2-k)-\theta_1(1,k-1,2-k)$,
where $\theta_2(a_1,a_2)={E\{{N}^{-1}\sum_{j=1}^NY_{\buj}(1,M^1_{\buj}(1),\bM^1_{\bullet,-j}(a_1),M^2_{\buj}(1),\bM^2_{\bullet,-j}(a_2))\}}$ for $a_1,a_2\in\{0,1\}$.
Here, $\text{ESME}_k$ and $\text{EIME}_k$ quantify the components of the exit indirect effect through the $k$-th mediator attributable to neighbor individual spillover and to the individual himself, respectively. We next introduce the following identification assumptions.
\begin{assumption}\label{assump:2mediatiors-total}(Super-population framework)
Suppose that the following assumptions hold. (a) For all $ a \in \{0, 1\}$ and $\bm^k,k \in \mK$, representing realizations of $\bM^k$, we have $\bM^k = \bM^k(a)$ under $A = a$, and $Y_{\buj} = Y_{\buj}(a, \bm)$ under $A = a$ and $\bM = \bm$. 
(b) $A \sim \text{Bernoulli}(\pi)$, where $0 < \pi  < 1$, and $A \perp \{N, \bC, \bM(a), Y_{\buj}(a, \bm)\}$ for all $a \in \{0, 1\}$ and $\bm$. 
(c) $\boldsymbol{\mathcal{W}}_1, \ldots, \boldsymbol{\mathcal{W}}_n \overset{i.i.d}{\sim} \mathcal{P}^\mathcal{W}$,
where the orthogonal factorization of $\mathcal{P}^{\boldsymbol{\mathcal{W}}}$ is given by $\mathcal{P}^{\boldsymbol{\mathcal{W}}} = \mathcal{P}^{\boldsymbol{\mathcal{W}}|N} \times \mathcal{P}^N $, satisfying: (i) $\mathcal{P}^N$ has bounded support $\mathcal{S}$; (ii) $\mathcal{P}^{\boldsymbol{\mathcal{W}}|N}$ has bounded second moments; and (iii) conditional on $N$, the individual-level full data vectors, $\mathcal{W}_{\bullet1},\ldots,\mathcal{W}_{\bullet N}$, are marginally identically distributed. Furthermore, there exists a positive constant $f_{\text{min}}$ such that $ f(\bm | A, \bC, N) \geq f_{\text{min}} $ for all  $\bm \in \mathcal{M}$ with probability one.
\end{assumption}

\begin{assumption}
[\emph{Mediator-outcome ignorability}]
\label{assump:sequential-ignorability} For all $j\in[N]$, $a\in \{0, 1\}$, and $\bm$, $Y_{\buj}(a,\bm)\perp\left\{\bM(0),\bM(1)\right\}|A,N,\bC$.
\end{assumption}



Assumption \ref{assump:2mediatiors-total} is standard when conceptualizing a super population framework in CRTs \citep{cheng2024semiparametric}. Importantly, Assumption \ref{assump:2mediatiors-total} (c) accommodates a restricted informative cluster size scenari and permits an arbitrary within-cluster correlation structure. Due to cluster randomization, there is no unmeasured treatment-mediator or treatment-outcome confounding, and hence, Assumption \ref{assump:sequential-ignorability} assumes away the mediator-outcome confounding, which is a critical condition for point identification of the natural indirect effect. The next two assumptions are required for identification of the exit indirect effects.

\begin{assumption}
[\emph{Conditional homogeneity in average potential outcomes}]
\label{assump:homogeneous-between-mediator-ignorability} 

For all $k\in\{1,2\}$ 
and valid $\bm^{k}$, 
\begin{align*}
&E\left\{\overline{Y}(1,\bM^k(0),\bM^{-k}(1))|\bM^{k}(0)=\bm^{k},\bC,N\right\}=E\left\{\overline{Y}(1,\bm^k,\bM^{-k}(1))|\bC,N\right\},
\end{align*}
where $\overline{Y}(a,\bm)=N^{-1}\sum_{j=1}^NY_{\buj}(a,\bm)$ is the cluster-average potential outcome.
\end{assumption}
\begin{assumption}
[\emph{Cross-world mediator conditional independence}]
\label{assump:heterogeneous-between-mediator-ignorability} 
For all $j\in[N]$ and $k\in\{1,2\}$, $\bM^k_{\bumj}(0)\perp \lb M^k_{\buj}(1),\bM^{-k}(1)\rb|\bC,N$.

\end{assumption}

Assumption \ref{assump:homogeneous-between-mediator-ignorability} asserts that the conditional counterfactual cluster-average outcome is exchangeable across all strata defined by the values of the $k$-th potential mediator under control. 
We contrast Assumption \ref{assump:homogeneous-between-mediator-ignorability} with its counterpart (Assumption 5) in \cite{ohnishi2024bayesian}, which directly extends the work of \cite{xia2022decomposition} to CRTs. 
In particular, Assumption 5 of \cite{ohnishi2024bayesian} targets the differences in counterfactual outcomes, $\overline{Y}(1,\bM^k(1),\bM^{-k}(1)) - \overline{Y}(1,\bM^k(0),\bM^{-k}(1))$, whereas Assumption \ref{assump:homogeneous-between-mediator-ignorability} solely focuses on the second component, $\overline{Y}(1,\bM^k(0),\bM^{-k}(1))$. We consider this modification because the first component in $\text{EIE}_k$, $\theta_1(1,1,1)$,
can already be identified under Assumptions \ref{assump:2mediatiors-total}-\ref{assump:sequential-ignorability}. 

Finally, Assumption \ref{assump:heterogeneous-between-mediator-ignorability} is needed to point identify $\text{ESME}_k$ and $\text{EIME}_k$ for the finest effect decomposition. Assumption \ref{assump:heterogeneous-between-mediator-ignorability} says that there are no unmeasured confounders between cross-world counterfactual mediators. This assumption holds if and only if the between-individual cross-world conditional independence assumption in \cite{cheng2024semiparametric} and the between-mediator cross-world conditional independence, i.e., $\bM^k(0) \perp \bM^{-k}(1) | \bC, N$, hold simultaneously. 
Analogous to Assumption \ref{assump:homogeneous-between-mediator-ignorability}, Assumption \ref{assump:heterogeneous-between-mediator-ignorability} differs from its counterpart (Assumption 6) in \cite{ohnishi2024bayesian}. 
To summarize, we adopt different assumptions than \cite{ohnishi2024bayesian} to focus on identification of the essential component of each estimand, thus making it feasible to handle estimands on the ratio scale and to generalize to settings with $K \geq 3$ mediators (see Section \ref{sec:K-mediators}). 
To proceed, we refer to Assumption \ref{assump:2mediatiors-total}
as Set I, Assumptions \ref{assump:2mediatiors-total}–\ref{assump:homogeneous-between-mediator-ignorability} 
as Set II, and Assumptions 
\ref{assump:2mediatiors-total}-\ref{assump:sequential-ignorability}
along with \ref{assump:heterogeneous-between-mediator-ignorability} as Set III, where each successive set of assumptions, from I to III, is needed to point identify a different set of mediation estimands. The following result shows the nonparametric identification formulas for $\theta_1$ and $\theta_2$ with the base $K=2$ setting.

\begin{theorem}\label{thm:iden-NIE}
Under Assumptions Set I, for $(a_1,a_2) \in \{(0,0),(1,0),(1,1)\}$, we have 
\begin{align*}
 \theta_1(a_1,a_2,a_2)=
&{\displaystyle E\lb 
\frac{1}{N}\sum_{j=1}^N\int\eta_{\buj}(a_1, \bm, \bC,N)f(\bm|a_2,\bC,N)d\bm\rb},
\end{align*}
where $\eta_{\buj}(a, \bm, \bC,N) = E\lb Y_{\buj} | A = a, \bM=\bm, \bC, N \rb$ denotes the conditional outcome mean for individual $j$. Under Assumptions Set II, for $(a_1, a_2) \in \{(1,0),(0,1)\}$, we have
\begin{align*}
    \theta_1(1,a_1,a_2)=
    &{\displaystyle E\lb \frac{1}{N}\sum_{j=1}^N\int\eta_{\buj}(1,\bm, \bC,N) \prod_{k=1}^2 f(\bm^k|a_k,\bC,N)d\bm\rb}.
\end{align*}
Finally, under Assumptions Set III, for $k\in\{1,2\}$, we have 
\begin{align*}
\theta_2(k-1,2-k)=
& E\bigg\{ \frac{1}{N}\sum_{j=1}^N\int \eta_{\buj}(1, m_{\buj}^k,\bm^k_{\bullet,-j}, \bm^{-k}, \bC,N)f(\bm^k_{\bullet,-j}|0,\bC,N)f(m_{\buj}^k,\bm^{-k}|1,\bC,N)d\bm\Bigg\}.
\end{align*}
\end{theorem}
Henceforth, we refer to the identification formulas in Theorem \ref{thm:iden-NIE} as \emph{mediation functionals} \citep{tchetgen2012semiparametric}. By Theorem \ref{thm:iden-NIE}, estimands EIE$_k$, INT, ESME$_k$, and EIME$_k$ are point identifiable, as they can be expressed as linear combinations of $\theta_1$ and $\theta_2$.

\section{Generalization to $K\geq 2$ mediators}\label{sec:K-mediators}
{Although causal mediation estimands are well understood in the basic setting with two mediators, the literature for an arbitrary number of mediators remains sparse. The difficulty arises from the complex combinatorial algebra required to preserve the symmetry inherent in the basic setting. Recently, \cite{tong2025permutationinvarianceprinciplecausal} proposed a unified framework for defining \emph{permutation-equivariant} estimands, whose interpretation is invariant to permutations of the mediator labels, and established a new bridge between the literature on factorial trials \citep{mukerjee2006modern} and causal mediation analysis. In this section, we adopt this framework to generalize $\text{EIE}_k$ and $\text{INT}$ in full generality and discuss new insights on practical scalability based on ideas from the factorial trials literature. Moreover, we derive novel decompositions into spillover-aware mediation effects for arbitrary interaction effects and new results for causal ratio estimands, including results that are new even in the basic setting with $K=2$. Finally, for nonparametric identification, we provide novel generalizations of the structural assumptions.}

We first introduce some additional notations to accommodate this more complex but general setting. We define $ \ba = (a_1, \ldots, a_K) $ as a binary tuple of length $ K $, representing the treatment assignment sequence to index $K$ mediators. For $\mJ \subseteq \mK$, we define $\ba_{\mJ}$ as a binary tuple such that $a_k = 0$ for all $k \in \mJ$ and $a_k = 1$ for all $k \notin \mJ$. In particular, $\ba_\emptyset$ is a vector containing all ones, and $\ba_\mK$ is a vector containing all zeros. Let $\bM^\mJ$ denote the collection of mediators indexed by the set $\mJ$ in a given cluster, and let $\bM^\mJ(a)$ represent the counterfactual version of $\bM^\mJ$ under treatment level $a$. We also write $\bM(\ba) = (\bM^1(a_1)^\top, \ldots, \bM^K(a_K)^\top)^\top$ to denote the $K$ potential mediators under arbitrary treatment assignment sequence $\ba=(a_1,\ldots,a_K)$ in a given cluster. Of note, we have equivalent notation $\bM(1)=\bM(\ba_{\emptyset})$, $\bM(0)=\bM(\ba_{\mK})$, $ \bM^k =\bM^{\lb k\rb}$, and $ \bM^{-k} =\bM^{\mK\backslash\lb k\rb}$, where $A \backslash B = A \cap B^c$ denotes the set difference. For example, when $K = 4$ and $\mJ = \{2, 4\}$, we have that $\ba_{\mJ} = (1, 0, 1, 0)$, $\bM^\mJ=(\bM^2,\bM^4)$, $\bM^{\mK\backslash \mJ}=(\bM^1,\bM^3)$, and $\bM^{-1}=(\bM^2,\bM^3,\bM^4)$. 

Based on the general definitions, the decomposition of TE into the sum of the NDE and NIE are analogous to the case of two mediators: 
\begin{align*}
&\text{TE}=\theta_1(1,\ba_\emptyset)-\theta_1(0,\ba_{\mK})=\text{NIE}+\text{NDE},\\
    &\text{NDE}=\theta_1(1,\ba_{\mK})-\theta_1(0,\ba_{\mK}),~~\text{NIE}=\theta_1(1,\ba_\emptyset)-\theta_1(1,\ba_{\mK}),
\end{align*}
where $\theta_1(a^\ast,\ba)={E\left\{{N}^{-1}\sum_{j=1}^NY_{\buj}(a^\ast,\bM(\ba))\right\}}$.
It is important to note that the estimands defined in Section \ref{sec:two-mediators} are permutation-equivariant; that is, their interpretation is invariant to the labeling of the mediators. To maintain this equivariance in the setting of causally unordered mediators, we follow \cite{tong2025permutationinvarianceprinciplecausal} and define $\text{INT}_\mJ$ as the interaction effect among the mediators indexed by any non-empty subset $\emptyset \neq \mJ \subseteq \mK$:
\begin{align*}
    \text{INT}_\mJ=&\sum_{\mJ^\ast\subseteq \mJ}(-1)^{|\mJ^\ast|}\theta_1(1,\ba_{\mJ^\ast}).
\end{align*}
We refer to $\text{INT}_\mJ$ as the $|\mJ|$-th order interaction effect, quantifying the synergistic influence of the mediators within the set $\mJ$. In particular, the first-order interaction effect for a singleton $\lb k\rb$ coincides with $\text{EIE}_k$, that is, $\text{INT}_{\lb k\rb} = \text{EIE}_k= \theta_1(1, \ba_\emptyset) - \theta_1(1, \ba_{\{k\}})$. Hereafter, we use $\text{EIE}_k$ and $\text{INT}_{\{k\}}$ interchangeably. {Notably, there are a total of $2^K-1$ interaction effects, which increases exponentially with the number of mediators. To guide practical use, we borrow the ideas of \emph{hierarchy} and \emph{heredity} from factorial trials \citep{mukerjee2006modern}, since both settings arise from a common combinatorial framework, although their extension to our setting requires substantial new development. Specifically, hierarchy assumes that $\text{INT}_{\mJ_1}$ is more important than $\text{INT}_{\mJ_2}$ only if $|\mJ_1|<|\mJ_2|$, and heredity assumes that, provided $|\mJ_1|<|\mJ_2|$, $\text{INT}_{\mJ_2}$ is important only if $\text{INT}_{\mJ_1}$ is important. For example, under hierarchy and heredity, $\text{EIE}_1$ and $\text{EIE}_2$ are equally important, and both are more important than $\text{INT}$.} 
Similarly, we derive the  decomposition for the NIE as
\begin{align}\label{eq:decomposition-NIE-K}
    \text{NIE}=&\sum_{\emptyset\neq \mJ\subseteq\mK}(-1)^{|\mJ|+1}\text{INT}_\mJ.
\end{align}

Here, the combinatorial structure underlying this decomposition is intimately connected to the \emph{weighted inclusion-exclusion principle} defined in \citet{tong2025permutationinvarianceprinciplecausal}. 
To facilitate interpretation, we provide the following alternative characterization of $\text{INT}_\mJ$: for any $k\in \mJ$
\begin{align}\label{eq:decomposition-INT}
    &\text{INT}_\mJ=\sum_{\mJ^\ast\subseteq \mJ\backslash\lb k\rb} (-1)^{|\mJ^\ast|}\left\{\theta_1(1,\ba_{\mJ^\ast})-\theta_1\left(1,\ba_{\mJ^\ast\cup\lb k\rb}\right)\right\}.
\end{align}
Equation \eqref{eq:decomposition-INT} shows that the $\text{INT}_\mJ$ can be interpreted as the extent to which the effect of any single mediator is modified by the interaction of the remaining $|\mJ|-1$ mediators, with the choice of pivotal mediator being arbitrary. 
Moreover, Equation \eqref{eq:decomposition-INT} motivates an analogous decomposition of $\text{INT}_\mJ$ into two components: one arising from neighborhood spillover, termed the \emph{spillover interaction mediation effect} (SIME), and the other stemming from the individual, termed the \emph{individual interaction mediation effect} (IIME). Formally, it follows that, for any $k\in\mJ$, we obtain $\text{INT}_\mJ=\text{SIME}_\mJ(k)+\text{IIME}_\mJ(k)$ with
\begin{align}
    &\text{SIME}_\mJ(k)=\sum_{\mJ^\ast\subseteq \mJ\backslash\lb k\rb}(-1)^{|\mJ^\ast|}\left\{\theta_1(1,\ba_{\mJ^\ast})-\theta_2(k,\mJ^\ast)\right\},\nonumber\\
    &\text{IIME}_\mJ(k)=\sum_{\mJ^\ast\subseteq \mJ\backslash\lb k\rb}(-1)^{|\mJ^\ast|}\left\{\theta_2(k,\mJ^\ast)-\theta_1\left(1,\ba_{\mJ^\ast\cup\lb k\rb}\right)\right\},\label{eq:SIMEJ-IIMEJ}
\end{align}
where $\theta_2(k,\mJ^\ast)={E\left\{{N}^{-1}\sum_{j=1}^NY_{\buj}(1,\bM^{\mJ^\ast}(0),\bM_{\bumj}^k(0),M_{\buj}^k(1),\bM^{\mK\backslash (\mJ^\ast\cup\lb k\rb)}(1))\right\}}$ refines $\theta_1(1,\ba_{\mJ^\ast})$ by replacing $\bM_{\bumj}^k(1)$ with $\bM_{\bumj}^k(0)$.
To interpret, $\text{IIME}_\mJ(k)$ and $\text{SIME}_\mJ(k)$ respectively measure how the interaction among the other $|\mJ|-1$ mediators modifies the individual and spillover effects of the $k$-th mediator. When $\mJ=\{k\}$ and $K=2$, our decomposition reduces to that of \cite{ohnishi2024bayesian}, with $\text{ESME}_k=\text{SIME}_{\lb k\rb}(k)$ and $\text{EIME}_k=\text{IIME}_{\lb k\rb}(k)$ as a special case; but definition \eqref{eq:SIMEJ-IIMEJ} is much more general. Importantly, our definition goes beyond difference effect measure. For example, with a binary outcome for which the causal estimands are defined on the ratio scale (risk ratio or odds ratio), we provide the definitions as follows.
\begin{remark}
\label{eg:relative-risk}
Let $h_{\text{RR}}(x)=x$ and $h_{\text{OR}}(x)=x/(1-x)$ denote the transformation functions for defining the causal risk ratio and causal odds ratio estimands, respectively. Then, for $h\in\{h_{\text{RR}},h_{\text{OR}}\}$, the TE, NDE, and NIE can be defined as $\text{TE}={h(\theta_1(1,\ba_\emptyset))}/{h(\theta_1(0,\ba_\mK))}$, $\text{NDE}={h(\theta_1(1,\ba_\mK))}/{h(\theta_1(0,\ba_\mK))}$, and $\text{NIE}={h(\theta_1(1,\ba_\emptyset))}/{h(\theta_1(1,\ba_\mK))}$, and they satisfy $\text{TE}=\text{NIE}\times\text{NDE}$.
Moreover, the interaction effects among mediators in $\mJ$ can be defined as $\text{INT}_\mJ=\prod_{\mJ^\ast\subseteq \mJ}h(\theta_1(1,\ba_{\mJ^\ast}))^{(-1)^{|\mJ^\ast|}}$; the decomposition can be expressed as $\text{NIE}=\prod_{\emptyset\neq \mJ\subseteq\mK}\text{INT}_\mJ^{(-1)^{|\mJ|+1}}$.
Finally, to quantify spillover and individual interaction mediation effects, we define $\text{SIME}_\mJ(k)=\prod_{\mJ^\ast\subseteq \mJ\backslash\lb k\rb}\left[\frac{h(\theta_1(1,\ba_{\mJ^\ast}))}{h(\theta_2(k,\mJ^\ast))}\right]^{(-1)^{|\mJ^\ast|}}$ and $\text{IIME}_\mJ(k)=\prod_{\mJ^\ast\subseteq \mJ\backslash\lb k\rb}\left[\frac{h(\theta_2(k,\mJ^\ast))}{h(\theta_1(1,\ba_{\mJ^\ast\cup\lb k\rb}))}\right]^{(-1)^{|\mJ^\ast|}}$,
which satisfies $\text{INT}_\mJ=\text{SIME}_\mJ(k)\times \text{IIME}_\mJ(k)$.
\end{remark}

Algebraically, a logarithmic transformation converts the additive structure of causal mediation effects on the difference scale into a multiplicative structure on the ratio scale. To identify the parameters $\theta_1$ and $\theta_2$ with $K\geq 2$, we continue with Assumption \ref{assump:2mediatiors-total}
but propose the following two assumptions to extend and replace Assumptions \ref{assump:homogeneous-between-mediator-ignorability} and \ref{assump:heterogeneous-between-mediator-ignorability}. 

\begin{assumption}[\emph{Conditional homogeneity in average potential outcomes}]
\label{K-assump:homogeneous-between-mediator-ignorability-INT} 
For all $\mJ^\ast\subseteq \mJ$, $\mJ^\ast\neq\mK,\emptyset$, and $\bm^{\mJ^\ast}$ and $\bm^{\mK \backslash \mJ^\ast}$ as the realizations of $\bM^{\mJ^\ast}$ and $\bM^{\mK\backslash \mJ^\ast}$, respectively, 
$$E\left\{\overline{Y}(1,\bM(\ba_{\mJ^\ast}))|\bM^{\mJ^\ast}(0)=\bm^{\mJ^\ast},\bC,N\right\}=E\left\{\overline{Y}(1,\bm^{\mJ^\ast},\bM^{\mK\backslash \mJ^\ast}(1))|\bC,N\right\}.$$
\end{assumption}

\begin{assumption}
[\emph{Cross-world mediator conditional independence}]
\label{K-assump:generalized-heterogeneous-between-mediator-ignorability} 
$\{ \bM^{\mJ^\ast}(0),\bM_{\bumj}^k(0)\}\perp \{ M_{\buj}^k(1),\bM^{\mK\backslash (\mJ^\ast\cup\{ k\})}(1)\}|\bC,N$ holds for all $j\in[N]$ and $\mJ^\ast\subseteq \mJ\backslash\lb k\rb$.
\end{assumption}

Assumption \ref{K-assump:homogeneous-between-mediator-ignorability-INT} generalizes Assumption \ref{assump:homogeneous-between-mediator-ignorability} and is indexed by $\mathcal{J}$. That is, Assumption \ref{K-assump:homogeneous-between-mediator-ignorability-INT}-$\mJ$ states that conditional on $\{\bC,N\}$, the expectation of the  counterfactual cluster-average outcome $\overline Y(1,\bM(\ba_{\mJ^\ast}))$ remains invariant across strata defined by the counterfactual mediators under control, $\bM^{\mJ^\ast}(0)$. Interestingly, Supplementary Material Section \ref{sec:alternative-GCH-assumption} shows that an alternative assumption produces the same identification formulas. Moreover, to improve interpretability, the cross-world intra-mediator independence assumption, $\bM^{\mJ^\ast}(0) \perp \bM^{\mK\backslash\mJ^\ast}(1)\mid \bC,N$, can replace Assumption \ref{K-assump:homogeneous-between-mediator-ignorability-INT}.
Similarly, Assumption \ref{K-assump:generalized-heterogeneous-between-mediator-ignorability} generalizes Assumption \ref{assump:heterogeneous-between-mediator-ignorability} and is indexed by $(k,\mathcal{J})$, which
states that, two sets of cross-world potential mediators $\bM^{\mJ^\ast}(0)$ and $\{ M_{\buj}^k(1), \bM^{\mK\backslash (\mJ^\ast\cup{ k})}(1) \}$, are independent conditioning on $\{\bC,N\}$. Importantly, Assumption \ref{assump:heterogeneous-between-mediator-ignorability} or Assumption \ref{K-assump:generalized-heterogeneous-between-mediator-ignorability} does not constrain the observed-data likelihood, because it restricts only cross-world mediator distributions while allowing arbitrary correlations within same-world mediator distributions.
To proceed, we let Assumptions \ref{assump:2mediatiors-total} and \ref{K-assump:homogeneous-between-mediator-ignorability-INT} constitute Assumption Set II-$\mathcal{J}$, which facilitates the identification of, $\theta_1^{II}(1, \mathbf{a}_{\mathcal{J}^\ast})$, defined as $\theta_1(1,\mathbf{a}_{\mathcal{J}^\ast})$ for any $\mathcal{J}^\ast \subseteq \mathcal{J}$ such that $\mJ^\ast\notin\{\mK,\emptyset\}$. Similarly, Assumptions \ref{assump:2mediatiors-total} and \ref{K-assump:generalized-heterogeneous-between-mediator-ignorability} constitute Assumption Set III-$(k,\mathcal{J})$, enabling the identification of $\theta_2(k,\mathcal{J}^\ast)$ for all $\mathcal{J}^\ast \subseteq \mathcal{J} \setminus \{k\}$. The theorem below provides nonparametric identification formulas for $\theta_1$ and $\theta_2$ when $K \geq 2$, based on which the mediation estimands can be expressed as additive or multiplicative combinations of $\theta_1$ and $\theta_2$. For simplicity, we denote the estimand identified under Assumption Set I by $\theta_1^{I}(a^\ast, \mathbf{a}_{\mathcal{J}})$, defined as $\theta_1(a^\ast, \mathbf{a}_{\mathcal{J}})$ for $(a^\ast, \mathcal{J}) \in \{(0, \mathcal{K}), (1, \mathcal{K}), (1, \emptyset)\}$. 

\begin{theorem}\label{thm:iden-NIE-K}
Under Assumptions Set I, we have
\begin{align*}
\theta_1^{I}(a^\ast, \mathbf{a}_{\mathcal{J}})=&{\displaystyle E\lb 
 \frac{1}{N}\sum_{j=1}^N\int_{\mathcal{M}}\eta_{\buj}(a^\ast, \bm, \bC,N)f(\bm|\ba_{\mJ1},\bC,N)d\bm\rb},
\end{align*}
where $\ba_{\mJ1}$ denotes the first coordinate of $\ba_{\mJ}$, with $\ba_{\mathcal{K}1} = 0$ and $\ba_{\emptyset1} = 1$. Under Assumptions Set II-$\mJ$, for all $\mJ^\ast\subseteq\mJ$ such that $\mJ^\ast\notin\{\mK,\emptyset\}$, we have 
\begin{align*}
   \theta_1^{II}(1,\ba_{\mJ^\ast})= &{\displaystyle E\lb \frac{1}{N}\sum_{j=1}^N\int\eta_{\buj}(1,\bm, \bC,N) f(\bm^{\mJ^\ast}|0,\bC,N)f(\bm^{\mK\backslash\mJ^\ast}|1,\bC,N)d\bm\rb}.
\end{align*}
Finally, under Assumptions set III-$(k,\mJ)$, for all $\mJ^\ast\subseteq\mJ\backslash\lb k\rb$, we have
\begin{align*}
\theta_2(k,\mJ^\ast)=&{\displaystyle E\left\{ \frac{1}{N}\sum_{j=1}^N\int \eta_{\buj}(1, \bm, \bC,N)f(\bm^{\mJ^\ast},\bm_{\bumj}^k|0,\bC,N)f(m_{\buj}^k,\bm^{\mK\backslash (\mJ^\ast\cup\lb k\rb)}|1,\bC,N)d\bm\right\}}.
\end{align*}
\end{theorem}
The empirical analogue of the moment conditions in Theorem \ref{thm:iden-NIE-K} motivates a set of g-computation estimators $\widehat{\theta}^{\text{g}}$ that standardizes the outcome regression to the target population. In particular, the joint mediator density $f(\bM \mid A, \bC, N)$ serves as a key nuisance function in the identification formulas and we next introduce a specific regression model suitable for the CRT context to operationalize the estimation of this density function.

\section{A copula model for joint mediator density in CRTs}\label{ss:semi-copula}
To enable flexibly modeling the distribution of all mediators in CRTs, we propose a semiparametric elliptical copula marginal regression approach. This approach includes the following features: (i) accommodating mixed-type mediators (continuous and binary, as in our PPACT application study); (ii) allowing for heavy-tailed copula structures and distributions; and (iii) providing interpretable copula association parameters to capture both inter-mediator and within-cluster between-individual dependencies.
To proceed, the joint mediator density can be characterized as two variational-independent components: (i) the one-dimensional marginal mediator densities $\{ f(M_{\buj}^k | A, \bC, N) : j \in [N], k \in \mK \}$ and (ii) the dependence structure captured by the copula $\mathcal{C}(\boldsymbol{u}^1,\ldots,\boldsymbol{u}^K|A,\bC,N)$, where $\boldsymbol{u}^k=(u_{\bullet1}^k,\ldots,u_{\bullet N}^k)^\top,k\in\mK$. By the Sklar's theorem \citep{nelsen2006introduction}, these two components are connected through $  F(\bM | A, \bC, N)=\mathcal{C}(F(M_{\bullet1}^1),\ldots,F(M_{\bullet N}^1),\ldots,F(M_{\bullet 1}^K),\ldots,F(M_{\bullet N}^K)|A,\bC,N)$,
where $F(M_{\buj}^k)$ denotes the cumulative distribution function (CDF) of the mediator $M_{\buj}^k$ conditional on $\{A,\bC,N\}$.
The proposed ECMR model generalizes the Gaussian copula marginal regression model \citep{masarotto2012gaussian} by incorporating the class of elliptical distributions \citep{owen1983class}. The $d$-dimensional elliptical distribution denoted as $\mathcal{E}_d(\boldsymbol{\mu},\boldsymbol{\Omega},g)$, for a random vector $\bX\in\mathbb{R}^d$ is characterized by the density function of the form $f_{\mathcal{E}_d(\boldsymbol{\mu},\boldsymbol{\Omega},g)}(\bx)=\det(\boldsymbol{\Omega})^{-1/2}g((\bx-\boldsymbol{\mu})^\top \boldsymbol{\Omega}^{-1} (\bx-\boldsymbol{\mu}))$,
where $g: \mathbb{R}^{+} \to \mathbb{R}^{+} \cup \{+\infty\}$ is the generator function satisfying the normalizing equation, $\boldsymbol{\mu}\in\mathbb{R}^d$ is the location parameter, and $\boldsymbol{\Omega}\in\mathbb{R}^{d\times d}$ is the scaled covariance matrix. We require $\boldsymbol{\Omega} \succ 0$ so that the covariance matrix is positive definite, where `$\succ$' denotes the Löwner order. The multivariate normal, $t$-, and symmetric Laplace distributions are well-known elliptical family members, with generator functions listed in the Table \ref{tab:elliptical_distributions}. 

\begin{table}[ht!]
\centering
\caption{
Summary of popular members of the elliptical distribution family. Here, $\Gamma(z)$ is the Gamma function and $K_\nu(z)$ is the modified Bessel function of the second kind.}
\label{tab:elliptical_distributions}
\resizebox{0.9\textwidth}{!}{%
\begin{tabular}{lc}
\toprule
\textbf{Elliptical Distribution} & \textbf{Generator Function $g$} \\ \hline
Normal & 
\(\displaystyle g_{\mathcal{N}}(u)=\displaystyle(2\pi)^{-d/2}\exp\left(-{u}/{2}\right)\)\\ [1ex]
$t$ with degree of freedom $\nu$ & 
\(\displaystyle g_t(u;\nu) ={\displaystyle\Gamma\left({d+\nu}/{2}\right)}/{\displaystyle\Gamma\left({\nu}/{2}\right)} (\nu\pi)^{-d/2}\left(1+\frac{u}{\nu}\right)^{-{(\nu+d)}/{2}}\) \\ [0.5ex]
Cauchy & $g_t(u;1)$ \\ [0.5ex]
Symmetric Laplace with $\boldsymbol{\mu}=\mathbf{0}$ &  $g_L(u)={2\displaystyle(2\pi)^{-d/2}}\left({u}/{2}\right)^{(1-d/2)/2}K_{1-d/2}(\sqrt{2u})$\\ 
\bottomrule
\end{tabular}}
\end{table}

The ECMR model then assumes that
\begin{align}
    &M_{\buj}^k=F^{-1}(F_{\mathcal{E}(g)}(\underline{\epsilon}_{\buj}^k)),\label{eq:elliptical-copula-regression}\\
    &\mathcal{C}_{\mathcal{E}_{NK}}=\boldsymbol{F}_{\mathcal{E}_{NK}(\mathbf{0},\boldsymbol{R}, g)}\left(F_{\mathcal{E}(g)}^{-1}(u_{\bullet1}^1),\ldots,F_{\mathcal{E}(g)}^{-1}(u_{\bullet N}^1),\ldots,F_{\mathcal{E}(g)}^{-1}(u_{\bullet 1}^K),\ldots,F_{\mathcal{E}(g)}^{-1}(u_{\bullet N}^K)\right),\label{eq:model-elliptical-copula}
\end{align}
where $F^{-1}$ is the quantile function of $M_{\bullet j}^k$ conditional on $\{A,\bC,N\}$, $\underline{\epsilon}_{\buj}^k\sim\mathcal{E}(g)\equiv\mathcal{E}_1(0,1, g)$ is the error term or latent variable, $F_{\mathcal{E}(g)}$ is the CDF of $\mathcal{E}(g)$, and $\boldsymbol{F}_{\mathcal{E}_{NK}(\mathbf{0},\boldsymbol{R}, g)}$ is the CDF of $\mathcal{E}_{NK}(\mathbf{0},\boldsymbol{\Omega}, g)$ such that $\boldsymbol{\Omega}$ is standardized to be the correlation matrix $\boldsymbol{R}$, i.e., $R_{ij}=\Omega_{ij}/\sqrt{\Omega_{ii}\Omega_{jj}}$. 
Here, Equations \eqref{eq:elliptical-copula-regression} and \eqref{eq:model-elliptical-copula} correspond to components (i) (marginal) and (ii) (dependence structure) of the copula model, respectively, with Equation \eqref{eq:model-elliptical-copula} defining an elliptical copula \citep{derumigny2022identifiability}. Moreover, the mapping between $\underline{\epsilon}_{\bullet j}^k$ and $M_{\bullet j}^k$ in Equation \eqref{eq:elliptical-copula-regression} is one-to-one for continuous mediators, with $\underline{\epsilon}_{\bullet j}^k=F_{\mathcal{E}(g)}^{-1}(F(M_{\bullet j}^k))$, and many-to-one for discrete mediators. In other words, Equation \eqref{eq:elliptical-copula-regression} is conceptually analogous to a one-dimensional degenerate copula or subcopula in that it links the mediator to the error term or latent variable through a CDF transformation. Equation \eqref{eq:elliptical-copula-regression} affords considerable flexibility by leaving the CDF $F$ unspecified, thereby encompassing several prominent classes of semiparametric models. For instance, this framework encompasses: (i) the additive model for continuous mediators, defined as $M^k_{\buj} = \underline{\eta}^k_{\buj}(A, \bC, N) + \underline{\sigma} \underline{\epsilon}_{\buj}^k$, where $\underline{\eta}^k_{\buj}(A, \bC, N) = E\{ M^k_{\buj} | A, \bC, N\}$ and $\underline{\sigma}^2$ is a scale parameter quantifying the residual variability after conditioning on $\{A, \bC, N\}$; and (ii) generalized additive models (e.g., semiparametric logistic or probit regression) for the binary mediators. Finally, for continuous mediators, the error terms follow a multivariate standard elliptical distribution:
\begin{align}\label{eq:alternative-ellipitical copula}
    \boldsymbol{\underline{\epsilon}}\equiv({\boldsymbol{\underline{\epsilon}}^1}^\top,\ldots,{\boldsymbol{\underline{\epsilon}}^K}^\top)^\top\sim\mathcal{E}_{NK}(\mathbf{0},\boldsymbol{R},g),
\end{align}
where $\boldsymbol{\underline{\epsilon}}^k=(\underline{\epsilon}_{\bullet1}^k,\ldots,\underline{\epsilon}_{\bullet N}^k)^\top$, as shown in Supplementary Material Section \ref{supp;ss:proof of proposition 1}. Since the elliptical copula in Equation \eqref{eq:model-elliptical-copula} is identifiable only at the jumps for discrete variables, that is, only as a subcopula, we propose replacing Equation \eqref{eq:model-elliptical-copula} with Equation \eqref{eq:alternative-ellipitical copula} to ensure identifiability. We show in the Supplementary Material that \eqref{eq:alternative-ellipitical copula} implies \eqref{eq:model-elliptical-copula}.

The proposed ECMR model nests the Gaussian copula marginal regression \citep{masarotto2012gaussian} as a special case by setting $g=g_{\mathcal{N}}$, the generator function for the normal distribution.  To align with regression-based methods for CRTs, we assume that copula association is independent of $\{A,\bC,N\}$ and exchangeable across individuals. Under this assumption, the correlation matrix $\boldsymbol{R}$ can be parameterized with $K^2$ parameters and is expressed as: $\boldsymbol{R}(\boldsymbol{\mathcal{Q}} )=(\boldsymbol{\mathcal{Q}}_0-\boldsymbol{\mathcal{Q}}_1)\otimes\boldsymbol{I}_{N}+ \boldsymbol{\mathcal{Q}_1}\otimes \boldsymbol{1}_N \boldsymbol{1}_N^\top$,
where $\boldsymbol{\mathcal{Q}}_0\succ0$ represents the symmetric within-individual correlation matrix 
and $\boldsymbol{\mathcal{Q}}_1$ represents the symmetric unstructured between-individual correlation matrix. 
We require $\boldsymbol{\mathcal{Q}}_0\succ\boldsymbol{\mathcal{Q}}_1$ to guarantee $\boldsymbol{R}(\boldsymbol{\mathcal{Q}} )\succ0$; the between‐individual correlation must be ``smaller'' than the within‐individual correlation. Often, the elements of $\boldsymbol{\mathcal{Q}}\equiv(\boldsymbol{\mathcal{Q}}_0,\boldsymbol{\mathcal{Q}}_1)$ are referred to as intracluster correlation coefficients (ICCs). Then the joint mediator density for the proposed ECMR model satisfies $f(\bM|A,\bC,N)\propto \int_{\mathcal{D}(\bM,\mK_d)} f_{\mathcal{E}_{NK}(\boldsymbol{R},g)}({\underline{\widetilde{\boldsymbol{\epsilon}}}}_{\mK_c},{\underline{\widetilde{\boldsymbol{\epsilon}}}}_{\mK_d})d{\underline{\boldsymbol{\epsilon}}}_{\mK_d}$, where $\mathcal{K}_c$ and $\mathcal{K}_d$ ($\mK=\mathcal{K}_c\cup\mathcal{K}_d$) respectively denote the sets of indices corresponding to the continuous and discrete mediators, $\mathcal{D}(\bM,\mK_d)=\prod_{j=1}^N\prod_{k\in\mK_d}[F_{\mathcal{E}(g)}^{-1}\{F((M_{\buj}^k)^{-})\},F_{\mathcal{E}(g)}^{-1}\{F(M_{\buj}^k)\}]$ is the cartesian product of intervals among discrete mediators, $F((M_{\buj}^k)^{-})$ is the left-limit of $F$ at $M_{\buj}^k$, 
and ${\underline{\boldsymbol{\epsilon}}}_{\mK_c}$ and ${\underline{\boldsymbol{\epsilon}}}_{\mK_d}$ are the vectors of error terms. To improve stability in finite samples and mitigate the curse of dimensionality, one may employ a pre-specified generator $g$, selected based on prior knowledge. For instance, when modeling heavy-tailed data, it is recommended to utilize a generator function corresponding to a $t$-distribution with low degrees of freedom. Subsequently, we estimate the marginal component $\widehat{F}$ in Equation \eqref{eq:elliptical-copula-regression} using data-adaptive machine learning methods, such as SuperLearner with cluster-level cross-validation \citep{luedtke2016super}, and propose a pseudo-likelihood estimator \citep{genest1995semiparametric} for the ICC parameters $\boldsymbol{\mathcal{Q}}$: subject to $\boldsymbol{\mathcal{Q}}_0\succ\boldsymbol{\mathcal{Q}}_1\text{ and }\boldsymbol{\mathcal{Q}}_0\succ0$,
\begin{align}
    \widehat{\boldsymbol{\mathcal{Q}}}=&\underset{{{{\boldsymbol{\mathcal{Q}}}}}}{\arg\max} \prod_{i}^I \int_{\widehat{\mathcal{D}}_i(\bM_i,\mK_d)} f_{\mathcal{E}_{N_iK}(\boldsymbol{R}(\boldsymbol{\mathcal{Q}}),g)}(\widehat{{\underline{{\boldsymbol{\epsilon}}}}}_{\mK_c,i},{\underline{{\boldsymbol{\epsilon}}}}_{\mK_d})d{\underline{\boldsymbol{\epsilon}}}_{\mK_d},\label{eq:MLE-estimation-ECMR}
\end{align}
where 
$\widehat{\mathcal{D}}_i(\bM,\mK_d)=\prod_{j=1}^{N_i}\prod_{k\in\mK_d}[F_{\mathcal{E}(g)}^{-1}\{\widehat{F}((M_{\buj}^k)^{-})\},F_{\mathcal{E}(g)}^{-1}\{\widehat{F}(M_{\buj}^k)\}]$ and $\widehat{\underline{\epsilon}}_{ij}^k=F_{\mathcal{E}(g)}^{-1}(\widehat{F}(M_{\buj}^k))$. 
A more detailed discussion of the ECMR model, including its extension to a data-adaptive generator using sieve methods, and properties that facilitate density calculations for any subset of mediators, is provided in Section \ref{sec:SI-ECMR-SUPP} of the Supplementary Material. 

\section{Semiparametric efficient estimation and inference}\label{Sec:semi-estimation}
Because g-computation estimators $\widehat{\theta}^\text{g}$ can be biased under working model misspecification, a common tool to facilitate robust and semiparametrically efficient inference is the construction of one-step estimators based on efficient influence functions (EIFs) \citep{tsiatis2006semiparametric}. Following \cite{kennedy2022semiparametric}, when derive the EIFs, we place no restrictions on the distribution of the observed data $\bmO_i$, except for the assumption that the treatment probability is known by design. For illustration, we focus on estimating the mediation functionals $\theta_1$ and $\theta_2$, from which all mediation effects, including $\text{TE}$, $\text{NDE}$, $\text{NIE}$, $\text{INT}_{\mJ}$, $\text{SIME}_{\mJ}(k)$, and $\text{IIME}_{\mJ}(k)$, are obtained as they are additive or multiplicative combinations of $\theta_1$ and $\theta_2$. 
To simplify the exposition, let $(k,j)$ denote the index for the $k$-th mediator of individual $j$, and let 
$\mathcal{T} = \mathcal{K}\times [N_i]$ represent the collection of indices for cluster $i$. For any $\mathcal{U}\subseteq \calT$, we define, for any $a_1, a_2, a_3 \in \{0,1\}$:
\begin{align*}
    &\lambda^{(1)}_{\buj}(\mathbf{C}, N;\calU,a_1a_2a_3) = \int \eta_{\buj}(a_1, \mathbf{M}, \mathbf{C}, N) f(\widetilde{\mathbf{M}} | a_2, \mathbf{C}, N) f(\widetilde{\mathbf{M}}^c | a_3, \mathbf{C}, N) \, d\mathbf{M},\\
    &\lambda^{(2)}_{\buj}(\widetilde{\bM}^c,\mathbf{C}, N;\mathcal{U},a_1) = \int \eta_{\buj}(1, \mathbf{M}, \mathbf{C}, N) f(\widetilde{\mathbf{M}} | a_1, \mathbf{C}, N)  \, d\widetilde{\mathbf{M}},\\
&\lambda^{(3)}_{\buj}(\bM,\bC,N;\calU,a_1 a_2 a_3)=\frac{f(\widetilde{\bM}|a_1,\bC,N)f(\widetilde{\bM}^c|a_2,\bC,N)}{f(\bM|a_3,\bC,N)},
\end{align*}
where $\widetilde{\mathbf{M}}$ and $\widetilde{\mathbf{M}}^c$ denote the subvectors of $\mathbf{M}$ indexed by $\mathcal{U}$ and $\mathcal{T} \setminus \mathcal{U}$, respectively. Here, $\lambda^{(1)}_{\buj}$ denotes the integral of the individual conditional outcome mean function, $\eta_{\buj}$, with respect to a composite mediator distribution; $\lambda^{(2)}_{\buj}$ denotes the integral of $\eta_{\buj}$ with respect to the distribution of a subset of mediators under condition $a_1$; $\lambda^{(3)}_{\buj}$ denotes the marginal-to-joint mediator density ratio associated with the partition $\{\mathcal{U}, \mathcal{T} \setminus \mathcal{U}\}$. 

To characterize the EIFs, we introduce the following auxiliary functions: (i) $\eta^\ast_{\buj,a^\ast\ba_{\mJ1}}(\bC,N)=\int \eta_{\buj}(a^\ast,\bm,\bC,N)f(\bm|\ba_{\mJ1},\bC,N)d\bm$; (ii) $r^{k\mJ^\ast a_1a_2 a_3}_{\buj}(\bM,\bC,N)=\lambda^{(3)}_{\buj}(\bM,\bC,N;\calU,a_1 a_2 a_3)$ for $\calU=(\mJ^\ast\cup\{k\}\times[N_i])\backslash \{(k,j)\}$; (iii) $\widetilde{r}^{\mJ^\ast a_1a_2 a_3}_{\buj}(\bM,\bC,N)=\lambda^{(3)}_{\buj}(\bM,\bC,N;\calU,a_1 a_2 a_3)$ for $\calU=\mJ^\ast\times[N_i]$; (iv) $\kappa^{k\mJ^\ast}_{\buj}(\bC,N)=\lambda^{(1)}_{\buj}(\mathbf{C}, N;\calU,1 10)$, $\widetilde{\kappa}_{\buj}^{k\mJ^\ast}(\bM_{\bumj}^k,\bM^{\mJ^\ast},\bC,N)=\lambda^{(2)}_{\buj}(\widetilde{\bM}^c,\mathbf{C}, N;\mathcal{U},1)$, and $\check{r}_{\buj}^{k\mJ^\ast}(M_{\buj}^k,\bM^{\mK\backslash(\mJ^\ast\cup\{k\})},\bC,N)=\lambda^{(2)}_{\buj}(\widetilde{\bM}^c,\mathbf{C}, N;\calT\backslash \mathcal{U},0)$ for $\calU=(\mK\backslash (\mJ^\ast\cup\lb k\rb)\times [N_i])\cup\{(k,j)\}$; (v) $\tau_{\buj}^{\mJ^\ast}(\bC,N)=\lambda^{(1)}_{\buj}(\mathbf{C}, N;\calU,101)$, $\widetilde{\tau}_{\buj}^{\mJ^\ast}(\bM^{J^\ast},\bC,N)=\lambda^{(2)}_{\buj}(\widetilde{\bM}^c,\mathbf{C}, N;\calT\backslash \mathcal{U},1)$, and $\check{\tau}_{\buj}^{\mJ^\ast}(\bM^{\mK \backslash J^\ast},\bC,N)=\lambda^{(2)}_{\buj}(\widetilde{\bM}^c,\mathbf{C}, N;\mathcal{U},0)$ for $\calU=\mJ^\ast\times[N_i]$. Closer inspection reveals that these auxiliary functions are primarily composed of two key nuisance functions: individual conditional outcome mean $\eta_{\buj}$ and the mediator densities conditional on $\{A,\mathbf{C},N\}$. We then obtain the EIFs for both $\theta_1$ and $\theta_2$.
\begin{theorem}\label{thm:EIF}
Under Assumptions Set I, the nonparametric EIF for $\theta^{I}_1(a^\ast,\ba_{\mJ})$, denoted as $\varphi(\bmO;\theta^{I}_1(a^\ast,\ba_{\mJ}))$, is given by
 \begin{align*}
&\varphi(\bmO;\theta^{I}_1(a^\ast,\ba_{\mJ}))=\frac{w(\bV,N)}{N}\sum_{j=1}^N\bigg [\frac{\mathcal{I}(A=a^\ast)}{\Pr(A=a^\ast)}\frac{f(\bM|\ba_{\mJ1},\bC,N)}{f(\bM|a^\ast,\bC,N)}\lb Y_{\buj}-\eta_{\buj}(a^\ast,\bM,\bC,N)\rb+\\
    &\frac{\mathcal{I}(A=\ba_{\mJ1})}{\Pr(A=\ba_{\mJ1})}\lb\eta_{\buj}(a^\ast,\bM,\bC,N)-\eta^{\ast}_{\buj,a^\ast\ba_{\mJ1}}(\bC,N)\rb+\eta^{\ast}_{\buj,a^\ast\ba_{\mJ1}}(\bC,N)\bigg]-\theta_1(a^\ast,\ba_{\mJ}).
\end{align*}
Under Assumptions Set II-$\mJ$, for all $\mJ^\ast\subseteq\mJ$ such that $\mJ^\ast\notin\{\mK,\emptyset\}$, the nonparametric EIF for $\theta^{II}_1(1,\ba_{\mJ^\ast})$, denoted as $\varphi(\bmO;\theta^{II}_1(1,\ba_{\mJ^\ast}))$, is given by
  \begin{align*}   
&\frac{w(\bV,N)}{N}\sum_{j=1}^N\Bigg[\frac{\mathcal{I}(A=0)}{\Pr(A=0)}\widetilde{\tau}^{\mJ^\ast}_{\buj}(\bM^{J^\ast},\bC,N) +\left\{1-\sum_{a=0}^1\frac{\mathcal{I}(A=a)}{\Pr(A=a)}\right\}\tau^{\mJ^\ast}_{\buj}(\bC,N)+\\
    &\frac{\mathcal{I}(A=1)}{\Pr(A=1)}\Bigg[\widetilde{r}^{\mJ^\ast 011}_{\buj}(\bM,\bC,N)
    \left\{ Y_{\buj}-\eta_{\buj}(1,\bM,\bC,N)\right\}+\check{\tau}^{\mJ^\ast}_{\buj}(\bM^{\mK\backslash J^\ast},\bC,N)\Bigg]\Bigg]-\theta_1(1,\ba_{\mJ^\ast}).
\end{align*}   
Finally, under Assumptions set III-$(k,\mJ)$, for all $\mJ^\ast\subseteq\mJ\backslash\lb k\rb$, the nonparametric EIF for $\theta_2(k,\mJ^\ast)$, denoted as $\varphi(\bmO;\theta_2(k,\mJ^\ast))$, is given by
   \begin{align*}     
&\frac{w(\bV,N)}{N}\sum_{j=1}^N\Bigg[\frac{\mathcal{I}(A=0)}{\Pr(A=0)}\widetilde{\kappa}_{\buj}^{k\mJ^\ast}(\bM_{\bumj}^k,\bM^{\mJ^\ast},\bC,N) +\left\{1-\sum_{a=0}^1\frac{\mathcal{I}(A=a)}{\Pr(A=a)}\right\}\kappa^{k\mJ^\ast}_{\buj}(\bC,N)+\\
    &\frac{\mathcal{I}(A=1)}{\Pr(A=1)}\left[r^{k\mJ^\ast011}_{\buj}(\bM,\bC,N)\left\{ Y_{\buj}-\eta_{\buj}(1,\bM,\bC,N)\right\}+\check{\kappa}_{\buj}^{k\mJ^\ast}(M_{\buj}^k,\bM^{\mK\backslash(\mJ^\ast\cup\{k\})},\bC,N)\right]\Bigg]-\theta_2(k,\mJ^\ast).
\end{align*}
\end{theorem}

To compute the one-step estimators based on EIFs, we would need to model the nuisance functions, specifically $\eta_{\buj}$ and the conditional mediator densities. To address the random dimensions of individual-level mediators and covariates in estimating $\eta_{\buj}$, we assume that ${\eta}_{\buj}(A, \bM, \bC, N) \equiv {\eta}_{\buj}(A, \xi_{\buj}(\bM), \xi^\ast_{\buj}(\bX), \bV, N)$, where $\xi_{\buj}(\bM)$ and $\xi^\ast_{\buj}(\bX)$ are known, non-random-dimensional mappings that capture the effect mechanisms of the individual-level mediators and individual-level covariates, respectively. 
The above model assumption is weak, as the functional forms of $\xi_{\buj}(\bM)$ and $\xi^\ast_{\buj}(\bX)$ remain unspecified. For example, $\xi_{\buj}(\bM) = \bM_{\buj}$ and $\xi^\ast_{\buj}(\bX) = \bX_{\buj}$ if each individual’s outcome depends only on their own mediators and covariates, with no within-cluster interference effects at the mediator level or covariate level. In contrast, if a within-cluster interference effect exists and is separable, then $\xi_{\buj}(\bM) = (\bM_{\buj}, \widetilde{\xi}_{\buj}(\bM)^\top)$ and $\xi^\ast_{\buj}(\bX) = (\bX_{\buj}, \widetilde{\xi}^\ast_{\buj}(\bX)^\top)$, where the non-random-dimensional functions $\widetilde{\xi}_{\buj}(\bM)$ and $\widetilde{\xi}^\ast_{\buj}(\bX)$ capture the paths of within-cluster interference effects (e.g., average values of mediators and covariates over neighbor individuals within the same cluster) of mediators and individual-level covariates, respectively. 

To enable robust and efficient inference, we propose one-step estimators for $\theta \in \{ \theta^{I}_1(a^\ast, \ba_{\mJ}),\allowbreak\theta^{II}_1(1, \ba_{\mJ^\ast}), \theta_2(k, \mJ^\ast) \}$
based on the EIFs in Theorem \ref{thm:EIF}. For ease of presentation, we simplify the notation for conditional mediator densities: (i) $f\equiv f(\bM|A,\bC,N)$, $f_{\mJ^\ast}\equiv f(\bM^{\mJ^\ast}|A,\bC,N)$, and $f_{\mK\backslash\mJ^\ast}\equiv f(\bM^{\mK\backslash\mJ^\ast}|A,\bC,N)$; and (ii) $f_{\bumj}\equiv f(\bM^{\mJ^\ast},\bM^k_{\bullet,-j}|A,\bC,N)$ and $f_{\buj}\equiv f(M_{\buj}^k,\bM^{\mK\backslash (\mJ^\ast\cup\lb k\rb)}|A,\bC,N)$. We adopt the notation $f^{a} \equiv f(\mathbf{M} | a, \mathbf{C}, N)$ to emphasize the dependence on the treatment value $a$; analogous notation is applied to other nuisance functions. Let $\bGamma(\theta)$ denote the collection of nuisance functions for a generic mediation functional $\theta$. Specifically, $\bGamma(\theta^{I}_1(a^\ast, \ba_{\mJ}))=\{\eta_{\buj},f^{\ba_{\mJ1}},f^{a^\ast}:j\in [N]\}$, $\bGamma(\theta^{II}_1(1,\ba_{\mJ^\ast}))=\{\eta_{\buj},f, f_{\mJ^\ast},f_{\mK\backslash\mJ^\ast}:j\in [N]\}$, and $\bGamma(\theta_2(k,\mJ^\ast))=\{\eta_{\buj},f, f_{\bumj},f_{\buj}:j\in [N]\}$. We write $\varphi(\bmO; \theta) = \varphi(\bmO; \theta, \mathbf{\Gamma}(\theta))$ to explicitly indicate the dependence of the EIF on the set of nuisance functions $\mathbf{\Gamma}(\theta)$, and denote the  estimated nuisance functions by $\widehat{\mathbf{\Gamma}}$.

Then the semiparametric one-step estimator is formally given by 
$$\widehat{\theta}=\mP_I\{\phi(\theta;\widehat{\bGamma
}(\theta))\},$$
where $\phi(\theta)=\varphi(\theta)+\theta$ in the uncentered EIF and $\mP_I\lb V \rb = I^{-1} \sum_{i=1}^I V_i$ denotes the empirical average over the clusters. Here, the one-step estimator can be implemented using nuisance functions estimated via parametric working models; the resulting estimator, denoted by $\widehat{\theta}^{\text{EIF.PAR}}$, is robust to certain misspecifications of the working models (precise robustness properties are detailed subsequently). Alternatively, the estimator can be implemented using nuisance functions estimated via data-adaptive machine learning algorithms, such as the Super Learner with cluster-level cross-validation \citep{luedtke2016super}; this resulting \emph{debiased machine learning} estimator is denoted by $\widehat{\theta}^{\text{EIF.DML}}$. To mitigate overfitting bias, we follow \cite{dml} by implementing $\widehat{\theta}^{\text{EIF.DML}}$ using sample cross-fitting.
To proceed, the sample is randomly partitioned at the cluster level into $S$ independent folds of approximately equal size.
Let $\bmO^s$ be the collection of observed data in $s$-th fold,
and $I_s$ be the number of clusters in $\bmO^s$. For each fold, we estimate the nuisance functions $\bGamma(\theta)$ using the training sample $\bmO^{-s} \equiv \bmO\backslash\bmO^s$. The resulting estimated nuisance functions are denoted by $\widehat{\bGamma}^s(\theta)$. For each testing sample $\bmO^s$ and the corresponding empirical measure, $\mathbb{P}_{s}(V)\equiv I_s^{-1}\sum_{i
\in\bmO^s} V_i$, we define $\widehat{\phi}^s(\bmO;\theta)=\mathbb{P}_s\{ \phi(\bmO;\theta,\widehat{\bGamma}^s(\theta))\}$. Finally, we obtain that $\widehat{\theta}^{\text{EIF.DML}}=1/I{\sum_{s=1}^SI_s\widehat{\phi}^s(\bmO;\theta)}$. The following theorem summarizes the large-sample property of the semiparametric one-step estimator $\widehat{\theta}$.

\begin{theorem}\label{thm:dml} Assume that the following regularity conditions hold:
(i) $|w(\bV, N)|$ is uniformly bounded;
(ii) $f$ is uniformly bounded away from zero and $\{\widehat{\eta}, f_{\buj},f_{\mK\backslash\mJ^\ast}, \widehat{f}_{\buj}, \widehat{f}_{\mK\backslash\mJ^\ast},\widehat{f}_{\bumj},\allowbreak\widehat{f}_{\mJ^\ast},\widehat{f}:j\in[N]\}$ are uniformly bounded with probability one; (iii) $E[\{Y_{\buj}-\eta_{\buj}(a,\bM,\bC,N)\}^2|\allowbreak a,\bM,\bC,N]<\infty$ uniformly for all $a\in\{0,1\}$ and $j\in[N]$ with probability one; and (iv) conditional on $N$, $\lVert \widehat{\boldsymbol{\Gamma}}(\theta)-\boldsymbol{\Gamma}(\theta)\rVert_{2,\max}=o_\mP(1)$ with probability one, where $\lVert \boldsymbol{q}\rVert_{2,\max}=\max_r\lVert q_r\rVert_{2}$, $q_r$ is the $r$th component of $\boldsymbol{q}$, and $\lVert \bullet\rVert_2$ is $\mathcal{L}^2(\mP)$ norm. Then the estimation error satisfies
  \begin{align*}
      &\widehat{\theta}-\theta=(\mP_I-E)\lb \varphi(\bmO;\theta)\rb +\text{Rem}_{\theta},
  \end{align*}  
where $\text{Rem}_{\theta}$ satisfies: (i) provided that $\eta_{\buj}^1=\eta_{\buj}(1,\bM,\bC,N)$ and $\overline{\eta}=N^{-1}\sum_{j=1}^N \eta^1_{\buj}$,
\begin{align*}
   \text{Rem}_{\theta^{I}_1(a^\ast,\ba_{\mJ})}=&O_\mP\left(\lVert \widehat{\overline{\eta}}-\overline{\eta}\rVert_2\{\lVert \widehat{f}^{a^\ast}-f^{a^\ast}\rVert_2+\lVert \widehat{f}^{\ba_{\mJ1}}-f^{\ba_{\mJ1}}\rVert_2\}\right); 
\end{align*}
(ii) for $\theta=\theta^{II}_1(1,\ba_{\mJ^\ast})$, given $d_2(\widehat{g}_j,g_j)=\left\lVert \max_{j}|\widehat{g}_j-g_j|\right\rVert_2$ for any $g$ indexed by $j$,
\begin{align*}
\text{Rem}_{\theta}=O_{\mP}&\left(d_2(\widehat{\eta}_{\buj}^1,{\eta}^1_{\buj})\left\{\lVert \widehat{f}^1-f^1\rVert_2+\lVert\widehat{f}_{\mK\backslash\mJ^\ast}^1-f_{\mK\backslash\mJ^\ast}^1\rVert_2+\lVert\widehat{f}_{\mJ^\ast}^0-f_{\mJ^\ast}^0\rVert_2\right\}+\right.\\
   &\left.~~\lVert\widehat{f}_{\mK\backslash\mJ^\ast}^1-f_{\mK\backslash\mJ^\ast}^1\rVert_2 \lVert\widehat{f}_{\mJ^\ast}^0-f_{\mJ^\ast}^0\rVert_2\right);
\end{align*}
and (iii) for $\theta=\theta_2(k,\mJ^\ast)$,
\begin{align*}
\text{Rem}_{\theta}=&O_{\mP}\left(d_2(\widehat{\eta}^1_{\buj},{\eta}^1_{\buj})\left\{\lVert \widehat{f}^1-f^1\rVert_2+d_2(\widehat{f}_{\buj}^1,f_{\buj}^1)+d_2(\widehat{f}_{\bumj}^0,f_{\bumj}^0)\right\}+
   d_2(\widehat{f}_{\buj}^1,f_{\buj}^1) d_2(\widehat{f}_{\bumj}^0,f_{\bumj}^0)\right).
\end{align*}

\end{theorem}

The regularity conditions in Theorem \ref{thm:dml} are mild and standard: RC (i)–(iii) impose uniform boundedness and square‐integrability of the relevant functions, and RC (iv) ensures the nuisance estimators are consistent in the $\mathcal{L}_2(\mathbb{P})$ norm. Provided that $\text{Rem}_{\theta} = o_\mP(I^{-1/2})$, the proposed one-step estimator $\widehat{\theta}$ is $\sqrt{I}$-consistent, asymptotically normal, and semiparametrically efficient such that the asymptotic variance achieves the variance lower bound $E\lb\varphi(\bmO; \theta)^2\rb$. That is, it is sufficient to require the convergence rates of the nuisance function estimators of $o_\mP(I^{-1/4})$ in the $\mathcal{L}_2(\mP)$ norm to ensure the above asymptotic optimality of $\widehat{\theta}$, which is a standard rate of convergence condition in the semiparametric literature \citep{dml}. 
Finally, for inference based on one-step estimators with nuisance functions estimated from parametric working models, a robust approach is to use a nonparametric cluster bootstrap. On the other hand, for inference based on the debiased machine learning estimator, we use the cross-fitted empirical second moment of the EIF as a consistent variance estimator. For example, the variance estimator of NIE on the difference scale is: $\widehat{\mathbb{V}}(\theta_1(1,\ba_{\mK})-\theta_1(0,\ba_{\mK}))\equiv I^{-1} \sum_{s=1}^S I_s\mP_s\{[\widehat{\varphi}^s(\bmO;\theta_1(1,\ba_{\mK}))-\widehat{\varphi}^s(\bmO;\theta_1(0,\ba_{\mK}))]^2\}$,
where $\widehat{\varphi}^s(\bmO;\theta)=\varphi(\bmO;\widehat{\phi}^s(\bmO;\theta),\widehat{\boldsymbol{\Gamma}}^s(\theta))$. 

Finally, we comment on the robustness of $\widehat{\theta}^{\text{EIF.PAR}}$, as implied by the remainder terms in Theorem \ref{thm:dml}. To remain general, we do not limit the discussion to the ECMR model. In brief, the robustness depends on the mediation functional. First, for the functionals used to identify the $\text{TE}$, $\text{NDE}$, and $\text{NIE}$, the estimator is doubly robust. That is, it is consistent and asymptotically normal (CAN) if either the outcome mean model ($\{\eta_{\buj}:j\in[N]\}$) or the mediator density model ($\{f^a:a\in\{0,1\}\}$) is correctly specified, but not necessarily both. This is consistent with property of the semiparametric efficient estimator developed for independent data \citep{tchetgen2012semiparametric}. Second, for the functionals used to identify the interaction effects $\text{INT}_{\mJ}$ (including exit indirect effects $\text{EIE}_k$), $\widehat{\theta}^{\text{EIF.PAR}}$ is triply robust, meaning it is CAN if one of three combinations—the $\{f_{\mJ^\ast}^a,\eta_{\buj}:a\in\{0,1\},j\in[N]\}$ models, the $\{f_{\mK\backslash\mJ^\ast}^a,\eta_{\buj}:a\in\{0,1\},j\in[N]\}$ models, or the $\{f^a:a\in\{0,1\}\}$ model—is correctly specified; this provides three chances to accurately estimate the EIE and the entire set of interaction effects that offer a decomposition of the NIE.  
Finally, for $\theta_2$, the functionals used to identify $\text{SIME}_{\mJ}(k)$ and $\text{IIME}_{\mJ}(k)$, $\widehat{\theta}^{\text{EIF.PAR}}$ is conditionally doubly robust, meaning that it is CAN if either $\{f_{\buj}^a,\eta_{\buj}:a\in\{0,1\},j\in[N]\}$ or $\{f^a:a\in\{0,1\}\}$ is correctly specified. This conditional double robustness is conceptually similar to the semiparametric efficient estimator in \cite{cheng2024semiparametric} based on a single mediator in CRTs. 
Finally, all the proposed estimators rely on inverse-mediator-density weights, which may become unstable if the boundedness condition of the mediator density is practically violated. To mitigate this issue, we outline a stabilization approach based on weighted least squares in the following Remark.


\begin{remark}
To implement the stabilization while preserving the robustness property, one fits a weighted ordinary least squares model to obtain a refined outcome regression estimate, $\widehat{\eta}^s_{\buj}$. This model is fitted with intercept only and conditional on $A=1$, using $\widehat{r}^{k\mJ^\ast011}_{\buj}w(\bV,N)/N$ as the weights and $\widehat{\eta}_{\buj}$ as the offset term. Following \cite{tchetgen2012semiparametric}, $\widehat{\theta}$ no longer explicitly depends on the inverse-mediator-density weights, because the following terms are empirically zero: $\sum_{i=1}^I\sum_{j=1}^{N_i}{w(\bV_i,N_i)}/{N_i}\mathcal{I}(A_i=1)\widehat{r}^{k\mJ^\ast011}_{ij}(\bM_i,\bC_i,N_i)\{Y_{ij}-\widehat{\eta}_{ij}^{\text{s}}(1,\bM_i,\bC_i,N_i)\}=0$. {Steps for implementing the stabilized estimators are provided in Section \ref{sec:supp:stablization-implementation} of the Supplementary Material.}
\end{remark}

\section{Simulation experiments}\label{Sec:simulation}

\subsection{Simulation design}\label{ss:simulation.design}
We carried out simulation studies to illustrate the properties of different mediation effect estimators in a context mimicking the PPACT CRT. We consider $I = 100$ clusters with $K=2$ mediators to facilitate clear presentation of results, without loss of generality. We generate cluster size $N_i$ based on a discrete uniform distribution ranging from $N_{\min}=10$ to $N_{\max}=30$. We then simulate a continuous cluster-level covariate as $V_i\sim\mathcal{N}(2N_i/N_{\max},1)$, and a continuous individual-level covariate as $\bX_i\sim\mathcal{N}(2V_i\mathbf{1}_{N_i},0.95\boldsymbol{I}_{N_i}+0.05\mathbf{1}_{N_i}\mathbf{1}_{N_i}^\top)$. We then generate the treatment $A_i \sim \text{Bernoulli}(0.5)$. We employ the ECMR model to simulate the two mediators. We first generate the residuals as $\boldsymbol{\underline{\epsilon}}_i\sim\mathcal{E}_{N_i\times2}((\boldsymbol{\mathcal{Q}}_0-\boldsymbol{\mathcal{Q}}_1)\otimes\boldsymbol{I}_{N_i}+ \boldsymbol{\mathcal{Q}_1}\otimes \boldsymbol{1}_{N_i} \boldsymbol{1}_{N_i}^\top,g_{\mathcal{N}})$, where $\boldsymbol{\mathcal{Q}}_0=0.9\boldsymbol{I}_2+0.1\mathbf{1}_2\mathbf{1}_2^\top$ and $\boldsymbol{\mathcal{Q}}_1=0.05\mathbf{1}_2\mathbf{1}_2^\top$ are the within-individual and between-individual mediator correlation matrices, respectively. Then, we simulate (i) one continuous mediator whose marginal distribution follows a location–scale normal family: $M_{ij}^1\sim\mathcal{N}\left(-1+0.2A_i+\frac{(1+A_i)N_i}{2N_{\max}}+0.5V_i+0.5X_{ij},\frac{25}{4}\right)$,
and (ii) one binary mediator with a logistic marginal regression model: $M_{ij}^2\sim\text{Bernoulli}\left(\text{expit}\left(0.1A_i+\frac{(1+A_i)N_i}{5N_{\max}}+0.3V_i-0.3X_{ij}\right)\right)$,
where $\text{expit}(x)=\{1+\exp(-x)\}^{-1}$. The observed outcome $\mathbf{Y}_i$ is generated from a multivariate normal distribution with mean $0.2A_i + \frac{(1+A_i)N_i}{2N_{\max}} + 0.5V_i + 0.5X_{ij} + 0.6M^1_{ij} + 0.9M^2_{ij} - 0.8M^1_{ij}M^2_{ij} - 0.7\frac{M^\ast_{ij}}{N_i-1}$ and covariance matrix $0.9\boldsymbol{I}_{N_i} + 0.1\mathbf{1}_{N_i}\mathbf{1}_{N_i}^\top$,
where $M^\ast_{ij}=\sum_{j'\neq j} (M^1_{ij'}+2M^2_{ij'}+0.8M^1_{ij'}M^2_{ij'})$.

We consider five estimators: (i) the g-computation estimator, $\widehat{\theta}^{\text{g}}$; (ii) the semiparametric one-step estimator with parametric nuisance estimation, $\widehat{\theta}^{\text{EIF.PAR}}$; (iii) the debiased machine learning estimator, $\widehat{\theta}^{\text{EIF.DML}}$; (iv) a stabilized variant of $\widehat{\theta}^{\text{EIF.PAR}}$, denoted by $\widehat{\theta}^{\text{EIF.PAR.S}}$; and (v) a stabilized variant of $\widehat{\theta}^{\text{EIF.DML}}$, denoted by $\widehat{\theta}^{\text{EIF.DML.S}}$. For the marginal component defined in Equation \eqref{eq:elliptical-copula-regression}, we assume a semiparametric additive model with unspecified mean function $\underline{\eta}^1_{\buj}$ and a semiparametric logistic regression model with unspecified covariate effect function $\underline{\eta}^2_{\buj}$.
For each estimator, we evaluate performance under five scenarios based on the specification of three key components: the design matrices for $\{\underline{\eta}^1_{\buj},\underline{\eta}^2_{\buj}\}$, the design matrix for $\eta_{\buj}$, and the generator function $g$ for the ECMR model. Scenario (a) assumes all components are correctly specified; (b) misspecifies $\{\underline{\eta}^1_{\buj},\underline{\eta}^2_{\buj}\}$; (c) misspecifies $\eta_{\buj}$; (d) misspecifies the generator $g$; and (e) misspecifies all three components simultaneously. We implement the misspecification as follows. First, for the design matrix, we assume the transformed covariates are used for inference \citep{Kang2007STS}: $\widetilde{V}_i=\Phi(V_i)$ and $\widetilde{X}_{ij}=\widetilde{V}_i\text{expit}(-X_{ij}/2)$. Second, we use the misspecified generator $g_t(u;2)$, corresponding to a $t$-distribution with 2 degrees of freedom, to model the heaviest possible tail behavior in the $t$-family. The true values of the mediation estimands are computed based on a superpopulation of $5\times 10^5$ clusters. For inference, the variances of the parametric estimators are estimated using $100$ cluster bootstrap resamples, whereas the debiased machine learning estimators use the proposed closed-form variance estimators with cross-fitting. Performance metrics, including bias, average empirical standard errors (AESE), and empirical coverage probabilities, are computed over $1000$ iterations.

\subsection{Simulation results}\label{ss:simulation.results}
Simulation results for the NDE, $\text{INT}_{\{1\}}$, $\text{INT}_{\{2\}}$, and $\text{INT}_{\{1,2\}}$ are summarized in Figure \ref{fig:box-NDE-INT} (box plots of point estimates) and Figure \ref{fig:box-NDE-INT-CP} (95\% confidence interval coverage). Several observations are worth noting. First, all five estimators yield valid inference with negligible bias, nominal coverage, and comparable efficiency when all working models are correctly specified. Second, different working model misspecifications lead to different implications. The g-computation estimator can be biased and exhibit undercoverage if any working model is misspecified. In addition, the one-step estimator with parametric nuisance working models provides some protection against working model misspecifications, which aligns with the theoretical discussion in Section \ref{Sec:semi-estimation}. The debiased machine learning estimator further provides approximately valid inference, with only mild bias and undercoverage, even when all models are misspecified with respect to the design matrix input and all other estimators are biased. This demonstrates its potential to capture complex mediator and outcome surfaces in challenging settings. In particular, misspecification of the copula generator function appears to have negligible impact on the one-step estimators and the debiased machine learning estimators. In contrast, it can induce notable bias in the g-computation estimator, further illustrating its vulnerability in this complex setting. Interestingly, we find that misspecifying the mediator conditional mean functions appears to have a smaller impact than misspecifying the outcome mean functions. {To see this, all estimators exhibit only mild bias and undercoverage in scenario (b), although this is not guaranteed by theory. In contrast, misspecification of the outcome mean in scenario (c) can lead to severely biased inference for $\widehat{\theta}^{\text{g}}$.} Third, stabilization can improve efficiency and mitigate unstable point estimates, except for the total effect, particularly for the debiased machine learning estimators. {For example, the median efficiency gain from stabilization for $\widehat{\theta}^{\text{EIF.DML.S}}$ relative to $\widehat{\theta}^{\text{EIF.DML}}$ is 66\%, with the largest gains for interaction effects.} 
Finally, simulation results for all other mediation estimands are provided in Supplementary Material Figures \ref{fig:box-SIME}–\ref{fig:box-IIME} (point estimates) and Supplementary Material Figures \ref{fig:box-SIME-CP}–\ref{fig:box-IIME-CP} (95\% confidence interval coverage), and the results are generally qualitatively similar. In conclusion, we recommend the stabilized debiased machine learning estimator, as it offers enhanced robustness to model misspecification while maintaining finite-sample stability and efficiency comparable to competing methods even when all nuisance models are correctly specified.

\begin{figure}[ht!]
    \centering
    \includegraphics[width=1\linewidth]{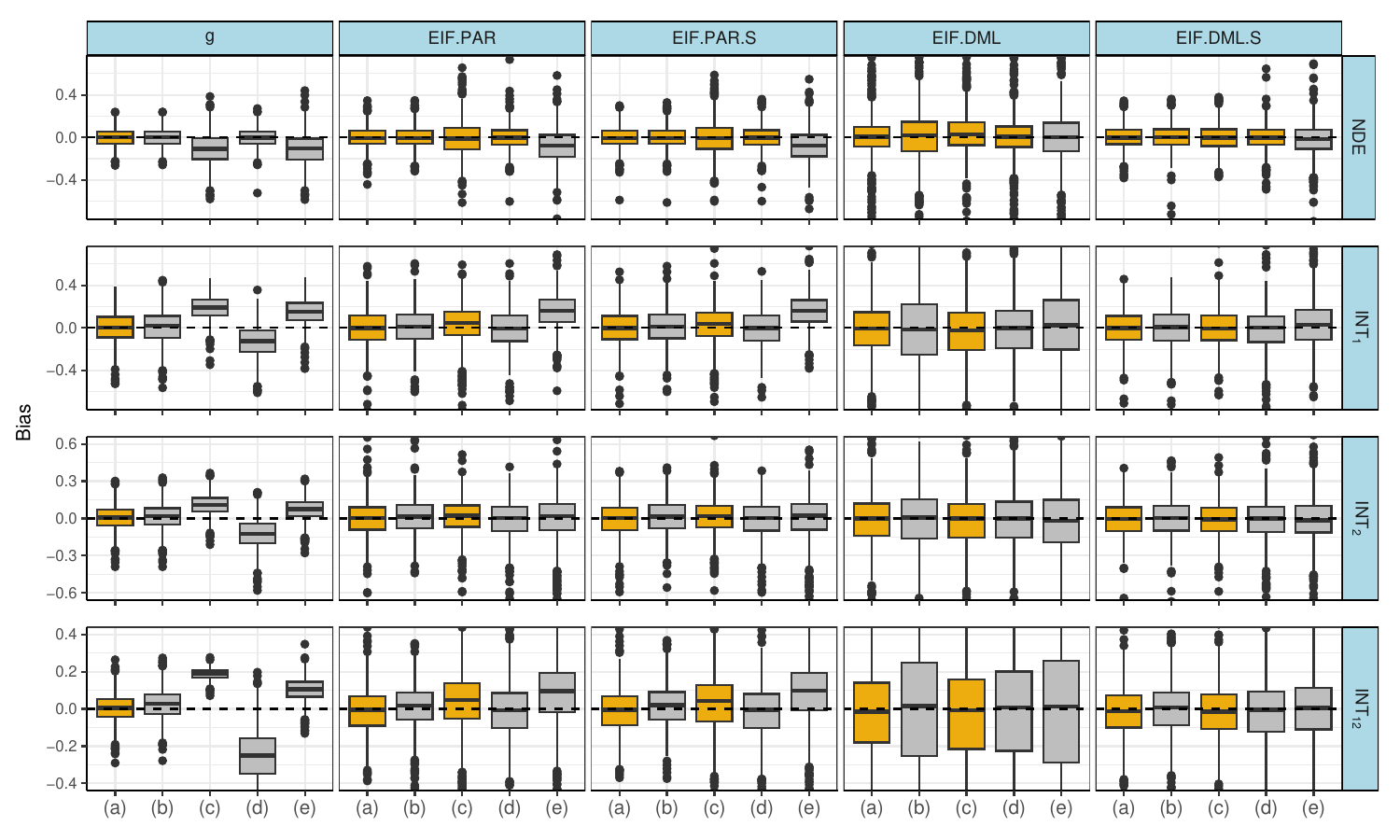}
    \caption{\small Simulation results for the bias of the natural indirect effect and interaction effects. 
    The nuisance model misspecifications follow: scenario (a) assumes all components are correctly specified; (b) misspecifies the mediator conditional means $\{\underline{\eta}_{\buj}^1,\underline{\eta}_{\buj}^2\}$; (c) misspecifies the outcome conditional mean $\eta_{\buj}$; (d) misspecifies the copula generator $g$; and (e) misspecifies all three simultaneously. The estimators that are expected to be theoretically valid are highlighted in orange. }
    \label{fig:box-NDE-INT}
\end{figure}

\begin{figure}[ht!]
    \centering
    \includegraphics[width=\linewidth]{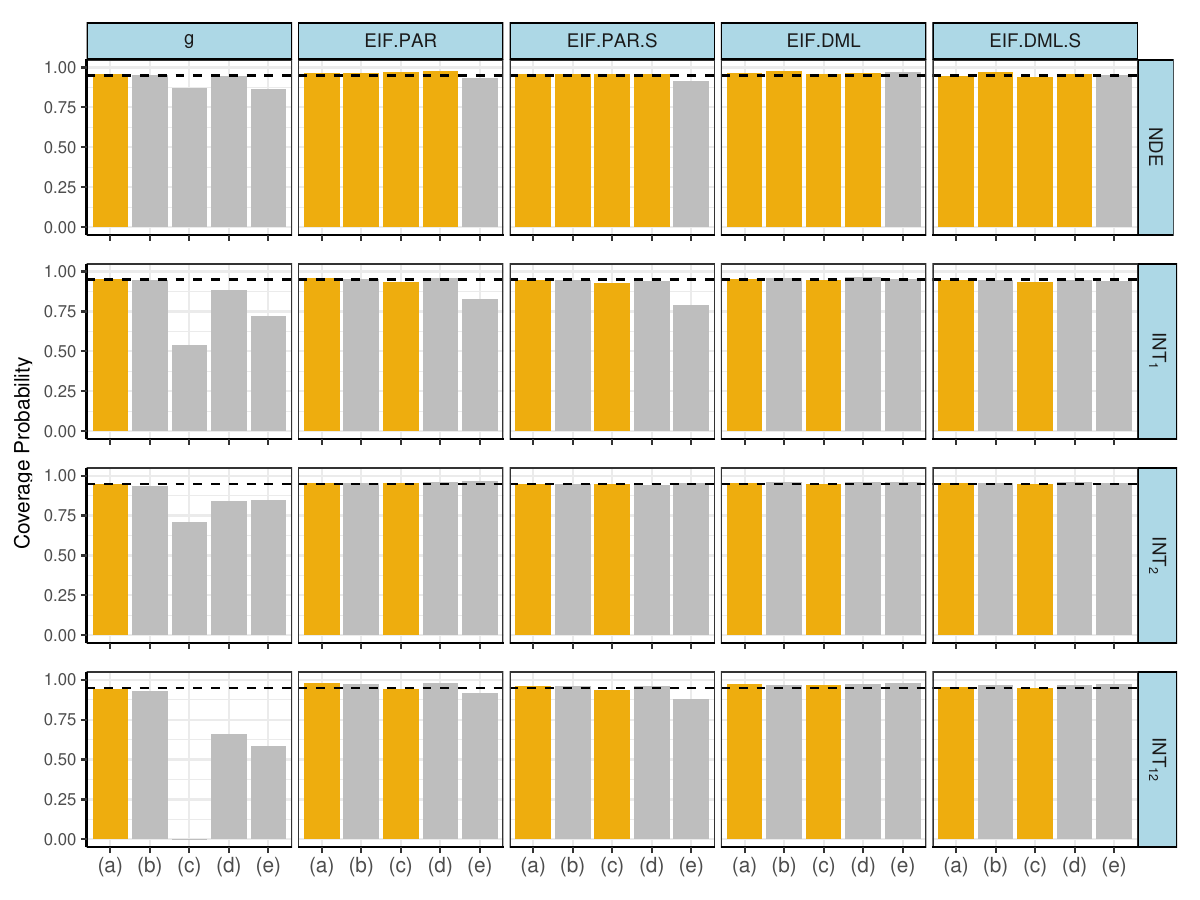}
    \caption{\small Simulation results for the empirical coverage probabilities of the natural indirect effect and interaction effects. 
    The nuisance model misspecifications follow: scenario (a) assumes all components are correctly specified; (b) misspecifies the mediator conditional means $\{\underline{\eta}_{\buj}^1,\underline{\eta}_{\buj}^2\}$; (c) misspecifies the outcome conditional mean $\eta_{\buj}$; (d) misspecifies the copula generator $g$; and (e) misspecifies all three simultaneously. The estimators that are expected to be theoretically valid are highlighted in orange. }
    \label{fig:box-NDE-INT-CP}
\end{figure}

\section{Data Application: Analysis of PPACT CRT}\label{sec:PPACT-application}
We reanalyze the PPACT CRT study, a cluster-randomized trial conducted across three Kaiser Permanente health care regions in Georgia, Hawaii, and the Northwest \citep{debar2022primary}. The objective of the PPACT CRT was to evaluate the effectiveness of an interdisciplinary, primary care-based CBT intervention for adults with chronic pain receiving long-term opioid therapy. The primary outcome was patient-reported pain impact at 12 months, measured by the PEGS scale (a composite score assessing pain intensity and interference with enjoyment of life, general activity, and sleep). The PPACT CRT comprises $850$ patients nested in $106$ primary care providers (i.e., the clusters), with the cluster size ranging from $3$ to $13$. We consider three mediators: the Roland-Morris Disability Questionnaire (RMDQ) score (continuous, $M_{ij}^1$), an indicator of improvement in satisfaction with primary care services from baseline (binary, $M_{ij}^2$), and the average daily opioid dose (continuous, $M_{ij}^3$). Since all three mediators are measured at 6 months, it is natural to treat them as unstructured, and our methods therefore apply naturally.

In the analysis, we (i) adjust for cluster size and $8$ additional individual-level covariates: sex, age, diagnosis of $\geq2$ chronic conditions, count of pain types, benzodiazepine receipt in 6 months prior to randomization, baseline PEGS score, baseline RMDQ score, and average daily opioid dose per day in 6 months prior to randomization; (ii) assume 
that the conditional outcome mean depends on the cluster means of other individuals' mediators. For illustration, we (i) focus primarily on the TE, NDE, NIE, $\text{INT}_{\mJ}$, and the ESME and EIME,
following the hierarchy and heredity principles; (ii) focus primarily on the one-step estimators with parametric nuisance estimation, specifically $\widehat{\theta}^{\text{EIF.PAR}}$ and $\widehat{\theta}^{\text{EIF.PAR.S}}$ (using linear regression for the outcome and continuous mediators, and logistic regression for the binary mediator, with standard errors computed from $250$ bootstrap resamples), as well as the debiased machine learning estimators $\widehat{\theta}^{\text{EIF.DML}}$ and $\widehat{\theta}^{\text{EIF.DML.S}}$ (with nuisance functions estimated via Super Learner \citep{luedtke2016super}); and (iii) assume a Gaussian copula, as misspecification of the copula generator appears to have limited impact on the one-step estimators, as suggested by the simulations. Point estimates and 95\% Wald confidence intervals for all one-step estimators are summarized in Table \ref{tab:ppact-3-mediators}.

Several noteworthy findings emerge from the analysis. First, the CBT intervention yields a statistically significant TE in reducing chronic pain, concordant with the primary analysis in \citet{debar2022primary}. The NIE accounts for approximately $41.3\%$ of the TE when estimated using $\widehat{\theta}^{\text{EIF.PAR}}$ or $\widehat{\theta}^{\text{EIF.PAR.S}}$, and $36.7\%$ and $24.5\%$ when estimated using $\widehat{\theta}^{\text{EIF.DML.S}}$ and $\widehat{\theta}^{\text{EIF.DML}}$, respectively. Second, the estimated EIEs for all mediators are consistently negative, with the indirect effect mediated by the RMDQ score notably larger in magnitude, suggesting that this pathway plays a more prominent role than those through satisfaction with primary care services or average daily opioid dose. Third, the interaction indirect effects among any two or three mediators are generally estimated to be small, except for the empirically unstable $\widehat{\theta}^{\text{EIF.DML}}$. This suggests that, although the mediators are assumed to be unstructured a priori, there is little complementary or offsetting interaction among them in mediating the effect of CBT on pain impact. It also suggests that the natural indirect effect can be approximated by the sum of the three independent exit indirect effects, signaling an approximately independent effect mechanism. Fourth, the estimated spillover indirect effects transmitted through each mediator are generally comparable to or slightly larger in magnitude than their individual counterparts. This suggests that neighborhood spillover among care providers may represent an important pathway in the mediation process. Finally, we comment on the empirical statistical properties. In general, stabilization yields notable efficiency gains and produces more stable, reasonable estimates, especially for the debiased machine learning estimators. Furthermore, the estimates from $\widehat{\theta}^{\text{EIF.PAR}}$ and $\widehat{\theta}^{\text{EIF.PAR.S}}$ are close to $\widehat{\theta}^{\text{EIF.DML.S}}$, which may suggest that the parametric working models provide a reasonable approximation to the truth in this data example. We also noticed that for certain estimands, the confidence intervals derived from $\widehat{\theta}^{\text{EIF.PAR}}$ or $\widehat{\theta}^{\text{EIF.PAR.S}}$ are narrower than those obtained from $\widehat{\theta}^{\text{EIF.DML.S}}$ or $\widehat{\theta}^{\text{EIF.DML}}$. This discrepancy may reflect the trade-off between model flexibility and finite-sample stability.

\begin{table}[ht!]
\centering
\caption{Point estimates and 95\% Wald confidence intervals are reported for all one-step estimators in the PPACT CRT data under a Gaussian copula. Mediator labels 1, 2, and 3 correspond to the RMDQ score, improvement in satisfaction with primary care services from baseline, and average daily opioid dose, respectively.
}\label{tab:ppact-3-mediators}
\resizebox{\textwidth}{!}{
\begin{tabular}{lrcrcrcrc}
\toprule
 & \multicolumn{2}{c}{$\widehat{\theta}^{\text{EIF.PAR}}$} & \multicolumn{2}{c}{$\widehat{\theta}^{\text{EIF.PAR.S}}$} & \multicolumn{2}{c}{$\widehat{\theta}^{\text{EIF.DML}}$} & \multicolumn{2}{c}{$\widehat{\theta}^{\text{EIF.DML.S}}$} \\
\cmidrule(lr){2-3} \cmidrule(lr){4-5} \cmidrule(lr){6-7} \cmidrule(lr){8-9}
Estimand & Est. & 95\% CI & Est. & 95\% CI & Est. & 95\% CI & Est. & 95\% CI \\
\midrule
TE & {-0.46} & {(-0.81, -0.12)} & {-0.46} & {(-0.81, -0.12)} & {-0.49} & {(-0.84, -0.14)} & {-0.49} & {(-0.83, -0.14)} \\
NDE & -0.28 & (-0.71, 0.15) & -0.27 & (-0.68, 0.13) & -0.37 & (-1.10, 0.35) & -0.31 & (-0.74, 0.13) \\
NIE & -0.19 & (-0.52, 0.14) & -0.19 & (-0.49, 0.11) & -0.12 & (-0.78, 0.55) & -0.18 & (-0.48, 0.12) \\[1.0ex]
EIE$_1$ & -0.16 & (-0.55, 0.22) & -0.17 & (-0.46, 0.13) & -0.09 & (-0.78, 0.60) & -0.17 & (-0.50, 0.15) \\
EIE$_2$ & -0.02 & (-0.10, 0.06) & -0.02 & (-0.10, 0.06) & 0.00 & (-1.04, 1.05) & -0.03 & (-0.48, 0.43) \\
EIE$_3$ & -0.02 & (-0.10, 0.06) & -0.02 & (-0.10, 0.06) & -0.08 & (-0.80, 0.63) & -0.03 & (-0.20, 0.13) \\[1.0ex]
INT$_{\{1,2\}}$ & -0.01 & (-0.20, 0.18) & -0.01 & (-0.12, 0.11) & 0.13 & (-0.89, 1.16) & -0.02 & (-0.44, 0.41) \\
INT$_{\{1,3\}}$ & -0.00 & (-0.10, 0.10) & -0.00 & (-0.07, 0.07) & 0.05 & (-0.43, 0.54) & -0.01 & (-0.12, 0.11) \\
INT$_{\{2,3\}}$ & -0.02 & (-0.09, 0.06) & -0.02 & (-0.09, 0.06) & -0.08 & (-1.68, 1.53) & -0.03 & (-0.57, 0.50) \\
INT$_{\{1,2,3\}}$ & -0.01 & (-0.09, 0.07) & -0.01 & (-0.07, 0.05) & 0.16 & (-1.21, 1.54) & -0.01 & (-0.46, 0.45) \\[1.0ex]
ESME$_1$ & -0.09 & (-0.45, 0.28) & -0.09 & (-0.35, 0.18) & -0.04 & (-0.61, 0.54) & -0.08 & (-0.41, 0.24) \\
ESME$_2$ & -0.01 & (-0.08, 0.06) & -0.01 & (-0.08, 0.06) & 0.03 & (-1.79, 1.85) & -0.03 & (-0.79, 0.73) \\
ESME$_3$ & -0.01 & (-0.08, 0.05) & -0.01 & (-0.08, 0.05) & -0.00 & (-0.18, 0.17) & -0.02 & (-0.14, 0.11) \\
EIME$_1$ & -0.08 & (-0.23, 0.07) & -0.08 & (-0.21, 0.05) & -0.05 & (-0.29, 0.19) & -0.09 & (-0.24, 0.06) \\
EIME$_2$ & -0.01 & (-0.05, 0.03) & -0.01 & (-0.04, 0.03) & -0.03 & (-1.05, 1.00) & 0.01 & (-0.59, 0.61) \\
EIME$_3$ & -0.01 & (-0.04, 0.03) & -0.01 & (-0.04, 0.03) & -0.08 & (-0.82, 0.66) & -0.01 & (-0.10, 0.07) \\
\bottomrule
\end{tabular}
}
\end{table}

\section{Conclusion}\label{ss:concluson}
In this paper, we develop a unified framework for causal mediation analysis in cluster-randomized trials with multiple causally unordered mediators. Our contributions are threefold. First, we addressed the complexity of multiple mediators in cluster-randomized trials in full generality by characterizing the complete combinatorial algebra underlying effect decomposition for an arbitrary number of mediators. Second, we introduce a set of structural assumptions under which all causal mediation effects are nonparametrically identified. Third, we establish semiparametric efficiency theory by deriving the efficient influence function, constructing one-step estimators, and studying their asymptotic properties. Notably, when combined with parametric working models, the proposed one-step estimators enjoy varying degrees of robustness depending on the target estimand: the estimators for the natural direct effect, and natural indirect effect are doubly robust; the estimators for the interaction effects $\text{INT}_{\mathcal{J}}$ (including the exit indirect effects as a special case) are triply robust; and the estimators for the spillover and individual interaction mediation effects are conditionally doubly robust. When combined with data-adaptive machine learning algorithms via cross-fitting, the resulting debiased estimators achieve $\sqrt{I}$-consistency and semiparametric efficiency under mild rate conditions on the nuisance estimators. In particular, we propose a flexible elliptical copula marginal regression model for the joint mediator density, a key high-dimensional nuisance function, with interpretable parameterizations.

This work focuses on causally unordered mediators, which is natural when mediators are measured contemporaneously, as in the motivating PPACT CRT. When mediators are observed sequentially, incorporating causal ordering may yield more refined inference. Besides, the elliptical copula model is also of independent interest for joint density modeling beyond the mediation context. In this work, the generator function is treated as fixed to mitigate the curse of dimensionality, but may be made data adaptive in moderate dimensions. A promising direction for future research is to develop sieve-based procedures for estimating the generator function; see Supplementary Material Section \ref{ss:SI-ECMR-SUPP-data-adaptive-g} for additional discussions. Finally, identification relies on cross-world structural assumptions, and it is of interest to develop sensitivity methods in future work to assess robustness to potential assumption violations.



\section*{Acknowledgments}
Research in this article was supported by the Patient-Centered Outcomes Research Institute\textsuperscript{\textregistered} (PCORI\textsuperscript{\textregistered} Award ME-2023C1-31350). All statements expressed in this article are solely those of the authors and do not necessarily reflect the views of PCORI\textsuperscript{\textregistered}, its Board of Governors or Methodology Committee.

\section*{Data Availability Statement}
The PPACT CRT data that support the findings of this study are openly available in GitHub at \url{https://github.com/PainResearch/PPACT}. 

\bigskip
\begin{center}
{\large\bf SUPPLEMENTARY MATERIAL}
\end{center}
\setcounter{section}{0}
\renewcommand{\thesection}{S\arabic{section}}

\section{Summary}
\label{sec:intro_supple}
This supplementary material is organized as follows. Section \ref{sec:alternative-GCH-assumption} provides an alternative to Assumption \ref{K-assump:homogeneous-between-mediator-ignorability-INT} in the main manuscript. Section \ref{sec:derivation-estimands} presents the derivations of the proposed estimands and their decompositions under the risk difference, risk ratio, and odds ratio scales. Section \ref{sec:SI-ECMR-SUPP} presents additional technical results regarding the proposed elliptical copula marginal regression (ECMR) model. Section \ref{sec:supp:stablization-implementation} presents a step-by-step implementation of the stabilization technique for the semiparametric one-step estimators. Section \ref{sec:proofs-all} presents the proofs of all technical results in the main manuscript. Of note, for full generality, we retain the weights $w(\bV, N)$ throughout the proof, whereas the main manuscript primarily focuses on the cluster-average mediation effect estimands by setting $w(\bV, N) = 1$. If our interest is in the individual-average mediation effect estimand, we may simply set $w(\bV, N)=N$.

\section{An alternative to Assumption \ref{K-assump:homogeneous-between-mediator-ignorability-INT}}\label{sec:alternative-GCH-assumption}
\begin{assumption}
\label{K-assump:homogeneous-between-mediator-ignorability-INT-alternative} 
For all $\mJ^\ast\subseteq \mJ$, $\mJ^\ast\neq\mK,\emptyset$, and $\bm^{\mJ^\ast}$ and $\bm^{\mK \backslash \mJ^\ast}$ as the realizations of $\bM^{\mJ^\ast}$ and $\bM^{\mK\backslash \mJ^\ast}$, respectively,
\begin{align*}
E\left\{\overline{Y}(1,\bM(\ba_{\mJ^\ast}))|\bM^{\mK\backslash \mJ^\ast}(1)=\bm^{\mK\backslash \mJ^\ast},\bC,N\right\}=E\left\{\overline{Y}(1,\bM^{\mJ^\ast}(0),\bm^{\mK\backslash \mJ^\ast})|\bC,N\right\}.
\end{align*}
\end{assumption}
Assumption \ref{K-assump:homogeneous-between-mediator-ignorability-INT-alternative} serves as an alternative to Assumption \ref{K-assump:homogeneous-between-mediator-ignorability-INT} presented in the main manuscript and is similarly indexed by $\mJ$. Specifically, Assumption \ref{K-assump:homogeneous-between-mediator-ignorability-INT-alternative}-$\mJ$ posits that, conditional on $\{\bC,N\}$, the expectation of the counterfactual cluster-average outcome $\overline{Y}(1,\bM(\ba_{\mJ^\ast}))$ remains invariant across strata defined by the counterfactual mediators under treatment, $\bM^{\mK\backslash\mJ^\ast}(1)$. In other words, both Assumption \ref{K-assump:homogeneous-between-mediator-ignorability-INT-alternative} and Assumption \ref{K-assump:homogeneous-between-mediator-ignorability-INT} impose mean ignorability for $\overline{Y}(1,\bM(\ba_{\mJ^\ast}))$, conditional on a subset of potential mediators under the same treatment status. Substituting Assumption \ref{K-assump:homogeneous-between-mediator-ignorability-INT-alternative} for Assumption \ref{K-assump:homogeneous-between-mediator-ignorability-INT} yields identical identification formulas. 

\section{Derivations of estimands and their decompositions}\label{sec:derivation-estimands}
In constructing the estimands, we follow \cite{tong2025permutationinvarianceprinciplecausal}, specifically applying their permutation-equivariant, complete, and residual-free estimands derived via combinatorial algebra. Briefly, permutation equivariance ensures that the interpretation of the interaction effects remains invariant under mediator label permutations. To proceed, we set $\mathcal{Y} = \mJ$, $f(\mathcal{Z}) = \theta_1(1, \ba_{\mathcal{Z}})$, and $\mJ^\ast = \mathcal{Z}$. Then, the expressions for the interaction effects and the decomposition of the NIE respectively follow from the formulas 
\begin{align*}
    &\lb\Delta_{\mathcal{Y}}=\sum_{\mathcal{Z}\subseteq \mathcal{Y}}(-1)^{|\mathcal{Z}|}f(\mathcal{Z}):\emptyset\neq \mathcal{Y}\in 2^{\mathcal{X}}\rb,\\
    &f(\emptyset)-f(\mathcal X)=\sum_{\emptyset\neq \mathcal Y\subseteq \mathcal X}(-1)^{|\mathcal Y|+1}\Delta_{\mathcal Y},
\end{align*}
as given in \cite{tong2025permutationinvarianceprinciplecausal}. The alternative characterization in Equation \eqref{eq:decomposition-INT} of the main manuscript follows directly from Proposition 1 in \cite{tong2025permutationinvarianceprinciplecausal}. The results under the risk ratio and odds ratio (Remark \ref{eg:relative-risk}) follow from Corollaries 1 and 2 in \cite{tong2025permutationinvarianceprinciplecausal}, respectively. Finally, to generalize Equations \eqref{eq:SIMEJ-IIMEJ} of the main manuscript from the risk difference to ratio scales, we replace the building block $\theta \in \{\theta_1, \theta_2\}$ with its logarithmic transformation $\log \circ h (\theta)$, where $h(x) = x$ for the risk ratio and $h(x) = x / (1 - x)$ for the odds ratio. Consequently, we obtain the decomposition of the interaction effect on the log scale as follows
\begin{align*}
\log (\text{INT}_\mJ)=&\underbrace{\sum_{\mJ^\ast\subseteq \mJ\backslash\lb k\rb}(-1)^{|\mJ^\ast|}\left\{\log\circ h(\theta_1(1,\ba_{\mJ^\ast}))-\log\circ h(\theta_2(k,\mJ^\ast))\right\}}_{\log(\text{SIME}_\mJ(k))}+\\
&\underbrace{\sum_{\mJ^\ast\subseteq \mJ\backslash\lb k\rb}(-1)^{|\mJ^\ast|}\left\{\log\circ h(\theta_2(k,\mJ^\ast))-\log\circ h(\theta_1(1,\ba_{\mJ^\ast\cup\lb k\rb}))\right\}}_{\log(\text{IIME}_\mJ(k))},    
\end{align*}
which further implies that
\begin{align*}
\log (\text{INT}_\mJ)=&\underbrace{\log\left\{\prod_{\mJ^\ast\subseteq \mJ\backslash\lb k\rb}\left(\frac{h(\theta_1(1,\ba_{\mJ^\ast}))}{h(\theta_2(k,\mJ^\ast))}\right)^{(-1)^{|\mJ^\ast|}}\right\}}_{\log(\text{SIME}_\mJ(k))}+\underbrace{\log\left\{\prod_{\mJ^\ast\subseteq \mJ\backslash\lb k\rb}\left(\frac{h(\theta_2(k,\mJ^\ast))}{h(\theta_1(1,\ba_{\mJ^\ast\cup\lb k\rb}))}\right)^{(-1)^{|\mJ^\ast|}}\right\}}_{\log(\text{IIME}_\mJ(k))}.  
\end{align*}
Thus, it follows that
\begin{align*}
\text{INT}_\mJ=&\underbrace{\left\{\prod_{\mJ^\ast\subseteq \mJ\backslash\lb k\rb}\left(\frac{h(\theta_1(1,\ba_{\mJ^\ast}))}{h(\theta_2(k,\mJ^\ast))}\right)^{(-1)^{|\mJ^\ast|}}\right\}}_{\text{SIME}_\mJ(k)}\times\underbrace{\left\{\prod_{\mJ^\ast\subseteq \mJ\backslash\lb k\rb}\left(\frac{h(\theta_2(k,\mJ^\ast))}{h(\theta_1(1,\ba_{\mJ^\ast\cup\lb k\rb}))}\right)^{(-1)^{|\mJ^\ast|}}\right\}}_{\text{IIME}_\mJ(k)}.  
\end{align*}
Finally, Remark \ref{eg:relative-risk} is obtained by setting $h(x) = x$ and $h(x) = x/(1-x)$, respectively.




\section{Supporting information for the ECMR model}\label{sec:SI-ECMR-SUPP}
\subsection{Extension with a data-adaptive generator}\label{ss:SI-ECMR-SUPP-data-adaptive-g}
To maximize model flexibility, we assume that the generator function belongs to an infinite-dimensional class $\mathcal{G}$. For example, typical smooth function classes include the Hölder class of smoothness order $s$. In this case, we follow \cite{derumigny2022identifiability} to approximate $\mathcal{G}$ using a sieve space $\mathcal{G}_U$ based on Bernstein polynomials of order $U$. Formally, the sieve space $\mathcal{G}_U$ can be defined as, for any given $a>0$, 
\begin{align*}
   \mathcal{G}_U = &\{g \in \mathcal{G}^\ast : g(x) = \sum_{q=0}^U \beta_q (a+x)^q (a-x)^{U-q} \mathbf{1}(x \in [0,a]), \beta_q \geq 0 \}. 
\end{align*}
Here, $\mathbf{1}(\cdot)$ denotes the indicator function, the complexity of the sieve space $U = U_I$ is assumed to grow slowly with the sample size $I$ to balance the bias-variance trade-off, and the space $\mathcal{G}^\ast$ is subject to specific identifiability constraints (see Equations (2) and (6) in \cite{derumigny2022identifiability}). Define $\boldsymbol{\beta}=(\beta_1,\ldots,\beta_q)^\top$. For estimation, we propose to using the sieve maximum likelihood estimator
\begin{align*}
    (\widehat{\boldsymbol{\beta}},\widehat{\boldsymbol{\mathcal{Q}}})=&\underset{{{\{\beta_q\}_{q=1}^U,{\boldsymbol{\mathcal{Q}}}}}}{\arg\max} \prod_{i}^I \int_{\widehat{\mathcal{D}}_i(\bM_i,\mK_d)} f_{\mathcal{E}_{N_iK}(\boldsymbol{R}(\boldsymbol{\mathcal{Q}}),g)}(\widehat{{\underline{{\boldsymbol{\epsilon}}}}}_{\mK_c,i},{\underline{{\boldsymbol{\epsilon}}}}_{\mK_d})d{\underline{\boldsymbol{\epsilon}}}_{\mK_d}\\
    &\text{subject to }\boldsymbol{\mathcal{Q}}_0\succ\boldsymbol{\mathcal{Q}}_1,\boldsymbol{\mathcal{Q}}_0\succ0, g\in\mathcal{G}^\ast, \text{ and }\bbeta\succeq0.
\end{align*}
Finally, the estimator of the generator function, $\widehat{g}$, is given by
\begin{align*}
    \widehat{g}(x)=\sum_{q=0}^U \widehat{\beta}_q (a+x)^q (a-x)^{U-q} \mathbf{1}(x \in [0,a]).
\end{align*}
Notably, the complexity of the sieve space, $U$, should grow slowly with the sample size $I$; in practice, an optimal $U$ can be selected using information criterion under cross-validation.

\subsection{Some useful properties}
We present closed-form expressions for the joint densities of any subset of mediators. By Proposition \ref{prop:property-of-ECMR} and Property (P.1) in \cite{owen1983class}, any subvector of $\boldsymbol{\widetilde{\underline{\epsilon}}}$ also follows a multivariate standard elliptical distribution with the corresponding submatrix of the correlation matrix. That is, $\boldsymbol{\widetilde{\underline{\epsilon}}}_{\mathcal{S}}\sim \mathcal{E}_{|S|}(\boldsymbol{R}_{\mathcal{S}},g)$, where $\mathcal{S}\subseteq\mathcal{L}_N=[N]\times\mK$. Thus, the estimators for the multivariate joint densities of any subset of mediators indexed by $\mathcal{S}$ admit closed-form expressions and are given by:
{\fontsize{9.5pt}{11.5pt}\selectfont
\begin{align*}
\widehat{f}(\bM_{\mathcal{S}}|A,\bC,N)=&\prod_{(j,k)\in\mK_c\cap\mathcal{S}}\frac{\widehat{f}(M_{\buj}^k|A,\bC,N)}{ f_{\mathcal{E}(g)}(\widehat{\underline{\epsilon}}_{\buj}^k)}\int_{\widehat{\mathcal{D}}(\bM,\mK_d\cap\mathcal{S})} f_{\mathcal{E}_{|\mathcal{S}|}(\boldsymbol{R}_{\mathcal{S}}(\widehat{\boldsymbol{\mathcal{Q}}}),g)}(\widehat{{\underline{\widetilde{\boldsymbol{\epsilon}}}}}_{\mK_c\cap\mathcal{S}},{\underline{\widetilde{\boldsymbol{\epsilon}}}}_{\mK_d\cap\mathcal{S}})d{\underline{\boldsymbol{\epsilon}}}_{\mK_d\cap\mathcal{S}},
\end{align*}}
where $\mK_r\cap\mathcal{S},r\in\{c,d\}$ is the shorthand for $(\mK_r\times[N])\cap\mathcal{S}$, and $\times$ means the Cartesian product of sets. 

\subsection{Additional technical results}
The following proposition establishes an important property of the error terms for continuous mediators.
\begin{prop}\label{prop:property-of-ECMR}
Assuming that all mediators are continuous, the error terms follow a multivariate standard elliptical distribution:
$$\boldsymbol{\underline{\epsilon}}\equiv(\boldsymbol{\underline{\epsilon}}^1,\ldots,\boldsymbol{\underline{\epsilon}}^K)\sim\mathcal{E}_{NK}(\mathbf{0},\boldsymbol{R},g),$$
where $\boldsymbol{\underline{\epsilon}}^k=(\underline{\epsilon}_{\bullet1}^k,\ldots,\underline{\epsilon}_{\bullet N}^k)$. 
\end{prop}

We then show why the model specification in Equation \eqref{eq:alternative-ellipitical copula} of the main manuscript implies the elliptical copula structure in Equation \eqref{eq:model-elliptical-copula} of the main manuscript. Specifically, the specification in Equation \eqref{eq:alternative-ellipitical copula} of the main manuscript captures the dependence structure of the elliptical copula while ensuring identifiability. To simplify notation and without loss of generality, we consider a ECMR model with $K=1$ (only for illustration) and $N=2$. Suppose that 
\begin{align*}
    &(\epsilon_{\bullet 1}^1,\epsilon_{\bullet 2}^1)^\top \sim \mathcal{E}_2(\mathbf{0},\boldsymbol{R},g),\\
    &M_{\bullet1}^1 =F^{-1}(F_{\mathcal{E}(g)}(\epsilon_{\bullet 1}^1)),\\
    &M_{\bullet2}^1 =F^{-1}(F_{\mathcal{E}(g)}(\epsilon_{\bullet 2}^1)).
\end{align*}
Then we have 
\begin{align}
 \Pr(M_{\bullet1}^1\leq a,M_{\bullet2}^1\leq b)=&\Pr(F^{-1}(F_{\mathcal{E}(g)}(\epsilon_{\bullet 1}^1))\leq a,F^{-1}(F_{\mathcal{E}(g)}(\epsilon_{\bullet 2}^1))\leq b)\nonumber\\
 =&\Pr(\epsilon_{\bullet 1}^1\leq F^{-1}_{\mathcal{E}(g)}(F(a)),\epsilon_{\bullet 2}^1\leq F^{-1}_{\mathcal{E}(g)}(F(b)))\nonumber\\
 =&\boldsymbol{F}_{\mathcal{E}_2(\mathbf{0},\boldsymbol{R},g)}(F^{-1}_{\mathcal{E}(g)}(F(a)),F^{-1}_{\mathcal{E}(g)}(F(b))),\label{eq:(6)imply(5)-proof}
\end{align}
where the second equality follows from the Galois inequalities for quantile functions. Finally, Equation \eqref{eq:(6)imply(5)-proof} implies that the copula for the joint mediator distribution is exactly the elliptical copula in Equation \eqref{eq:model-elliptical-copula} of the main manuscript.

\section{Stabilization technique}\label{sec:supp:stablization-implementation}
We outline the stabilization procedure in the Algorithm \ref{alg:stabilized_estimator}, using the estimation of $\theta = \theta_2$ as an illustration. A similar procedure applies to all other causal mediation functionals.

\begin{algorithm}
\caption{Stabilization of the semiparametric one-step estimator}
\label{alg:stabilized_estimator}
\small
\begin{algorithmic}[1]
\Require Observed data $\{\bmO_i\}_{i=1}^I$
\Ensure Stabilized semiparametric one-step estimator $\widehat{\theta}_2$

\Statex \textbf{Step 1: Initial Nuisance Parameter Estimation}
\State Compute the fitted individual-level outcome mean $\widehat{\eta}_{ij}(1, \bM_i, \bC_i, N_i)$ for all individuals
\State Compute the marginal-to-joint density ratio estimates $\widehat{r}_{ij}^{k\mJ^\ast011}(\bM_i, \bC_i, N_i)$ for all individuals based on the fitted joint mediator densities $f(\bM_i^{\mJ^\ast},\bM^k_{i,-j}|0,\bC_i,N_i)$, $f(M_{ij}^k,\bM_i^{\mK\backslash (\mJ^\ast\cup\{k\})}|1,\bC_i,N_i)$, and $f(\bM_i|1,\bC_i,N_i)$

\Statex \textbf{Step 2: Stabilization via Weighted OLS}
\State Subset the data to observations where the treatment condition holds, $A_i=1$
\State For these observations, compute the individual-level regression weights: 
    \[ \mathbb{W}_{ij} = \frac{\widehat{r}_{ij}^{k\mJ^\ast011}(\bM_i,\bC_i,N_i)w(\bV_i,N_i)}{N_i} \]
\State Fit an intercept-only weighted ordinary least squares (WOLS) model (without clustering) for the outcome $Y_{ij}$, conditional on $A_i = 1$, with weights $\mathbb{W}_{ij}$ and offset term $\widehat{\eta}_{ij}(1, \bM_i, \bC_i, N_i)$
\State Obtain the refined estimates of the individual-level outcome mean as $\widehat{\eta}^s_{ij}(1, \bM_i, \bC_i, N_i) = \widehat{\beta} + \widehat{\eta}_{ij}(1, \bM_i, \bC_i, N_i)$, where $\widehat{\beta}$ is the estimated intercept based on 
WOLS
\State Update $\widetilde{\kappa}_{ij}^{k\mJ^\ast}(\bM_{i,-j}^k,\bM_i^{\mJ^\ast},\bC_i,N_i)$, $\kappa^{k\mJ^\ast}_{ij}(\bC_i,N_i)$, and $\check{\kappa}_{ij}^{k\mJ^\ast}(M_{ij}^k,\bM_i^{\mK\backslash(\mJ^\ast\cup\{k\})},\bC_i,N_i)$ based on the refined estimates $\widehat{\eta}^s_{ij}(1, \bM_i, \bC_i, N_i)$ 
\Statex \textbf{Step 3: Final Stabilized One-Step Estimator}
\Statex Following \cite{robins2007comment,tchetgen2012semiparametric}, we have
\Statex \vspace{-0.5em}
    \[ \sum_{i=1}^I\sum_{j=1}^{N_i} \frac{w(\bV_i,N_i)}{N_i} \mathcal{I}(A_i=1) \widehat{r}^{k\mJ^\ast011}_{ij}(\bM_i,\bC_i,N_i) \left( Y_{ij}-\widehat{\eta}_{ij}^{\text{s}}(1,\bM_i,\bC_i,N_i) \right) = 0, \]
which simplifies the one-step estimator to 
\begin{align*}
    \widehat{\theta}_2=&\sum_{i=1}^I\frac{w(\bV_i,N_i)}{N_i}\sum_{j=1}^{N_i}\Bigg\{\frac{\mathcal{I}(A_i=0)}{\widehat{\Pr}(A_i=0)}\widetilde{\kappa}_{ij}^{k\mJ^\ast}(\bM_{i,-j}^k,\bM_i^{\mJ^\ast},\bC_i,N_i) +\\
    &\left(1-\sum_{a=0}^1\frac{\mathcal{I}(A_i=a)}{\widehat{\Pr}(A_i=a)}\right)\kappa^{k\mJ^\ast}_{ij}(\bC_i,N_i)+\frac{\mathcal{I}(A_i=1)}{\widehat{\Pr}(A_i=1)}\check{\kappa}_{ij}^{k\mJ^\ast}(M_{ij}^k,\bM_i^{\mK\backslash(\mJ^\ast\cup\{k\})},\bC_i,N_i)\Bigg\}
\end{align*}

\State \Return $\widehat{\theta}$
\end{algorithmic}
\end{algorithm}

\section{Proofs of technical results}\label{sec:proofs-all}
\subsection{Proof of Proposition \ref{prop:property-of-ECMR}}\label{supp;ss:proof of proposition 1}
\begin{proof}
To simplify notation and without loss of generality, we prove Proposition \ref{prop:property-of-ECMR} for the bivariate ECMR model. Suppose that
\begin{align*}
    &X=F_X^{-1}(F_{\mathcal{E}(g)}(\epsilon_x)),\\
    &Y=F_Y^{-1}(F_{\mathcal{E}(g)}(\epsilon_y)),\\
    &\mathcal{C}_{\mathcal{E}_2}(t_x,t_y)=\boldsymbol{F}_{\mathcal{E}_2(\mathbf{0},\boldsymbol{R},g)}(F^{-1}_{\mathcal{E}(g)}(t_x),F^{-1}_{\mathcal{E}(g)}(t_y)),
\end{align*}
where $F_X$ and $F_Y$ are the cumulative distribution functions (CDF) of the continuous random variables $X$ and $Y$, respectively, and $\epsilon_x$ and $\epsilon_y$ are the corresponding residual errors. Define $G(a,b)$ as the joint CDF for the residual errors $(\epsilon_x,\epsilon_y)^\top$. Then we have
\begin{align*}
    G(a,b)=&\Pr(\epsilon_x\leq a,\epsilon_y\leq b)\\
    =&\Pr(F_X^{-1}(F_{\mathcal{E}(g)}(\epsilon_x))\leq F_X^{-1}(F_{\mathcal{E}(g)}(a)),F_Y^{-1}(F_{\mathcal{E}(g)}(\epsilon_y))\leq F_Y^{-1}(F_{\mathcal{E}(g)}(b)))\\
    =&\Pr(X\leq F_X^{-1}(F_{\mathcal{E}(g)}(a)), Y\leq F_Y^{-1}(F_{\mathcal{E}(g)}(b)))\\
    =&\Pr(F_X(X)\leq F_X(F_X^{-1}(F_{\mathcal{E}(g)}(a))), F_Y(Y)\leq F_Y(F_Y^{-1}(F_{\mathcal{E}(g)}(b)))),\\
    =&\boldsymbol{F}_{\mathcal{E}_2(\mathbf{0},\boldsymbol{R},g)}(F^{-1}_{\mathcal{E}(g)}(F_X(F_X^{-1}(F_{\mathcal{E}(g)}(a)))),F^{-1}_{\mathcal{E}(g)}(F_Y(F_Y^{-1}(F_{\mathcal{E}(g)}(b)))))\\
    =&\boldsymbol{F}_{\mathcal{E}_2(\mathbf{0},\boldsymbol{R},g)}(a,b).
\end{align*}
Therefore, we conclude that $(\epsilon_x,\epsilon_y)^\top\sim \mathcal{E}_2(\mathbf{0},\boldsymbol{R},g)$.

\end{proof}

\subsection{Proof of Theorem \ref{thm:iden-NIE}}
The identification formula for $\theta_1(a_1,a_2,a_2)$ with $(a_1,a_2)\in\{(0,0),(1,0),(1,1)\}$ follows from Theorem 3.1 in \cite{cheng2024semiparametric} by noting that $\bM=(\bM^1, \bM^2)$ can be treated as a single unit. 

Next, we show the identification formula for $\theta_1(1,a_1,a_2)$ with $(a_1,a_2)\in\{(1,0),(0,1)\}$. 
By the law of total expectation (LOTE), we have 
\begin{align}
    &E\left\{w(\bV,N){N}^{-1}\sum_{j=1}^NY_{\buj}(1,\bM^k(0),\bM^{-k}(1))\right\}\nonumber\\
    =&E\lb E\lb w(\bV,N){N}^{-1}\sum_{j=1}^NY_{\buj}(1,\bM^k(0),\bM^{-k}(1))|\bC,N\rb\rb\label{eq:proof-set-II}.
\end{align}
Then we have
\begin{align*}
    \eqref{eq:proof-set-II}=&E\lb \int E\lb w(\bV,N){N}^{-1}\sum_{j=1}^NY_{\buj}(1,\bM^k(0),\bM^{-k}(1))|\bM^{k}(0)=\bm^{k},\bC,N\rb dP_{\bM^{k}(0)|\bC,N}(\bm^{k})\rb\\
    =&E\lb w(\bV,N){N}^{-1}\sum_{j=1}^N\int E\lb Y_{\buj}(1,\bm^k,\bM^{-k}(1))|\bC,N\rb dP_{\bM^{k}(0)|\bC,N}(\bm^{k})\rb,
\end{align*}
where the first equality follows from LOTE and the second equality follows from Assumption \ref{assump:homogeneous-between-mediator-ignorability}. Finally, we conclude from 
\begin{align*}
    &\int E\lb Y_{\buj}(1,\bm^k,\bM^{-k}(1))|\bC,N\rb dP_{\bM^{k}(0)|\bC,N}(\bm^{k})\\
    =&\int \int E\lb Y_{\buj}(1,\bm^k,\bM^{-k}(1))|\bM^{-k}(1)=\bm^{-k},\bC,N\rb dP_{\bM^{-k}(1)|\bC,N}(\bm^{-k})dP_{\bM^{k}(0)|\bC,N}(\bm^{k})\\
    =&\int \int E\lb Y_{\buj}(1,\bm^k,\bM^{-k}(1))|A=1,\bM^{-k}(1)=\bm^{-k},\bC,N\rb dP_{\bM^{-k}|A=1,\bC,N}(\bm^{-k})\\
    &dP_{\bM^{k}|A=0,\bC,N}(\bm^{k})~~~\text{(Assumption \ref{assump:2mediatiors-total})}\\
    =&\int \int E\lb Y_{\buj}(1,\bm^k,\bM^{-k}(1))|A=1,\bM^k(1)=\bm^k,\bM^{-k}(1)=\bm^{-k},\bC,N\rb \\
    &dP_{\bM^{-k}|A=1,\bC,N}(\bm^{-k})dP_{\bM^{k}|A=0,\bC,N}(\bm^{k})~~\text{(Assumption \ref{assump:sequential-ignorability})}\\
    =&\int \int \eta_{\buj}(1,\bm^k,\bm^{-k},\bC,N)  dP_{\bM^{-k}|A=1,\bC,N}(\bm^{-k})dP_{\bM^{k}|A=0,\bC,N}(\bm^{k})~\text{(Assumption \ref{assump:2mediatiors-total})},
\end{align*}
where the first equality follows from LOTE. 

Next, we prove the identification formula for $\theta_2(k-1,2-k)$ with $k\in\{1,2\}$. We prove for $k=1$ and similar procedure applies to the case with $k=2$. By LOTE, we have that
\begin{align}
    &E\lb w(\bV,N){N}^{-1}\sum_{j=1}^NY_{\buj}(1,M^1_{\buj}(1),\bM^1_{\bullet,-j}(0),M^2_{\buj}(1),\bM^2_{\bullet,-j}(1))\rb\nonumber\\
    =&E\lb E\lb w(\bV,N){N}^{-1}\sum_{j=1}^NY_{\buj}(1,M^1_{\buj}(1),\bM^1_{\bullet,-j}(0),\bM^2(1))|\bC,N\rb\rb.\label{eq:theta2-set-III}
\end{align}
Then
\begin{align*}
    \eqref{eq:theta2-set-III}=E\bigg\{ &w(\bV,N){N}^{-1}\sum_{j=1}^N\int E\lb Y_{\buj}(1,M^1_{\buj}(1),\bM^1_{\bullet,-j}(0),\bM^2(1))|\bM^1_{\bumj}(0)=\bm^1_{\bumj},\bC,N\rb \\
    &dP_{\bM^1_{\bumj}|A=0,\bC,N}(\bm^1_{\bumj})\bigg\}~~\text{(LOTE, Assumption \ref{assump:2mediatiors-total})}.
\end{align*}
Finally, we conclude from
\begin{align*}
    &E\lb Y_{\buj}(1,M^1_{\buj}(1),\bM^1_{\bumj}(0),\bM^2(1))|\bM^1_{\bumj}(0)=\bm^1_{\bumj},\bC,N\rb\\
    =&\int  E\lb Y_{\buj}(1,m_{\buj}^1,\bm^1_{\bumj},\bm^2)|M^1_{\buj}(1)=m_{\buj}^1,\bM^1_{\bumj}(0)=\bm^1_{\bumj},\bM^2(1)=\bm^2,\bC,N\rb \\
    &dP_{M^1_{\buj}(1),\bM^2(1)|\bM^1_{\bumj}(0)=\bm^1_{\bumj},\bC,N}(m_{\buj}^1,\bm^2)~~\text{(LOTE)}\\
    =&\int  E\lb Y_{\buj}(1,m_{\buj}^1,\bm^1_{\bumj},\bm^2)|M^1_{\buj}(1)=m_{\buj}^1,\bM^1_{\bumj}(0)=\bm^1_{\bumj},\bM^2(1)=\bm^2,\bC,N\rb \\
    &dP_{M^1_{\buj}(1),\bM^2(1)|\bC,N}(m_{\buj}^1,\bm^2)~~\text{(Assumption \ref{assump:heterogeneous-between-mediator-ignorability})}\\
    =&\int \eta_{\buj}(1, m_{\buj}^1,\bm^1_{\bullet,-j}, \bm^2, \bC,N) dP_{M^1_{\buj},\bM^2|A=1,\bC,N}(m_{\buj}^1,\bm^2)~\text{(Assumptions \ref{assump:2mediatiors-total}-\ref{assump:sequential-ignorability})}.
\end{align*}

\subsection{Proof of Theorem \ref{thm:iden-NIE-K}}
The identification formula for 
$\theta_1^{I}(a^\ast,\ba_{\mJ})$ follows from Theorem 3.1 in \cite{cheng2024semiparametric} by treating $\bM$ as a single unit. 

The proof of the identification formulas for $\theta^{II}_1(1,\ba_{\mJ^\ast})$ and $\theta_2(k,\mJ^\ast)$ follows reasoning analogous to that of Theorem \ref{thm:iden-NIE}. We detail the derivation under Assumption \ref{K-assump:homogeneous-between-mediator-ignorability-INT}; the case under Assumption \ref{K-assump:homogeneous-between-mediator-ignorability-INT-alternative} follows via an analogous argument. 
First, we prove the identification formulas for $\theta^{II}_1(1,\ba_{\mJ^\ast})$. To proceed, by LOTE, we have that
\begin{align}
    &E\left\{w(\bV,N){N}^{-1}\sum_{j=1}^NY_{\buj}(1,\bM^{\mJ^\ast}(0),\bM^{\mK\backslash
    \mJ^\ast}(1))\right\}\nonumber\\
    =&E\lb E\lb w(\bV,N){N}^{-1}\sum_{j=1}^NY_{\buj}(1,\bM^{\mJ^\ast}(0),\bM^{\mK\backslash
    \mJ^\ast}(1))|\bC,N\rb\rb\label{eq:proof-set-II-v2}.
\end{align}
Then we have
\begin{align*}
    &\eqref{eq:proof-set-II-v2}\\
    =&E\lb \int E\lb w(\bV,N){N}^{-1}\sum_{j=1}^NY_{\buj}(1,\bM^{\mJ^\ast}(0),\bM^{\mK\backslash \mJ^\ast}(1))|\bM^{\mJ^\ast}(0)=\bm^{\mJ^\ast},\bC,N\rb dP_{\bM^{\mJ^\ast}(0)|\bC,N}(\bm^{\mJ^\ast})\rb\\
    =&E\lb w(\bV,N){N}^{-1}\sum_{j=1}^N\int E\lb Y_{\buj}(1,\bm^{\mJ^\ast},\bM^{\mK\backslash\mJ^\ast}(1))|\bC,N\rb dP_{\bM^{\mJ^\ast}(0)|\bC,N}(\bm^{\mJ^\ast})\rb,
\end{align*}
where the first equality follows from LOTE and the second equality follows from Assumption \ref{K-assump:homogeneous-between-mediator-ignorability-INT}. Finally, we conclude from 
\begin{align*}
    &\int E\lb Y_{\buj}(1,\bm^{\mJ^\ast},\bM^{\mK\backslash\mJ^\ast}(1))|\bC,N\rb dP_{\bM^{\mJ^\ast}(0)|\bC,N}(\bm^{\mJ^\ast})\\
    =&\int \int E\lb Y_{\buj}(1,\bm^{\mJ^\ast},\bM^{\mK\backslash\mJ^\ast}(1))|\bM^{\mK\backslash\mJ^\ast}(1)=\bm^{\mK\backslash\mJ^\ast},\bC,N\rb dP_{\bM^{\mK\backslash\mJ^\ast}(1)|\bC,N}(\bm^{\mK\backslash
    \mJ^\ast})dP_{\bM^{\mJ^\ast}(0)|\bC,N}(\bm^{\mJ^\ast})\\
    =&\int \int E\lb Y_{\buj}(1,\bm^{\mJ^\ast},\bM^{\mK\backslash\mJ^\ast}(1))|A=1,\bM^{\mK\backslash\mJ^\ast}(1)=\bm^{\mK\backslash\mJ^\ast},\bC,N\rb dP_{\bM^{\mK\backslash\mJ^\ast}|A=1,\bC,N}(\bm^{\mK\backslash\mJ^\ast})\\
    &dP_{\bM^{\mJ^\ast}|A=0,\bC,N}(\bm^{\mJ^\ast})~~~\text{(Assumption \ref{assump:2mediatiors-total})}\\
    =&\int \int E\lb Y_{\buj}(1,\bm^{\mJ^\ast},\bM^{\mK\backslash\mJ^\ast}(1))|A=1,\bM^{\mJ^\ast}(1)=\bm^{\mJ^\ast},\bM^{\mK\backslash\mJ^\ast}(1)=\bm^{\mK\backslash\mJ^\ast},\bC,N\rb \\
    &dP_{\bM^{\mK\backslash\mJ^\ast}|A=1,\bC,N}(\bm^{\mK\backslash\mJ^\ast})dP_{\bM^{\mJ^\ast}|A=0,\bC,N}(\bm^{\mJ^\ast})~~\text{(Assumption \ref{assump:sequential-ignorability})}\\
    =&\int \int \eta_{\buj}(1,\bm^{\mJ^\ast},\bm^{\mK\backslash\mJ^\ast},\bC,N)  dP_{\bM^{\mK\backslash\mJ^\ast}|A=1,\bC,N}(\bm^{\mK\backslash\mJ^\ast})dP_{\bM^{\mJ^\ast}|A=0,\bC,N}(\bm^{\mJ^\ast})~\text{(Assumption \ref{assump:2mediatiors-total})},
\end{align*}
where the first equality follows from LOTE.

Second, we prove the identification formulas for $\theta_2(k,\mJ^\ast)$. By LOTE, we have that
\begin{align}
    &E\left\{ w(\bV,N){N}^{-1}\sum_{j=1}^NY_{\buj}(1,\bM^{\mJ^\ast}(0),\bM_{\bumj}^k(0),M_{\buj}^k(1),\bM^{\mK\backslash (\mJ^\ast\cup\lb k\rb)}(1))\right\}\nonumber\\
    =&E\lb E\lb w(\bV,N){N}^{-1}\sum_{j=1}^NY_{\buj}(1,\bM^{\mJ^\ast}(0),\bM_{\bumj}^k(0),M_{\buj}^k(1),\bM^{\mK\backslash (\mJ^\ast\cup\lb k\rb)}(1))|\bC,N\rb\rb.\label{eq:theta2-set-III-k}
\end{align}
Then
\begin{align*}
    &\eqref{eq:theta2-set-III-k}
\\=&E\bigg\{ \frac{w(\bV,N)}{N}\sum_{j=1}^N\int E\Big\{ Y_{\buj}(1,\bm^{\mJ^\ast},\bm^k_{\bumj},M_{\buj}^k(1),\bM^{\mK\backslash (\mJ^\ast\cup\lb k\rb)}(1))|\bM^{\mJ^\ast}(0)=\bm^{\mJ^\ast}\\
    &\bM^k_{\bumj}(0)=\bm^k_{\bumj},\bC,N\Big\} dP_{\bM^{\mJ^\ast},\bM^k_{\bumj}|A=0,\bC,N}(\bm^{\mJ^\ast},\bm^k_{\bumj})\bigg\}~~~\text{(LOTE, Assumption \ref{assump:2mediatiors-total})}.
\end{align*}
Finally, we conclude from
\begin{align*}
    &E\lb Y_{\buj}(1,\bm^{\mJ^\ast},\bm^k_{\bumj},M_{\buj}^k(1),\bM^{\mK\backslash (\mJ^\ast\cup\lb k\rb)}(1))|\bM^{\mJ^\ast}(0)=\bm^{\mJ^\ast},\bM^k_{\bumj}(0)=\bm^k_{\bumj},\bC,N\rb\\
    =&\int E\Big\{ Y_{\buj}(1,\bm^{\mJ^\ast},\bm^k_{\bumj},m_{\buj}^k,\bm^{\mK\backslash (\mJ^\ast\cup\lb k\rb)})|\bM^{\mJ^\ast}(0)=\bm^{\mJ^\ast},\bM^k_{\bumj}(0)=\bm^k_{\bumj},M_{\buj}^k(1)=m_{\buj}^k, \\
    &\bM^{\mK\backslash (\mJ^\ast\cup\lb k\rb)}(1)=\bm^{\mK\backslash (\mJ^\ast\cup\lb k\rb)},\bC,N\Big\} \\
    &dP_{M^k_{\buj}(1),\bM^{\mK\backslash (\mJ^\ast\cup\lb k\rb)}(1)|\bM^{\mJ^\ast}(0)=\bm^{\mJ^\ast},\bM^k_{\bumj}(0)=\bm^k_{\bumj},\bC,N}(m_{\buj}^k,\bm^{\mK\backslash (\mJ^\ast\cup\lb k\rb)})\text{(LOTE)}\\
    =&\int E\Big\{ Y_{\buj}(1,\bm^{\mJ^\ast},\bm^k_{\bumj},m_{\buj}^k,\bm^{\mK\backslash (\mJ^\ast\cup\lb k\rb)})|A=1,\bM^{\mJ^\ast}(1)=\bm^{\mJ^\ast},\bM^k_{\bumj}(1)=\bm^k_{\bumj},M_{\buj}^k(1)=m_{\buj}^k, \\
    &\bM^{\mK\backslash (\mJ^\ast\cup\lb k\rb)}(1)=\bm^{\mK\backslash (\mJ^\ast\cup\lb k\rb)},\bC,N\Big\} dP_{M^k_{\buj},\bM^{\mK\backslash (\mJ^\ast\cup\lb k\rb)}|A=1,\bC,N}(m_{\buj}^k,\bm^{\mK\backslash (\mJ^\ast\cup\lb k\rb)})\\
    &\text{(Assumptions \ref{assump:2mediatiors-total}, \ref{assump:sequential-ignorability},and \ref{K-assump:generalized-heterogeneous-between-mediator-ignorability})}\\
    =&\int  \eta_{\buj}(1, \bm, \bC,N) dP_{M^k_{\buj},\bM^{\mK\backslash (\mJ^\ast\cup\lb k\rb)}|A=1,\bC,N}(m_{\buj}^k,\bm^{\mK\backslash (\mJ^\ast\cup\lb k\rb)})~~~\text{(Assumption \ref{assump:2mediatiors-total}).}
\end{align*}

\subsection{Proof of Theorem \ref{thm:EIF}}
\subsubsection{Preliminary}
In this section, we derive the EIFs for the parameters $\theta_1$ and $\theta_2$, enabling the construction of the EIFs for the risk difference estimands via additive combinations. For estimands defined on the ratio scale, the corresponding EIFs can be derived by applying the chain rule, similar to the methods used in \cite{cheng2024semiparametric}. We follow standard techniques in semiparametric theory to find the pathwise derivatives of the target functionals \citep{tsiatis2006semiparametric}. Specifically, we derive the EIFs in the nonparametric sense, imposing no restrictions on the likelihood of the observed data vector $\mO$. To facilitate these derivations, we first introduce the following preliminaries. Denote $f(\mO)$ as the joint density function of $\mO$. Consider the following factorization
\begin{align*}
f(\mO)=&f(\bY|\bM,A,\bC,N)\times f(\bM|A,\bC,N)\times f(A)\times f(\bC,N)\\
\propto &f(\bY|\bM,A,\bC,N)\times f(\bM|A,\bC,N)\times f(\bC,N),
\end{align*}
where the treatment propensity $f(a) = \pi^a(1 - \pi)^{1 - a}$ is known a priori for a CRT. Following Theorem 4.4 and Theorem 4.5 in \cite{tsiatis2006semiparametric}, the tangent space $\mathcal{F}$ is the entire Hilbert space $\mathcal{H}\equiv\{h:E\{h(\mO)\}=0,E\{h(\mO)^2\}<\infty\}$. We consider a parametric submodel for the distribution of $\mO$, $\mathcal{P}=\{f_{\epsilon}(\mO):\epsilon\in\mathcal{T}\subseteq\mathbb{R}\}$, where $f_{0}=f$ is the true density and $E_{0}=E$ denotes the expectation with respect to $f_0$ for ease of notation. Let $\nabla_{\epsilon=0}f={\partial f}/{\partial\epsilon}|_{\epsilon=0}$ denote the pathwise derivative operator evaluated at $\epsilon = 0$. Consider the following orthogonal decomposition of the score vector
\begin{equation*}S(\mO)=S(\bY|\bM,A,\bC,N)+S(\bM|A,\bC,N)+S(\bC,N),
\end{equation*}
where 
\begin{align*}
    &S(\mO)=\nabep \log f(\mO),\\
    &S(\bY|\bM,A,\bC,N)=\nabep \log f(\bY|\bM,A,\bC,N),\\
    &S(\bM|A,\bC,N)=\nabep \log f(\bM|A,\bC,N),\\
    &S(\bC,N)=\nabep \log f(\bC,N).
\end{align*}
We define $\theta_{i,\epsilon}$ as the value of $\theta_i$ in the submodel and the truth is attained at $\epsilon=0$, i.e., $\theta_{i,0}=\theta_i$. Theorem 3.2 in \cite{tsiatis2006semiparametric} implies that the influence function $\varphi(\mO;\theta_i)\in\mathcal{H}$ for the submodel can be characterized by 
\begin{equation}\label{eq:influfuncchara}
  E\{\varphi(\mO;\theta_i)S(\mO)\}=\nabep \theta_{i,\epsilon}.
\end{equation}
\cite{kennedy2022semiparametric} showed that there is at most one solution to the differential equation \eqref{eq:influfuncchara}, without imposing any restrictions on the likelihood of $\mO$. By Theorem 4.3 in \cite{tsiatis2006semiparametric}, the EIF is indeed $\varphi(\mO;\theta_i)$ because the tangent space is the entire Hilbert space. Subsequently, the EIF solves Equation \eqref{eq:influfuncchara}. The following two lemmas are needed to facilitate the calculations of the EIF.
\begin{lemma}[\emph{Quotient rule}]\label{lemma:quotient}
    Let $\theta=\theta_n/\theta_d$. Suppose that $\varphi(\mO;\theta_n)$ and $\varphi(\mO;\theta_d)$ are the EIFs for the $\theta_n$ and $\theta_d$, respectively. Then
    $$\varphi(\mO;\theta)=\frac{1}{\theta_d}\varphi(\mO;\theta_n)-\frac{\theta_n}{\theta_d^2}\varphi(\mO;\theta_d).$$
\end{lemma}
\begin{proof}
    This result follows from the chain rule; one proof is provided by the Example 6 in \cite{kennedy2022semiparametric}.
\end{proof}
Lemma \ref{lemma:quotient} provides the expression for the EIF of a parameter represented as a quotient.
\begin{lemma}\label{lemma:score}
  Consider a random vector $(Y,X,Z)^\top$. Let $S(\bullet)$ be its associated score vector. Then 
  $$S(x|z)=E\{S(Y,x,z)|X=x,Z=z\}-E\{S(Y,X,z)|Z=z\}.$$
\end{lemma}
\begin{proof}
Since score has mean zero, we have that
    \begin{align*}
        &E\{S(Y,x,z)|X=x,Z=z\}=E\{S(Y|x,z)+S(x,z)|X=x,Z=z\}=S(x,z),\\
        &E\{S(Y,X,z)|Z=z\}=E\{S(Y,X|z)+S(z)|Z=z\}=S(z),
    \end{align*}
which completes the proof by noting that $S(x,z)=S(x|z)+S(z)$.
\end{proof}

\subsubsection{Main proof for $\theta_1^{I}(a^\ast,\ba_\mJ)$}
We derive the EIF for the mediation functional $\theta_1^{I}(a^\ast,\ba_\mJ)$. Let $\theta_d=E\lb w(\bV,N)\rb$ and $\theta_n(a^\ast,\ba_\mJ)=\theta_d\times\theta_1^{I}(a^\ast,\ba_\mJ)$. Let $\mathcal{I}(\bullet)$ denote the indicator function. By Theorem 4.1 in \cite{cheng2024semiparametric}, the EIF for $\theta_n(a^\ast,\ba_\mJ)$ is given by
\begin{align*}
    \varphi(\mO;\theta_n(a^\ast,\ba_\mJ))=&\phi(\mO;a^\ast,\ba_\mJ)-\theta_n(a^\ast,\ba_\mJ),
\end{align*}
where 
\begin{align*}
\eta_{\buj}(a^\ast,\bM,\bC,N)=&E\left\{Y_{\buj}|A=a^\ast,\bM,\bC,N\right\},\\
\eta^\ast_{\buj,a^\ast\ba_{\mJ1}}(\bC,N)=&\int \eta_{\buj}(a^\ast,\bM,\bC,N)f(\bm|\ba_{\mJ1},\bC,N)d\bm,\\
    \phi(\mO;a^\ast,\ba_\mJ)=&\frac{w(\bV,N)}{N}\sum_{j=1}^N\bigg \{\frac{\mathcal{I}(A=a^\ast)}{\Pr(A=a^\ast)}\frac{f(\bM|\ba_{\mJ1},\bC,N)}{f(\bM|a^\ast,\bC,N)}\lb Y_{\buj}-\eta_{\buj}(a^\ast,\bM,\bC,N)\rb+\\
    &\frac{\mathcal{I}(A=\ba_{\mJ1})}{\Pr(A=\ba_{\mJ1})}\lb\eta_{\buj}(a^\ast,\bM,\bC,N)-\eta^{\ast}_{\buj,a^\ast\ba_{\mJ1}}(\bC,N)\rb+\eta^{\ast}_{\buj,a^\ast\ba_{\mJ1}}(\bC,N)\bigg\}.
\end{align*}
To find the EIF for $\theta_d$, $\varphi(\mO;\theta_d)$, we note that 
\begin{align*}
 \nabep \theta_d=&\int\int\int w(\bv,n)S(\bc,n)f(\bc,n)d\bc dn\\
 =&\int\int\int w(\bv,n)\lb E\{S(\bY,\bM,A,\bc,n)|\bc,n\}-E\lb S(\mO)\rb\rb f(\bc,n)d\bc dn~\text{(Lemma \ref{lemma:score})}\\
 =&E\lb w(\bV,N)S(\mO)\rb-E\lb E\lb w(\bV,N)\rb S(\mO)\rb\\
 =&E\lb \left(w(\bV,N)-E\lb w(\bV,N)\rb\right)S(\mO)\rb.
\end{align*}
Thus, $\varphi(\mO;\theta_d)=w(\bV,N)-\theta_d$. By Lemma \ref{lemma:quotient}, we conclude that 
\begin{align*}
    \varphi(\mO;\theta_1(a^\ast,\ba_\mJ))=E\lb w(\bV,N)\rb^{-1}\left\{\phi(\mO;a^\ast,\ba_\mJ)-\theta^{I}_1(a^\ast,\ba_\mJ)w(\bV,N)\right\}.
\end{align*}

\subsubsection{Main proof for $\theta_1^{II}(1,\ba_{\mJ^\ast})$}
We derive the EIF for the mediation functional $\theta_1^{II}(1,\ba_{\mJ^\ast})$. Similarly, it suffices to derive the EIFs for $\theta_n(1,\ba_{\mJ^\ast})\equiv\theta^{II}_1(1,\ba_{\mJ^\ast})\times\theta_d$. Note that $\theta_{n,\epsilon}(1,\ba_{\mJ^\ast})$ can be expressed as 
\begin{align*}
    \theta_{n,\epsilon}(1,\ba_{\mJ^\ast})=&\int n^{-1}\sum_{j=1}^n\int\int w(\bv,n)\eta_{\buj,\epsilon}(1,\bm,\bc,n)f_\epsilon(\bm^{J^\ast}|0,\bc,n)f_\epsilon(\bm^{\mK\backslash J^\ast}|1,\bc,n)f_\epsilon(\bc,n) d\bm d\bc dn.
\end{align*}
Applying the chain rule with respect to the pathwise derivative yields:
\begin{align*}
    &\nabep \theta_{n,\epsilon}(1,\ba_{\mJ^\ast})\\
    =&\underbrace{\int n^{-1}\sum_{j=1}^n\int\int w(\bv,n)\nabep\eta_{\buj,\epsilon}(1,\bm,\bc,n)f(\bm^{J^\ast}|0,\bc,n)f(\bm^{\mK\backslash J^\ast}|1,\bc,n)f(\bc,n) d\bm d\bc dn}_{\text{T}_1}+\\
    &\underbrace{\int n^{-1}\sum_{j=1}^n\int\int w(\bv,n)\eta_{\buj}(1,\bm,\bc,n)\nabep f(\bm^{J^\ast}|0,\bc,n)f(\bm^{\mK\backslash J^\ast}|1,\bc,n)f(\bc,n) d\bm d\bc dn}_{\text{T}_2}+\\
    &\underbrace{\int n^{-1}\sum_{j=1}^n\int\int w(\bv,n)\eta_{\buj}(1,\bm,\bc,n)f(\bm^{J^\ast}|0,\bc,n)\nabep 
 f(\bm^{\mK\backslash J^\ast}|1,\bc,n)f(\bc,n) d\bm d\bc dn}_{\text{T}_3}+\\
    &\underbrace{\int n^{-1}\sum_{j=1}^n\int\int w(\bv,n)\eta_{\buj}(1,\bm,\bc,n)f(\bm^{J^\ast}|0,\bc,n)f(\bm^{\mK\backslash J^\ast}|1,\bc,n)\nabep f_\epsilon(\bc,n) d\bm d\bc dn}_{\text{T}_4}.
\end{align*}
For the term $\text{T}_1$, we have
\begin{align*}
    \text{T}_1=&\int n^{-1}\sum_{j=1}^n\int\int w(\bv,n)\nabep\eta_{\buj,\epsilon}(1,\bm,\bc,n)f(\bm^{J^\ast}|0,\bc,n)f(\bm^{\mK\backslash J^\ast}|1,\bc,n)f(\bc,n) d\bm d\bc dn\\
    =&\int n^{-1}\sum_{j=1}^n\int\int w(\bv,n)\int y_{\buj}\nabep f_\epsilon(\by|\bm,1,\bc,n)f(\bm^{J^\ast}|0,\bc,n)f(\bm^{\mK\backslash J^\ast}|1,\bc,n)f(\bc,n) d\by d\bm d\bc dn\\
    =&\int n^{-1}\sum_{j=1}^n\int\int w(\bv,n)\int y_{\buj}S(\by|\bm,1,\bc,n)f(\by|\bm,1,\bc,n)f(\bm^{J^\ast}|0,\bc,n)f(\bm^{\mK\backslash J^\ast}|1,\bc,n)\\
    &f(\bc,n) d\by d\bm d\bc dn\\
    =&\int n^{-1}\sum_{j=1}^n\int\int w(\bv,n)\int y_{\buj}\lb S(\by,\bm,1,\bc,n)-E\lb S(\bY,\bm,1,\bc,n)|\bm,1,\bc,n\rb\rb\\
    &\frac{f(\bm^{J^\ast}|0,\bc,n)f(\bm^{\mK\backslash J^\ast}|1,\bc,n)}{f(\bm|1,\bc,n)}f(\by|\bm,1,\bc,n)f(\bm|1,\bc,n)f(\bc,n) d\by d\bm d\bc dn~~\text{(Lemma \ref{lemma:score})}\\
    =&\int n^{-1}\sum_{j=1}^n\int\int w(\bv,n)\int \lb y_{\buj}-\eta_{\buj}(1,\bm,\bc,n)\rb \frac{f(\bm^{J^\ast}|0,\bc,n)f(\bm^{\mK\backslash J^\ast}|1,\bc,n)}{f(\bm|1,\bc,n)}S(\by,\bm,1,\bc,n)\\
    &f(\by|\bm,1,\bc,n)f(\bm|1,\bc,n)f(\bc,n) d\by d\bm d\bc dn\\
    =&E\lb  \frac{w(\bV,N)}{N}\sum_{j=1}^N\frac{\mathcal{I}(A=1)}{\Pr(A=1)}\lb Y_{\buj}-\eta_{\buj}(1,\bM,\bC,N)\rb \frac{f(\bM^{J^\ast}|0,\bC,N)f(\bM^{\mK\backslash J^\ast}|1,\bC,N)}{f(\bM|1,\bC,N)}S(\mO)\rb.
\end{align*}
For the term $\text{T}_2$, we have
\begin{align*}
    \text{T}_2=&\int n^{-1}\sum_{j=1}^n\int\int w(\bv,n)\eta_{\buj}(1,\bm,\bc,n)\nabep f(\bm^{J^\ast}|0,\bc,n)f(\bm^{\mK\backslash J^\ast}|1,\bc,n)f(\bc,n) d\bm d\bc dn\\
    =&\int n^{-1}\sum_{j=1}^n\int\int w(\bv,n)\eta_{\buj}(1,\bm,\bc,n)S(\bm^{J^\ast}|0,\bc,n)f(\bm^{J^\ast}|0,\bc,n)f(\bm^{\mK\backslash J^\ast}|1,\bc,n)f(\bc,n) d\bm d\bc dn\\
    =&\int n^{-1}\sum_{j=1}^n\int\int w(\bv,n)\eta_{\buj}(1,\bm,\bc,n) \times \\
    &\lb E\lb S(\bY,\bM^{\mK\backslash J^\ast},\bm^{J^\ast},0,\bc,n)|\bm^{J^\ast},0,
    \bc,n\rb-E\lb S(\bY,\bM,0,\bc,n)|0,\bc,n\rb\rb\\
    &f(\bm^{J^\ast}|0,\bc,n)f(\bm^{\mK\backslash J^\ast}|1,\bc,n)f(\bc,n) d\bm d\bc dn~~\text{(Lemma \ref{lemma:score})}\\  
    =&\int n^{-1}\sum_{j=1}^n\int w(\bv,n) \int \int \lb\widetilde{\tau}_{\buj}^{\mJ^\ast}(\bm^{J^\ast},\bc,n)- \tau^{J^\ast}_{\buj}(\bc,n)\rb S(\by,\bm,0,\bc,n) \\
    &f(\by|\bm,0,\bc,n)f(\bm|0,\bc,n)f(\bc,n) d\by d\bm d\bc dn\\
    =&E\lb  \frac{w(\bV,N)}{N}\sum_{j=1}^N\frac{\mathcal{I}(A=0)}{\Pr(A=0)}\lb\widetilde{\tau}_{\buj}^{\mJ^\ast}(\bM^{J^\ast},\bC,N)- \tau^{J^\ast}_{\buj}(\bC,N)\rb S(\mO)\rb,
\end{align*}
where 
\begin{align*}
  &\tau^{J^\ast}_{\buj}(\bc,n)=\int\eta_{\buj}(1,\bm,\bc,n)f(\bm^{J^\ast}|0,\bc,n)f(\bm^{\mK\backslash J^\ast}|1,\bc,n)d\bm,\\
  &\widetilde{\tau}_{\buj}^{\mJ^\ast}(\bm^{J^\ast},\bc,n)=\int\eta_{\buj}(1,\bm,\bc,n)f(\bm^{\mK\backslash J^\ast}|1,\bc,n)d\bm^{\mK\backslash J^\ast}.
\end{align*}
For the term $\text{T}_3$, we have
\begin{align*}
    \text{T}_3=&\int n^{-1}\sum_{j=1}^n\int\int w(\bv,n)\eta_{\buj}(1,\bm,\bc,n) f(\bm^{J^\ast}|0,\bc,n)\nabep f(\bm^{\mK\backslash J^\ast}|1,\bc,n)f(\bc,n) d\bm d\bc dn\\
    =&\int n^{-1}\sum_{j=1}^n\int\int w(\bv,n)\eta_{\buj}(1,\bm,\bc,n)f(\bm^{J^\ast}|0,\bc,n)S(\bm^{\mK\backslash J^\ast}|1,\bc,n)f(\bm^{\mK\backslash J^\ast}|1,\bc,n)f(\bc,n) d\bm d\bc dn\\
    =&\int n^{-1}\sum_{j=1}^n\int\int w(\bv,n)\eta_{\buj}(1,\bm,\bc,n) \times \\
    &\lb E\lb S(\bY,\bM^{J^\ast},\bm^{\mK\backslash J^\ast},1,\bc,n)|\bm^{\mK\backslash J^\ast},1,
    \bc,n\rb-E\lb S(\bY,\bM,1,\bc,n)|1,\bc,n\rb\rb\\
    &f(\bm^{J^\ast}|0,\bc,n)f(\bm^{\mK\backslash J^\ast}|1,\bc,n)f(\bc,n) d\bm d\bc dn~~\text{(Lemma \ref{lemma:score})}\\  
    =&\int n^{-1}\sum_{j=1}^n\int w(\bv,n) \int \int \lb\check{\tau}^{\mJ^\ast}_{\buj}(\bm^{\mK\backslash J^\ast},\bc,n)- \tau^{J^\ast}_{\buj}(\bc,n)\rb S(\by,\bm,1,\bc,n) \\
    &f(\by|\bm,1,\bc,n)f(\bm|1,\bc,n)f(\bc,n) d\by d\bm d\bc dn\\
    =&E\lb  \frac{w(\bV,N)}{N}\sum_{j=1}^N\frac{\mathcal{I}(A=1)}{\Pr(A=1)}\lb\check{\tau}^{\mJ^\ast}_{\buj}(\bM^{\mK\backslash J^\ast},\bC,N)- \tau^{J^\ast}_{\buj}(\bC,N)\rb S(\mO)\rb,
\end{align*}
where $\check{\tau}^{\mJ^\ast}_{\buj}(\bm^{\mK\backslash J^\ast},\bc,n)=\int\eta_{\buj}(1,\bm,\bc,n)f(\bm^{J^\ast}|0,\bc,n)d\bm^{J^\ast}$. For the term $\text{T}_4$, we have
\begin{align*}
    \text{T}_4=&\int n^{-1}\sum_{j=1}^n\int\int w(\bv,n)\eta_{\buj}(1,\bm,\bc,n)f(\bm^{J^\ast}|0,\bc,n)f(\bm^{\mK\backslash J^\ast}|1,\bc,n)\nabep f_\epsilon(\bc,n) d\bm d\bc dn\\
    =&\int n^{-1}\sum_{j=1}^n\int\int w(\bv,n)\eta_{\buj}(1,\bm,\bc,n)f(\bm^{J^\ast}|0,\bc,n)f(\bm^{\mK\backslash J^\ast}|1,\bc,n) S(\bc,n)f(\bc,n) d\bm d\bc dn\\
    =&\int n^{-1}\sum_{j=1}^n\int\int w(\bv,n)\eta_{\buj}(1,\bm,\bc,n)f(\bm^{J^\ast}|0,\bc,n)f(\bm^{\mK\backslash J^\ast}|1,\bc,n)   \\
    &\lb E\lb S(\bY,\bM,A,\bc,n)|\bc,n\rb -E\lb S(\mO)\rb\rb f(\bc,n)d\bm d\bc dn~~\text{(Lemma \ref{lemma:score})}\\
    =&\int n^{-1}\sum_{j=1}^n\int w(\bv,n)\tau^{J^\ast}_{\buj}(\bc,n) E\lb S(\bY,\bM,A,\bc,n)|\bc,n\rb f(\bc,n) d\bc dn-E\lb S(\mO)\rb\theta_n(1,\ba_{\mJ^\ast})\\
    =&\int n^{-1}\sum_{j=1}^n\int \sum_{a=0}^1w(\bv,n) \tau^{J^\ast}_{\buj}(\bc,n)  S(\by,\bm,a,\bc,n)f(\by|\bm,a,\bc,n)\\
    &f(\bm|a,\bc,n)\Pr(A=a)f(\bc,n)d\by d\bm d\bc dn-E\lb S(\mO)\rb\theta_n(1,\ba_{\mJ^\ast})\\
    =&E\lb \left[\frac{w(\bV,N)}{N}\sum_{j=1}^N  \tau^{J^\ast}_{\buj}(\bC,N)-\theta_n(1,\ba_{\mJ^\ast})\right]S(\mO) \rb.
\end{align*}
It is straightforward to verify that
$\varphi(\mO;\theta_n(1,\ba_{\mJ^\ast}))\in\mathcal{H}$, where
\begin{align*}
    &\varphi(\mO;\theta_n(1,\ba_{\mJ^\ast}))=\phi(\mO;\theta_1(1,\ba_{\mJ^\ast}))-\theta_n(1,\ba_{\mJ^\ast}),
\end{align*}
and
\begin{align*}
    &\phi(\mO;\theta_1(1,\ba_{\mJ^\ast}))\\
    =&\frac{w(\bV,N)}{N}\sum_{j=1}^N\Bigg\{\frac{\mathcal{I}(A=0)}{\Pr(A=0)}\widetilde{\tau}_{\buj}^{\mJ^\ast}(\bM^{J^\ast},\bC,N) +\left(1-\sum_{a=0}^1\frac{\mathcal{I}(A=a)}{\Pr(A=a)}\right)\tau^{J^\ast}_{\buj}(\bC,N)+\\
    &\frac{\mathcal{I}(A=1)}{\Pr(A=1)}\Bigg[\frac{f(\bM^{J^\ast}|0,\bC,N)f(\bM^{\mK\backslash J^\ast}|1,\bC,N)}{f(\bM|1,\bC,N)}\left( Y_{\buj}-\eta_{\buj}(1,\bM,\bC,N)\right)+\check{\tau}^{\mJ^\ast}_{\buj}(\bM^{\mK\backslash J^\ast},\bC,N)\Bigg]\Bigg\}.
\end{align*}
Therefore, we conclude that $\varphi(\mO;\theta_n(1,\ba_{\mJ^\ast}))$ is the EIF for $\theta_{n}(1,\ba_{\mJ^\ast})$. Eventually, by Lemma \ref{lemma:quotient}, we conclude that 
\begin{align*}
    \varphi(\mO;\theta_1(1,\ba_{\mJ^\ast}))=\displaystyle\frac{\phi(\mO;\theta_1(1,\ba_{\mJ^\ast}))-\theta_1(1,\ba_{\mJ^\ast})w(\bV,N)}{E\lb w(\bV,N)\rb}.
\end{align*}

\subsubsection{Main proof for $\theta_2(k,\mJ^\ast)$}
We next derive the EIFs for the mediation functional $\theta_2(k,\mJ^\ast)$. With a slight abuse of notation, we define $\theta_n(k,\mJ^\ast) = \theta_d \times \theta_2(k,\mJ^\ast)$. Similarly, it suffices to derive the EIF for $\theta_n(k,\mJ^\ast)$. Note that $\theta_{n,\epsilon}(k,\mJ^\ast)$ can be expressed as 
\begin{align*}
    &\int n^{-1}\sum_{j=1}^n\int\int w(\bv,n) \eta_{\buj,_\epsilon}(1,\bm, \bc,n)f_\epsilon(\bm^{\mJ^\ast},\bm^k_{\bullet,-j}|0,\bc,n)f_\epsilon(m_{\buj}^k,\bm^{\mK\backslash (\mJ^\ast\cup\{k\})}|1,\bc,n)\\
    &f_\epsilon(\bc,n)d\bm d\bc dn.
\end{align*}
Applying the chain rule with respect to the pathwise derivative yields:
\begin{align*}
    &\nabep \theta_{n,\epsilon}(k,\mJ^\ast)=\sum_{j=1}^4 T_j,
\end{align*}
where
\begin{align*}
    T_1=&\int n^{-1}\sum_{j=1}^n\int\int w(\bv,n)\nabep\eta_{\buj,_\epsilon}(1, \bm, \bc,n)f(\bm^{\mJ^\ast},\bm^k_{\bullet,-j}|0,\bc,n)f(m_{\buj}^k,\bm^{\mK\backslash (\mJ^\ast\cup\{k\})}|1,\bc,n)\\
    &f(\bc,n) d\bm d\bc dn,\\
    T_2=&\int n^{-1}\sum_{j=1}^n\int\int w(\bv,n)\eta_{\buj}(1, \bm, \bc,n)\nabep f_\epsilon(\bm^{\mJ^\ast},\bm^k_{\bullet,-j}|0,\bc,n)f(m_{\buj}^k,\bm^{\mK\backslash (\mJ^\ast\cup\{k\})}|1,\bc,n)\\
    &f(\bc,n) d\bm d\bc dn,\\
    T_3=&\int n^{-1}\sum_{j=1}^n\int\int w(\bv,n)\eta_{\buj}(1, \bm, \bc,n)f(\bm^{\mJ^\ast},\bm^k_{\bullet,-j}|0,\bc,n)\nabep f_\epsilon(m_{\buj}^k,\bm^{\mK\backslash (\mJ^\ast\cup\{k\})}|1,\bc,n)\\
    &f(\bc,n) d\bm d\bc dn,\\
    T_4=&\int n^{-1}\sum_{j=1}^n\int\int w(\bv,n)\eta_{\buj}(1, \bm, \bc,n)f(\bm^{\mJ^\ast},\bm^k_{\bullet,-j}|0,\bc,n) f(m_{\buj}^k,\bm^{\mK\backslash (\mJ^\ast\cup\{k\})}|1,\bc,n)\\
    &\nabep f_\epsilon(\bc,n) d\bm d\bc dn.
\end{align*}
For the term $\text{T}_1$, we have
\begin{align*}
    \text{T}_1=&\int n^{-1}\sum_{j=1}^n\int\int w(\bv,n)\nabep\eta_{\buj,_\epsilon}(1, \bm, \bc,n)f(\bm^{\mJ^\ast},\bm^k_{\bullet,-j}|0,\bc,n)\\
    &f(m_{\buj}^k,\bm^{\mK\backslash (\mJ^\ast\cup\{k\})}|1,\bc,n)f(\bc,n) d\bm d\bc dn\\
    =&\int n^{-1}\sum_{j=1}^n\int\int w(\bv,n)\int y_{\buj}\nabep f_\epsilon(\by|\bm,1,\bc,n)f(\bm^{\mJ^\ast},\bm^k_{\bullet,-j}|0,\bc,n)\\
     &f(m_{\buj}^k,\bm^{\mK\backslash (\mJ^\ast\cup\{k\})}|1,\bc,n)f(\bc,n) d\by d\bm d\bc dn\\
    =&\int n^{-1}\sum_{j=1}^n\int\int w(\bv,n)\int y_{\buj}S(\by|\bm,1,\bc,n)f(\by|\bm,1,\bc,n)f(\bm^{\mJ^\ast},\bm^k_{\bullet,-j}|0,\bc,n)\\
    &f(m_{\buj}^k,\bm^{\mK\backslash (\mJ^\ast\cup\{k\})}|1,\bc,n)f(\bc,n) d\by d\bm d\bc dn\\       
    =&\int n^{-1}\sum_{j=1}^n\int\int w(\bv,n)\int y_{\buj}\lb S(\by,\bm,1,\bc,n)-E\lb S(\bY,\bm,1,\bc,n)|\bm,1,\bc,n\rb\rb\\
    &f(\by|\bm,1,\bc,n)f(\bm^{\mJ^\ast},\bm^k_{\bullet,-j}|0,\bc,n)f(m_{\buj}^k,\bm^{\mK\backslash (\mJ^\ast\cup\{k\})}|1,\bc,n)f(\bc,n) d\by d\bm d\bc dn~~\text{(Lemma \ref{lemma:score})}\\
    =&\int n^{-1}\sum_{j=1}^n\int\int w(\bv,n)\int \lb y_{\buj}-\eta_{\buj}(1,\bm,\bc,n)\rb S(\by,\bm,1,\bc,n)\\
    &\frac{f(\bm^{\mJ^\ast},\bm^k_{\bullet,-j}|0,\bc,n)f(m_{\buj}^k,\bm^{\mK\backslash (\mJ^\ast\cup\{k\})}|1,\bc,n)}{f(\bm|1,\bc,n)}f(\by|\bm,1,\bc,n)f(\bm|1,\bc,n)f(\bc,n) d\by d\bm d\bc dn\\
    =&E\lb  \frac{w(\bV,N)}{N}\sum_{j=1}^N\frac{\mathcal{I}(A=1)}{\Pr(A=1)}\lb Y_{\buj}-\eta_{\buj}(1,\bM,\bC,N)\rb 
    r_{\buj}^{k\mJ^\ast011}(\bM,\bC,N)S(\mO)\rb.
\end{align*}
For the term $\text{T}_2$, we have
\begin{align*}
    \text{T}_2=&\int n^{-1}\sum_{j=1}^n\int\int w(\bv,n)\eta_{\buj}(1, \bm, \bc,n)\nabep f_\epsilon(\bm^{\mJ^\ast},\bm^k_{\bullet,-j}|0,\bc,n)\\
    &f(m_{\buj}^k,\bm^{\mK\backslash (\mJ^\ast\cup\{k\})}|1,\bc,n)f(\bc,n) d\bm d\bc dn\\
    =&\int n^{-1}\sum_{j=1}^n\int\int w(\bv,n)\eta_{\buj}(1,\bm,\bc,n)S(\bm^{\mJ^\ast},\bm^k_{\bullet,-j}|0,\bc,n)f(\bm^{\mJ^\ast},\bm^k_{\bullet,-j}|0,\bc,n)\\
    &f(m_{\buj}^k,\bm^{\mK\backslash (\mJ^\ast\cup\{k\})}|1,\bc,n)f(\bc,n) d\bm d\bc dn\\
    =&\int n^{-1}\sum_{j=1}^n\int\int w(\bv,n)\eta_{\buj}(1,\bm,\bc,n) \times \\
    &\lb E\lb S(\bY,M^k_{\buj},\bM^{\mK\backslash (\mJ^\ast\cup\{k\})},\bm^{\mJ^\ast},\bm_{\bumj}^k,0,\bc,n)|\bm^{\mJ^\ast},\bm_{\bumj}^k,0,
    \bc,n\rb-E\lb S(\bY,\bM,0,\bc,n)|0,\bc,n\rb\rb  \\
    &f(\bm^{\mJ^\ast},\bm^k_{\bullet,-j}|0,\bc,n)f(m_{\buj}^k,\bm^{\mK\backslash (\mJ^\ast\cup\{k\})}|1,\bc,n)f(\bc,n) d\bm d\bc dn~~\text{(Lemma \ref{lemma:score})}\\  
    =&\int n^{-1}\sum_{j=1}^n\int w(\bv,n) \int \int \lb\widetilde{\kappa}_{\buj}^{k\mJ^\ast}(\bm^k_{\bumj},\bm^{\mJ^\ast},\bc,n)- \kappa^{k\mJ^\ast}_{\buj}(\bc,n)\rb S(\by,\bm,0,\bc,n) f(\by|\bm,0,\bc,n)\\
    &f(\bm|0,\bc,n)f(\bc,n) d\by d\bm d\bc dn\\
    =&E\lb  \frac{w(\bV,N)}{N}\sum_{j=1}^N\frac{\mathcal{I}(A=0)}{\Pr(A=0)}\lb\widetilde{\kappa}_{\buj}^{k\mJ^\ast}(\bM^k_{\bumj},\bM^{\mJ^\ast},\bC,N)- \kappa^{k\mJ^\ast}_{\buj}(\bC,N)\rb S(\mO)\rb.
\end{align*}
where 
\begin{align*}
&\kappa^{k\mJ^\ast}_{\buj}(\bc,n)=\int\eta_{\buj}(1,\bm,\bc,n)f(\bm^{\mJ^\ast},\bm^k_{\bullet,-j}|0,\bc,n)f(m_{\buj}^k,\bm^{\mK\backslash (\mJ^\ast\cup\{k\})}|1,\bc,n)d\bm,\\
    &\widetilde{\kappa}_{\buj}^{k\mJ^\ast}(\bm^{\mJ^\ast},\bm_{\bumj}^k,\bc,n)= \int\eta_{\buj}(1,\bm,\bc,n)f(m_{\buj}^k,\bm^{\mK\backslash (\mJ^\ast\cup\{k\})}|1,\bc,n)d(m^k_{\buj},\bm^{\mK\backslash (\mJ^\ast\cup\{k\})}).
\end{align*}
For the term $\text{T}_3$, we have
\begin{align*}
    \text{T}_3=&\int n^{-1}\sum_{j=1}^n\int\int w(\bv,n)\eta_{\buj}(1, \bm, \bc,n)f(\bm^{\mJ^\ast},\bm^k_{\bullet,-j}|0,\bc,n)\\
    &\nabep f_\epsilon(m_{\buj}^k,\bm^{\mK\backslash (\mJ^\ast\cup\{k\})}|1,\bc,n)f(\bc,n) d\bm d\bc dn\\
    =&\int n^{-1}\sum_{j=1}^n\int\int w(\bv,n)\eta_{\buj}(1,\bm,\bc,n)S(m_{\buj}^k,\bm^{\mK\backslash (\mJ^\ast\cup\{k\})}|1,\bc,n)f(\bm^{\mJ^\ast},\bm^k_{\bullet,-j}|0,\bc,n)\\
    &f(m_{\buj}^k,\bm^{\mK\backslash (\mJ^\ast\cup\{k\})}|1,\bc,n)f(\bc,n) d\bm d\bc dn\\
    =&\int n^{-1}\sum_{j=1}^n\int\int w(\bv,n)\eta_{\buj}(1,\bm,\bc,n) \times \\
    &\Big[ E\lb S(\bY,\bM^{\mJ^\ast},\bM_{\bumj}^k,m^k_{\buj},\bm^{\mK\backslash (\mJ^\ast\cup\{k\})},1,\bc,n)|m^k_{\buj},\bm^{\mK\backslash (\mJ^\ast\cup\{k\})},1,
    \bc,n\rb-\\
    &E\lb S(\bY,\bM,1,\bc,n)|1,\bc,n\rb\Big]  \\
    &f(\bm^{\mJ^\ast},\bm^k_{\bullet,-j}|0,\bc,n)f(m_{\buj}^k,\bm^{\mK\backslash (\mJ^\ast\cup\{k\})}|1,\bc,n)f(\bc,n) d\bm d\bc dn~~\text{(Lemma \ref{lemma:score})}\\  
    =&\int n^{-1}\sum_{j=1}^n\int w(\bv,n) \int \int \lb\check{\kappa}_{\buj}^{k\mJ^\ast}(m_{\buj}^k,\bm^{\mK\backslash (\mJ^\ast\cup\{k\})},\bc,n)- \kappa^{k\mJ^\ast}_{\buj}(\bc,n)\rb S(\by,\bm,1,\bc,n) \\
    &f(\by|\bm,1,\bc,n)f(\bm|1,\bc,n)f(\bc,n) d\by d\bm d\bc dn\\
    =&E\lb  \frac{w(\bV,N)}{N}\sum_{j=1}^N\frac{\mathcal{I}(A=1)}{\Pr(A=1)}\lb\check{\kappa}_{\buj}^{k\mJ^\ast}(M_{\buj}^k,\bM^{-k},\bC,N)- \kappa^{k\mJ^\ast}_{\buj}(\bC,N)\rb S(\mO)\rb,
\end{align*}
where $\check{\kappa}_{\buj}^{k\mJ^\ast}(m_{\buj}^k,\bm^{\mK\backslash (\mJ^\ast\cup\{k\})},\bc,n)= \int\eta_{\buj}(1,\bm,\bc,n)f(\bm^{\mJ^\ast},\bm^k_{\bumj}|0,\bc,n)d(\bm^{\mJ^\ast},\bm^k_{\bumj})$.
For the term $\text{T}_4$, we have
\begin{align*}
    \text{T}_4=&\int n^{-1}\sum_{j=1}^n\int\int w(\bv,n)\eta_{\buj}(1, \bm, \bc,n)f(\bm^{\mJ^\ast},\bm^k_{\bullet,-j}|0,\bc,n) \\
    &f(m_{\buj}^k,\bm^{\mK\backslash (\mJ^\ast\cup\{k\})}|1,\bc,n)\nabep f_\epsilon(\bc,n) d\bm d\bc dn\\
    =&\int n^{-1}\sum_{j=1}^n\int\int w(\bv,n)\eta_{\buj}(1,\bm,\bc,n)f(\bm^{\mJ^\ast},\bm^k_{\bullet,-j}|0,\bc,n) f(m_{\buj}^k,\bm^{\mK\backslash (\mJ^\ast\cup\{k\})}|1,\bc,n)\\
    &S(\bc,n)f(\bc,n) d\bm d\bc dn\\
    =&\int n^{-1}\sum_{j=1}^n\int\int w(\bv,n)\eta_{\buj}(1,\bm,\bc,n) f(\bm^{\mJ^\ast},\bm^k_{\bullet,-j}|0,\bc,n) f(m_{\buj}^k,\bm^{\mK\backslash (\mJ^\ast\cup\{k\})}|1,\bc,n)  \\
    &\left[ E\lb S(\bY,\bM,A,\bc,n)|\bc,n\rb -E\lb S(\mO)\rb\right] f(\bc,n) d\bm d\bc dn~~\text{(Lemma \ref{lemma:score})}\\
    =&\int n^{-1}\sum_{j=1}^n\int w(\bv,n)\kappa^{k\mJ^\ast}_{\buj}(\bc,n) E\lb S(\bY,\bM,A,\bc,n)|\bc,n\rb f(\bc,n) d\bc dn-E\lb S(\mO)\rb\theta_n(k,\mJ^\ast)\\
    =&\int n^{-1}\sum_{j=1}^n\int \sum_{a=0}^1w(\bv,n) \kappa^{k\mJ^\ast}_{\buj}(\bc,n)  S(\by,\bm,a,\bc,n)\\
    &f(\by|\bm,a,\bc,n)f(\bm|a,\bc,n)\Pr(A=a)f(\bc,n)d\by d\bm d\bc dn-E\lb S(\mO)\rb\theta_n(k,\mJ^\ast)\\
    =&E\lb \left[\frac{w(\bV,N)}{N}\sum_{j=1}^N  \kappa^{k\mJ^\ast}_{\buj}(\bC,N)-\theta_n(k,\mJ^\ast)\right]S(\mO) \rb.
\end{align*}
It is straightforward to verify that
$\varphi(\mO;\theta_n(k,\mJ^\ast))\in\mathcal{H}$, where
\begin{align*}
    &\varphi(\mO;\theta_n(k,\mJ^\ast))=\phi(\mO;\theta_2(k,\mJ^\ast))-\theta_n(k,\mJ^\ast),
\end{align*}
and
   \begin{align*}     
&\phi(\mO;\theta_2(k,\mJ^\ast))\\
=&\frac{w(\bV,N)}{N}\sum_{j=1}^N\Bigg\{\frac{\mathcal{I}(A=0)}{\Pr(A=0)}\widetilde{\kappa}_{\buj}^{k\mJ^\ast}(\bM_{\bumj}^k,\bM^{\mJ^\ast},\bC,N) +\left(1-\sum_{a=0}^1\frac{\mathcal{I}(A=a)}{\Pr(A=a)}\right)\kappa^{k\mJ^\ast}_{\buj}(\bC,N)+\\
    &\frac{\mathcal{I}(A=1)}{\Pr(A=1)}\left[r^{k\mJ^\ast011}_{\buj}(\bM,\bC,N)\left( Y_{\buj}-\eta_{\buj}(1,\bM,\bC,N)\right)+\check{\kappa}_{\buj}^{k\mJ^\ast}(M_{\buj}^k,\bM^{\mK\backslash(\mJ^\ast\cup\{k\})},\bC,N)\right]\Bigg\}.
\end{align*}
Therefore, we conclude that $\varphi(\mO;\theta_n(k,\mJ^\ast))$ is the EIF for $\theta_{n}(k,\mJ^\ast)$. Eventually, by Lemma \ref{lemma:quotient}, we conclude that 
\begin{align*}
    \varphi(\mO;\theta_2(k,\mJ^\ast))=\displaystyle\frac{\phi(\mO;\theta_2(k,\mJ^\ast))-\theta_2(k,\mJ^\ast)w(\bV,N)}{E\lb w(\bV,N)\rb}.
\end{align*}


\subsection{Proof of Theorem \ref{thm:dml}}

We begin by presenting the following decomposition of the expected estimation error of $\phi(\mO;\theta,\widehat{\bGamma}(\theta))$ given the training sample. For simplicity, all expectations are taken conditional on $\bmO^{-s}$, although this dependence is suppressed in the notation.
\begin{lemma}\label{lemma:decomposition-Eerror}
The expected estimation error of $\phi(\mO;\theta,\widehat{\bGamma}(\theta))$ can be decomposed as: (i) for $\theta=\theta_1^{I}(a^\ast,\ba_{\mJ})$,
\begin{align*}
       &E\lb \phi(\mO;\theta,\widehat{\bGamma}(\theta))-\phi(\mO;\theta,\bGamma(\theta))\rb\\
          =&E\bigg\{ w(\bV,N)\bigg{[}\int \frac{\widehat{f}(\ba_{\mJ1},\bm,\bC,N)}{\widehat{f}(a^\ast,\bm,\bC,N)}[\widehat{f}(a^\ast,\bm,\bC,N)-f(a^\ast,\bm,\bC,N)][\widehat{\overline{\eta}}(a^\ast,\bm,\bC,N)-\overline{\eta}(a^\ast,\bm,\bC,N)]-\\
   &[\widehat{\overline{\eta}}(a^\ast,\bm,\bC,N)-\overline{\eta}(a^\ast,\bm,\bC,N)][\widehat{f}(\ba_{\mJ1},\bm,\bC,N)-f(\ba_{\mJ1},\bm,\bC,N)]d\bm\bigg{]}\bigg\};
\end{align*}
(ii) for $\theta=\theta_1^{II}(1,\ba_{\mJ^\ast})$,
\begin{align*}
   &E\lb \phi(\mO;\theta,\widehat{\bGamma}(\theta))-\phi(\mO;\theta,\bGamma(\theta))\rb\\
   =&E\bigg\{ \frac{w(\bV,N)}{N}\sum_{j=1}^N\bigg{[}\int \widehat{\widetilde{r}}^{\mJ^\ast 011}_{\buj}(\bm,\bC,N)[\widehat{f}(\bm|1,\bC,N)-f(\bm|1,\bC,N)]\times\\
   &[\widehat{\eta}_{\buj}(1,\bm,\bC,N)-\eta_{\buj}(1,\bm,\bC,N)]-\widehat{f}(\bm^{\mJ^\ast}|0,\bC,N)[\widehat{\eta}_{\buj}(1,\bm,\bC,N)-\eta_{\buj}(1,\bm,\bC,N)]\times\\
   &[\widehat{f}(\bm^{\mK\backslash \mJ^\ast}|1,\bC,N)-f(\bm^{\mK\backslash \mJ^\ast}|1,\bC,N)]-\\
   &f(\bm^{\mK\backslash \mJ^\ast}|1,\bC,N)[\widehat{\eta}_{\buj}(1,\bm,\bC,N)-\eta_{\buj}(1,\bm,\bC,N)][\widehat{f}(\bm^{\mJ^\ast}|0,\bC,N)-f(\bm^{\mJ^\ast}|0,\bC,N)]-\\
   &\widehat{\eta}_{\buj}(1,\bm,\bC,N)[\widehat{f}(\bm^{\mJ^\ast}|0,\bC,N)-f(\bm^{\mJ^\ast}|0,\bC,N)][\widehat{f}(\bm^{\mK\backslash \mJ^\ast}|1,\bC,N)-f(\bm^{\mK\backslash \mJ^\ast}|1,\bC,N)]\bigg{]}d\bm\bigg\},
\end{align*}
where
\begin{align*}
    \widehat{\widetilde{r}}^{\mJ^\ast 011}_{\buj}(\bm,\bC,N)=\frac{\widehat{f}(\bm^{\mJ^\ast}|0,\bC,N)\widehat{f}(\bm^{\mK\backslash \mJ^\ast}|1,\bC,N)}{\widehat{f}(\bm|1,\bC,N)};
\end{align*}
and (iii) for $\theta=\theta_2(k,\mJ^\ast)$,
\begin{align*}
   &E\lb \phi(\mO;\theta,\widehat{\bGamma}(\theta))-\phi(\mO;\theta,\bGamma(\theta))\rb\\
   =&E\bigg\{ \frac{w(\bV,N)}{N}\sum_{j=1}^N\bigg{[}\int \widehat{r}^{k\mJ^\ast 011}_{\buj}(\bm,\bC,N)[\widehat{f}(\bm|1,\bC,N)-f(\bm|1,\bC,N)]\times\\
   &[\widehat{\eta}_{\buj}(1,\bm,\bC,N)-\eta_{\buj}(1,\bm,\bC,N)]-\widehat{f}(\bm^{\mJ^\ast},\bm^k_{\bullet,-j}|0,\bC,N)[\widehat{\eta}_{\buj}(1,\bm,\bC,N)-\eta_{\buj}(1,\bm,\bC,N)]\times\\
   &[\widehat{f}(m_{\buj}^k,\bm^{\mK\backslash (\mJ^\ast\cup\lb k\rb)}|1,\bC,N)-f(m_{\buj}^k,\bm^{\mK\backslash (\mJ^\ast\cup\lb k\rb)}|1,\bC,N)]-\\
   &f(m_{\buj}^k,\bm^{\mK\backslash (\mJ^\ast\cup\lb k\rb)}|1,\bC,N)[\widehat{\eta}_{\buj}(1,\bm,\bC,N)-\eta_{\buj}(1,\bm,\bC,N)]\times\\
   &[\widehat{f}(\bm^{\mJ^\ast},\bm^k_{\bullet,-j}|0,\bC,N)-f(\bm^{\mJ^\ast},\bm^k_{\bullet,-j}|0,\bC,N)]-\\
   &\widehat{\eta}_{\buj}(1,\bm,\bC,N)[\widehat{f}(\bm^{\mJ^\ast},\bm^k_{\bullet,-j}|0,\bC,N)-f(\bm^{\mJ^\ast},\bm^k_{\bullet,-j}|0,\bC,N)]\times\\
   &[\widehat{f}(m_{\buj}^k,\bm^{\mK\backslash (\mJ^\ast\cup\lb k\rb)}|1,\bC,N)-f(m_{\buj}^k,\bm^{\mK\backslash (\mJ^\ast\cup\lb k\rb)}|1,\bC,N)]\bigg{]}d\bm\bigg\},
\end{align*}
where 
\begin{align*}
    \widehat{r}_{\buj}^{k\mJ^\ast011}(\bm,\bC,N)=&\frac{\widehat{f}(\bm^{\mJ^\ast},\bm^k_{\bullet,-j}|0,\bC,N)\widehat{f}(m_{\buj}^k,\bm^{\mK\backslash (\mJ^\ast\cup\lb k\rb)}|1,\bC,N)}{\widehat{f}(\bm|1,\bC,N)}.
\end{align*}
    
\end{lemma}
\begin{proof}
It follows from the proof of Lemma 1 in the Supplementary Material of \cite{cheng2024semiparametric}, by matching the corresponding mediator index.
\end{proof}
Define $\widehat{\theta}_{\text{nu}} = I^{-1} \sum_{s=1}^S I_s \widehat{\phi}^s(\mO; \theta)$ and $\theta_{\text{nu}} = \theta \times E\{w(\bV, N)\}$, where $\widehat{\theta}_{\text{nu}}$ is the debiased machine learning estimator targeting $\theta_{\text{nu}}$. We can decompose the estimation error of $\widehat{\theta}_{\text{nu}}$:
\begin{equation*}
    \widehat{\theta}_{\text{nu}}-\theta_{\text{nu}}=T_1+T_2+T_3,
\end{equation*}
where 
\begin{align*}
    &T_1=(\mathbb{P}_I-E)\{\phi(\mO;\theta)\},\\
    &T_2=\sum_{s=1}^S\frac{I_s}{I}(\mathbb{P}_{s}-E)\{\phi(\mO;\theta,\widehat{\Gamma}^s(\theta))-\phi(\mO;\theta,\Gamma(\theta))\},\\
    &T_3=\sum_{s=1}^S\frac{I_s}{I}E\{\phi(\mO;\theta,\widehat{\Gamma}^s(\theta))-\phi(\mO;\theta,{\Gamma}(\theta))\}.
\end{align*} 
We will show that $T_2 = o_\mP(I^{-1/2})$ and $T_3 = \text{Rem}_\theta$. Briefly, $T_3 = \text{Rem}_\theta$ follows from Lemma \ref{lemma:decomposition-Eerror}. In what follows, constants are denoted by the same symbol $C$, although their values may vary across different contexts. To see this, for $\theta=\theta_1^{I}(a^\ast,\ba_{\mJ})$,
\begin{align*}
    &|E\{\phi(\mO;\theta,\widehat{\Gamma}^s(\theta))-\phi(\mO;\theta,{\Gamma}(\theta))\}|\\
    \leq&CE\bigg\{\bigg| w(\bV,N)\bigg{[}\frac{\widehat{f}(\ba_{\mJ1},\bM,\bC,N)}{\widehat{f}(a^\ast,\bM,\bC,N)}[\widehat{f}(a^\ast,\bM,\bC,N)-f(a^\ast,\bM,\bC,N)]\\
    &[\widehat{\overline{\eta}}(a^\ast,\bM,\bC,N)-\overline{\eta}(a^\ast,\bM,\bC,N)]-\\
   &[\widehat{\overline{\eta}}(a^\ast,\bM,\bC,N)-\overline{\eta}(a^\ast,\bM,\bC,N)][\widehat{f}(\ba_{\mJ1},\bM,\bC,N)-f(\ba_{\mJ1},\bM,\bC,N)]\bigg{]}\bigg|\bigg\}~~\text{(LOTE and RC (ii))}\\
   \leq&C E\bigg\{ \bigg| [\widehat{f}(a^\ast,\bM,\bC,N)-f(a^\ast,\bM,\bC,N)][\widehat{\overline{\eta}}(a^\ast,\bM,\bC,N)-\overline{\eta}(a^\ast,\bM,\bC,N)]\bigg|\bigg\}+\\
   &CE\bigg\{\bigg|[\widehat{f}(\ba_{\mJ1},\bM,\bC,N)-f(\ba_{\mJ1},\bM,\bC,N)][\widehat{\overline{\eta}}(a^\ast,\bM,\bC,N)-\overline{\eta}(a^\ast,\bM,\bC,N)]\bigg|\bigg\}~~\text{(RCs (i)\&(ii))}\\
   \leq&O_\mP(\lVert \widehat{\overline{\eta}}-\overline{\eta}\rVert_2(\lVert \widehat{f}^{a^\ast}-f^{a^\ast}\rVert_2+\lVert \widehat{f}^{\ba_{\mJ1}}-f^{\ba_{\mJ1}}\rVert_2)),
\end{align*}
where the last inequality follows from the Cauchy-Schwarz inequality. For $\theta=\theta_2(k,\mJ^\ast)$, we have that
\begin{align*}
    &|E\{\phi(\mO;\theta,\widehat{\Gamma}^s(\theta))-\phi(\mO;\theta,{\Gamma}(\theta))\}|\\
    \leq&CE\bigg\{ \bigg|\frac{1}{N}\sum_{j=1}^N\bigg{[} \widehat{r}^{k\mJ^\ast011}_{\buj}(\bM,\bC,N)[\widehat{f}(\bM|1,\bC,N)-f(\bM|1,\bC,N)]\times\\
   &[\widehat{\eta}_{\buj}(1,\bM,\bC,N)-\eta_{\buj}(1,\bM,\bC,N)]-\widehat{f}(\bM^{\mJ^\ast},\bM^k_{\bullet,-j}|0,\bC,N)[\widehat{\eta}_{\buj}(1,\bM,\bC,N)-\eta_{\buj}(1,\bM,\bC,N)]\times\\
   &[\widehat{f}(m_{\buj}^k,\bM^{\mK\backslash (\mJ^\ast\cup\lb k\rb)}|1,\bC,N)-f(m_{\buj}^k,\bM^{\mK\backslash (\mJ^\ast\cup\lb k\rb)}|1,\bC,N)]-\\
   &f(m_{\buj}^k,\bM^{\mK\backslash (\mJ^\ast\cup\lb k\rb)}|1,\bC,N)[\widehat{\eta}_{\buj}(1,\bM,\bC,N)-\eta_{\buj}(1,\bM,\bC,N)]\times\\
   &[\widehat{f}(\bM^{\mJ^\ast},\bM^k_{\bullet,-j}|0,\bC,N)-f(\bM^{\mJ^\ast},\bM^k_{\bullet,-j}|0,\bC,N)]-\\
   &\widehat{\eta}_{\buj}(1,\bM,\bC,N)[\widehat{f}(\bM^{\mJ^\ast},\bM^k_{\bullet,-j}|0,\bC,N)-f(\bM^{\mJ^\ast},\bM^k_{\bullet,-j}|0,\bC,N)]\times\\
   &[\widehat{f}(m_{\buj}^k,\bM^{\mK\backslash (\mJ^\ast\cup\lb k\rb)}|1,\bC,N)-f(m_{\buj}^k,\bM^{\mK\backslash (\mJ^\ast\cup\lb k\rb)}|1,\bC,N)]\bigg{]}\bigg|\bigg\}~~\text{(LOTE and RCs (i)\&(ii))}\\
   \leq&C E\bigg\{ \bigg|\widehat{f}(\bM|1,\bC,N)-f(\bM|1,\bC,N)\bigg|\max_{1\leq j\leq N}\bigg|\widehat{\eta}_{\buj}(1,\bM,\bC,N)-\eta_{\buj}(1,\bM,\bC,N)\bigg|\bigg\}+\\
   &CE\bigg\{\max_{1\leq j\leq N}\bigg|\widehat{f}(m_{\buj}^k,\bM^{\mK\backslash (\mJ^\ast\cup\lb k\rb)}|1,\bC,N)-f(m_{\buj}^k,\bM^{\mK\backslash (\mJ^\ast\cup\lb k\rb)}|1,\bC,N)\bigg|\times\\
   &\max_{1\leq j\leq N}\bigg|\widehat{\eta}_{\buj}(1,\bM,\bC,N)-\eta_{\buj}(1,\bM,\bC,N)\bigg|\bigg\}\\
   &+CE\bigg\{\max_{1\leq j\leq N}\bigg|\widehat{f}(\bM^{\mJ^\ast},\bM^k_{\bullet,-j}|0,\bC,N)-f(\bM^{\mJ^\ast},\bM^k_{\bullet,-j}|0,\bC,N)\bigg|\times\\
   &\max_{1\leq j\leq N}\bigg|\widehat{\eta}_{\buj}(1,\bM,\bC,N)-\eta_{\buj}(1,\bM,\bC,N)\bigg|\bigg\}\\
   &CE\bigg\{\max_{1\leq j\leq N}\bigg|\widehat{f}(\bM^{\mJ^\ast},\bM^k_{\bullet,-j}|0,\bC,N)-f(\bM^{\mJ^\ast},\bM^k_{\bullet,-j}|0,\bC,N)\bigg|\times\\
   &\max_{1\leq j\leq N}\bigg|\widehat{f}(m_{\buj}^k,\bM^{\mK\backslash (\mJ^\ast\cup\lb k\rb)}|1,\bC,N)-f(m_{\buj}^k,\bM^{\mK\backslash (\mJ^\ast\cup\lb k\rb)}|1,\bC,N)\bigg|\bigg\}~~\text{(RC (ii))}\\
   \leq&O_{\mP}(\lVert \max_{1\leq j\leq N}|\widehat{\eta}_{\buj}^1-\eta_{\buj}^1|\rVert_2(\lVert \widehat{f}^1-f^1\rVert_2+\lVert \max_{1\leq j\leq N}|\widehat{f}_{\buj}^1-f_{\buj}^1|\rVert_2+\lVert \max_{1\leq j\leq N}|\widehat{f}_{\bumj}^0-f_{\bumj}^0|\rVert_2)+\\
   &\lVert \max_{1\leq j\leq N}|\widehat{f}_{\buj}^1-f_{\buj}^1|\rVert_2\lVert \max_{1\leq j\leq N}| \widehat{f}_{\bumj}^0-f_{\bumj}^0|\rVert_2)~~\text{(Cauchy-Schwarz)}.
\end{align*}
For $\theta=\theta^{II}_1(1,\ba_{\mJ^\ast})$, similar arguments based on Lemma \ref{lemma:decomposition-Eerror} show that 
\begin{align*}
    &T_3
    \\=&O_{\mP}\left(d_2(\widehat{\eta}_{\buj}^1,{\eta}^1_{\buj})\left\{\lVert \widehat{f}^1-f^1\rVert_2+\lVert\widehat{f}_{\mK\backslash\mJ^\ast}^1-f_{\mK\backslash\mJ^\ast}^1\rVert_2+\lVert\widehat{f}_{\mJ^\ast}^0-f_{\mJ^\ast}^0\rVert_2\right\}+
   \lVert\widehat{f}_{\mK\backslash\mJ^\ast}^1-f_{\mK\backslash\mJ^\ast}^1\rVert_2 \lVert\widehat{f}_{\mJ^\ast}^0-f_{\mJ^\ast}^0\rVert_2\right).
\end{align*}

Finally, we analyze the empirical process term $T_2$. By Lemma 4 in the Supplementary Material of \cite{tong2025semiparametric}, it suffices to show that
\begin{align*}
    E\lb \left(\phi(\mO;\theta,\widehat{\Gamma}^s(\theta))-\phi(\mO;\theta,\Gamma(\theta))\right)^2|\mO^{-s}\rb=o_{\mP}(1).
\end{align*}
The verification for $\theta=\theta^{I}(a^\ast,\ba_{\mJ})$ follows arguments similar to those in \cite{cheng2024semiparametric}. We present the verification for $\theta=\theta_2(k,\mJ^\ast)$; similar arguments apply to $\theta=\theta^{II}(1,\ba_{\mJ^\ast})$. For $\theta=\theta_2(k,\mJ^\ast)$, we have that
\begin{align*}
    &E\lb \left(\phi(\mO;\theta,\widehat{\Gamma}^s(\theta))-\phi(\mO;\theta,\Gamma(\theta))\right)^2|\mO^{-s}\rb\\
    =&E\left\{ \left[\frac{w(\bV,N)}{N}\sum_{j=1}^N\left\{\frac{\mathcal{I}(A=0)}{\Pr(A=0)}(\widehat{\widetilde{\kappa}}^{k\mJ^\ast}_{\buj}-\widetilde{\kappa}_{\buj}^{k\mJ^\ast}) +\left(1-\sum_{a=0}^1\frac{\mathcal{I}(A=a)}{\Pr(A=a)}\right)(\widehat{\kappa}^{k\mJ^\ast}_{\buj}-\kappa^{k\mJ^\ast}_{\buj})+\right.\right.\right.\\
    &\left.\left.\left.\frac{\mathcal{I}(A=1)}{\Pr(A=1)}\left[(\widehat{r}^{k\mJ^\ast 011}_{\buj}-r^{k\mJ^\ast 011}_{\buj})Y_{\buj}-\left(\widehat{r}^{k\mJ^\ast 011}_{\buj}\widehat{\eta}_{\buj} -r^{k\mJ^\ast 011}_{\buj}\eta_{\buj}\right)+(\widehat{\check{\kappa}}^{k\mJ^\ast}_{\buj}-\check{\kappa}_{\buj}^{k\mJ^\ast})\right]\right\}\right]^2|\mO^{-s}\right\}\\
    \leq&CE\left\{ \frac{1}{N}\left[\sum_{j=1}^N\left\{\frac{\mathcal{I}(A=0)}{\Pr(A=0)}(\widehat{\widetilde{\kappa}}^{k\mJ^\ast}_{\buj}-\widetilde{\kappa}_{\buj}^{k\mJ^\ast}) +\left(1-\sum_{a=0}^1\frac{\mathcal{I}(A=a)}{\Pr(A=a)}\right)(\widehat{\kappa}^{k\mJ^\ast}_{\buj}-\kappa^{k\mJ^\ast}_{\buj})+\right.\right.\right.\\
    &\left.\left.\left.\frac{\mathcal{I}(A=1)}{\Pr(A=1)}\left[(\widehat{r}^{k\mJ^\ast011}_{\buj}-r^{k\mJ^\ast011}_{\buj})Y_{\buj}-\left(\widehat{r}^{k\mJ^\ast011}_{\buj}\widehat{\eta}_{\buj} -r^{k\mJ^\ast011}_{\buj}\eta_{\buj}\right)+(\widehat{\check{\kappa}}_{\buj}^{k\mJ^\ast}-\check{\kappa}_{\buj}^{k\mJ^\ast})\right]\right\}\right]^2|\mO^{-s}\right\}\\
    &\text{(RC (i))}\\
    \leq&CE\left\{ \sum_{j=1}^N\left\{\frac{\mathcal{I}(A=0)}{\Pr(A=0)}(\widehat{\widetilde{\kappa}}^{k\mJ^\ast}_{\buj}-\widetilde{\kappa}_{\buj}^{k\mJ^\ast}-\widehat{\kappa}^{k\mJ^\ast}_{\buj}+\kappa^{k\mJ^\ast}_{\buj}) +(\widehat{\kappa}^{k\mJ^\ast}_{\buj}-\kappa^{k\mJ^\ast}_{\buj})+\frac{\mathcal{I}(A=1)}{\Pr(A=1)}\times\right.\right.\\
    &\left.\left.\left[(\widehat{r}^{k\mJ^\ast011}_{\buj}-r^{k\mJ^\ast011}_{\buj})(Y_{\buj}-\eta_{\buj})-\widehat{r}^{k\mJ^\ast011}_{\buj}\left(\widehat{\eta}_{\buj} -\eta_{\buj}\right)+(\widehat{\check{\kappa}}_{\buj}^{k\mJ^\ast}-\check{\kappa}_{\buj}^{k\mJ^\ast})-(\widehat{\kappa}^{k\mJ^\ast}_{\buj}-\kappa^{k\mJ^\ast}_{\buj})\right]\right\}^2|\mO^{-s}\right\}\\
    &\text{(QM-AM inequality)}\\
    \leq&CE\left\{ \sum_{j=1}^N\left\{C(\widehat{\widetilde{\kappa}}^{k\mJ^\ast}_{\buj}-\widetilde{\kappa}_{\buj}^{k\mJ^\ast})^2 +C(\widehat{\kappa}^{k\mJ^\ast}_{\buj}-\kappa^{k\mJ^\ast}_{\buj})^2+(\widehat{r}^{k\mJ^\ast011}_{\buj}-r^{k\mJ^\ast011}_{\buj})^2E\{(Y_{\buj}-\eta_{\buj})^2|1,\bM,\bC,N\}+\right.\right.\\
    &\left.\left.C\left[(\widehat{r}_{\buj}^{k\mJ^\ast011})^2\left(\widehat{\eta}_{\buj} -\eta_{\buj}\right)^2+(\widehat{\check{\kappa}}_{\buj}^{k\mJ^\ast}-\check{\kappa}_{\buj}^{k\mJ^\ast})^2\right]\right\}|\mO^{-s}\right\}~~~\text{(QM-AM inequality and LOTE)}\\
    \leq&CE\left\{ \sum_{j=1}^N\left\{C(\widehat{\widetilde{\kappa}}^{k\mJ^\ast}_{\buj}-\widetilde{\kappa}_{\buj}^{k\mJ^\ast})^2 +C(\widehat{\kappa}^{k\mJ^\ast}_{\buj}-\kappa^{k\mJ^\ast}_{\buj})^2+C(\widehat{r}^{k\mJ^\ast011}_{\buj}-r^{k\mJ^\ast011}_{\buj})^2+C(\widehat{\check{\kappa}}_{\buj}^{k\mJ^\ast}-\check{\kappa}_{\buj}^{k\mJ^\ast})^2\right\}|\mO^{-s}\right\}\\
    &+o_\mP(1)~~~\text{(RCs (ii)-(iv))}\\
    \leq&o_\mP(1),
\end{align*}
where the QM-AM inequality is $(x+y)^2\leq2(x^2+y^2)$, the last inequality follows from the LOTE, and the fact implied by RC (ii) and RC (iv) that $E\{(\widehat{\widetilde{\kappa}}^{k\mJ^\ast}_{\buj}-\widetilde{\kappa}_{\buj}^{k\mJ^\ast})^2|N\}=o_\mP(1)$, $E\{(\widehat{\kappa}^{k\mJ^\ast}_{\buj}-\kappa^{k\mJ^\ast}_{\buj})^2|N\}=o_\mP(1)$, $E\{(\widehat{r}^{k\mJ^\ast011}_{\buj}-r^{k\mJ^\ast011}_{\buj})^2|N\}=o_\mP(1)$, and $E\{(\widehat{\check{\kappa}}_{\buj}^{k\mJ^\ast}-\check{\kappa}_{\buj}^{k\mJ^\ast})^2|N\}=o_\mP(1)$ for all $1\leq j\leq N$ wp1 (the sum of at most finite $N_{\max}$ terms of $o_\mP(1)$ is $o_\mP(1)$).

\section{Additional tables and figures}\label{sec:tables_figures}

\begin{figure}[ht!]
    \centering
    \includegraphics[scale=0.65]{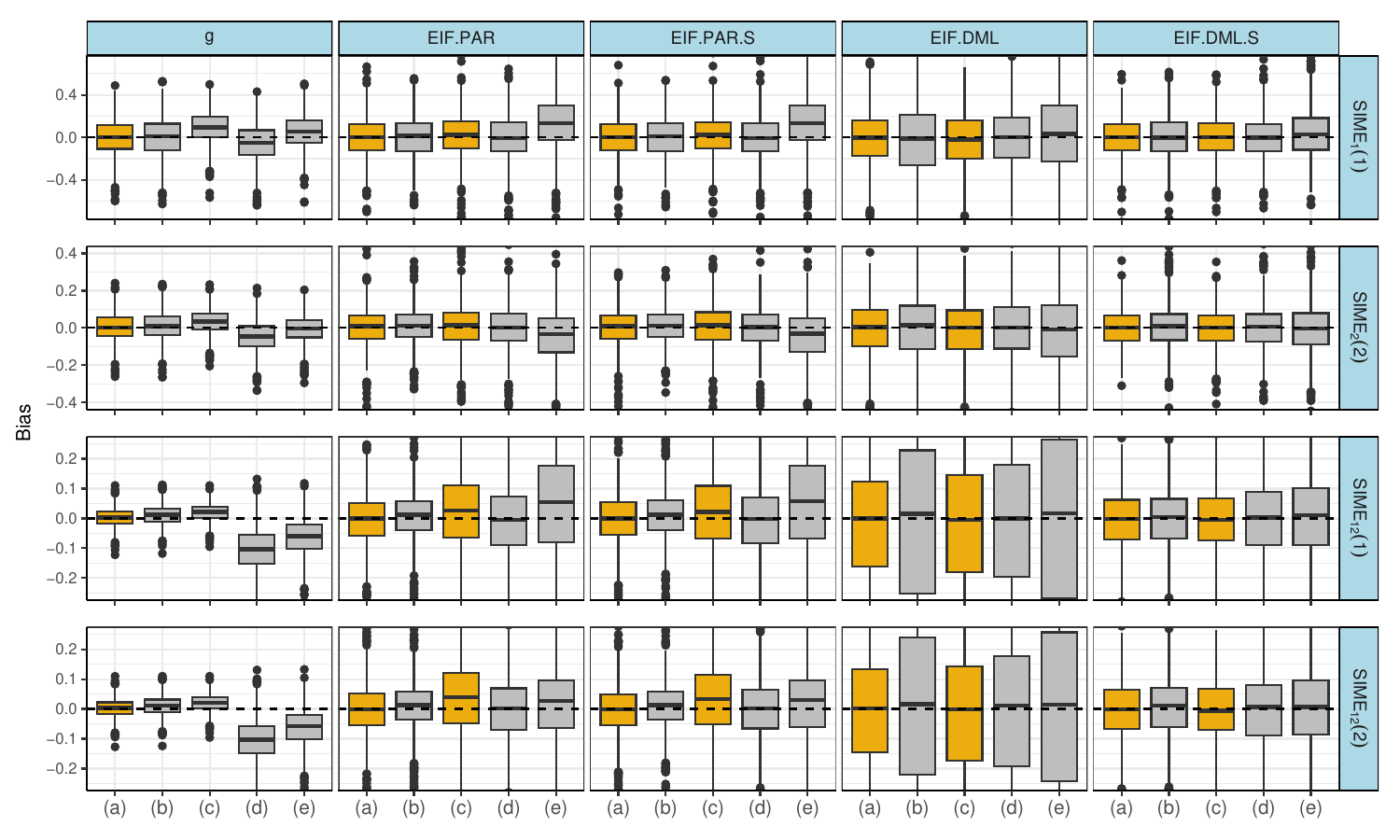}
    \caption{ Simulation results for the bias of the spillover interaction mediation effects $\text{SIME}_{\{1\}}(1),,\text{SIME}_{\{2\}}(2),\text{SIME}_{\{1,2\}}(1),\text{SIME}_{\{1,2\}}(2)$. We consider the following 5 scenarios: (a) assumes all components are correctly specified; (b) misspecifies the mediator conditional means $\{\underline{\eta}_{\buj}^1,\underline{\eta}_{\buj}^2\}$; (c) misspecifies the outcome conditional mean $\eta_{\buj}$; (d) misspecifies the copula generator $g$; and (e) misspecifies all three simultaneously. The estimators that are expected to be theocratically valid are highlighted in orange.  }
    \label{fig:box-SIME}
\end{figure}

\begin{figure}[ht!]
    \centering
    \includegraphics[scale=0.65]{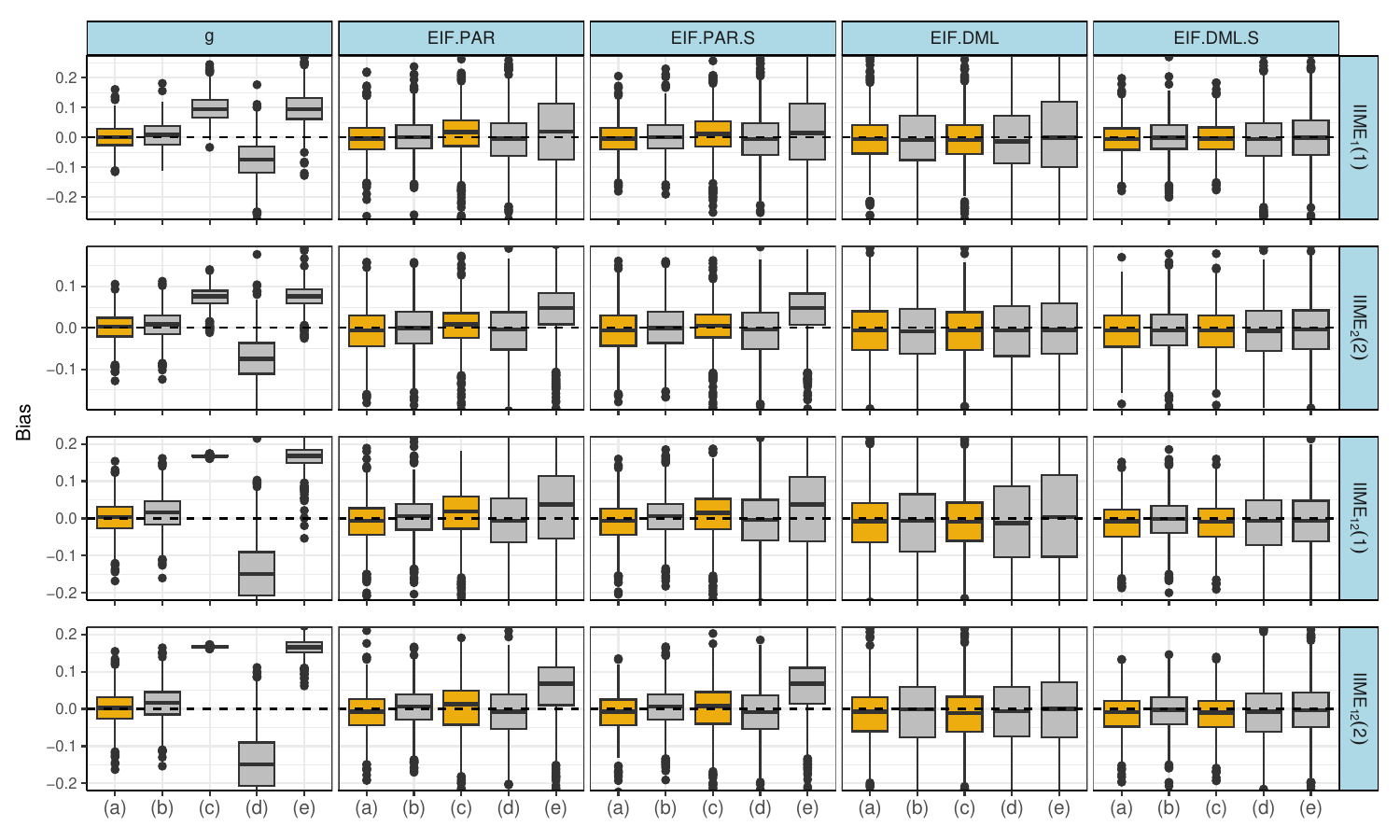}
    \caption{ Simulation results for the bias of the individual interaction mediation effects $\text{IIME}_{\{1\}}(1),,\text{IIME}_{\{2\}}(2),\text{IIME}_{\{1,2\}}(1),\text{IIME}_{\{1,2\}}(2)$. We consider the following 5 scenarios: (a) assumes all components are correctly specified; (b) misspecifies the mediator conditional means $\{\underline{\eta}_{\buj}^1,\underline{\eta}_{\buj}^2\}$; (c) misspecifies the outcome conditional mean $\eta_{\buj}$; (d) misspecifies the copula generator $g$; and (e) misspecifies all three simultaneously. The estimators that are expected to be theocratically valid are highlighted in orange. }
    \label{fig:box-IIME}
\end{figure}

\begin{figure}[ht!]
    \centering
    \includegraphics[scale=0.82]{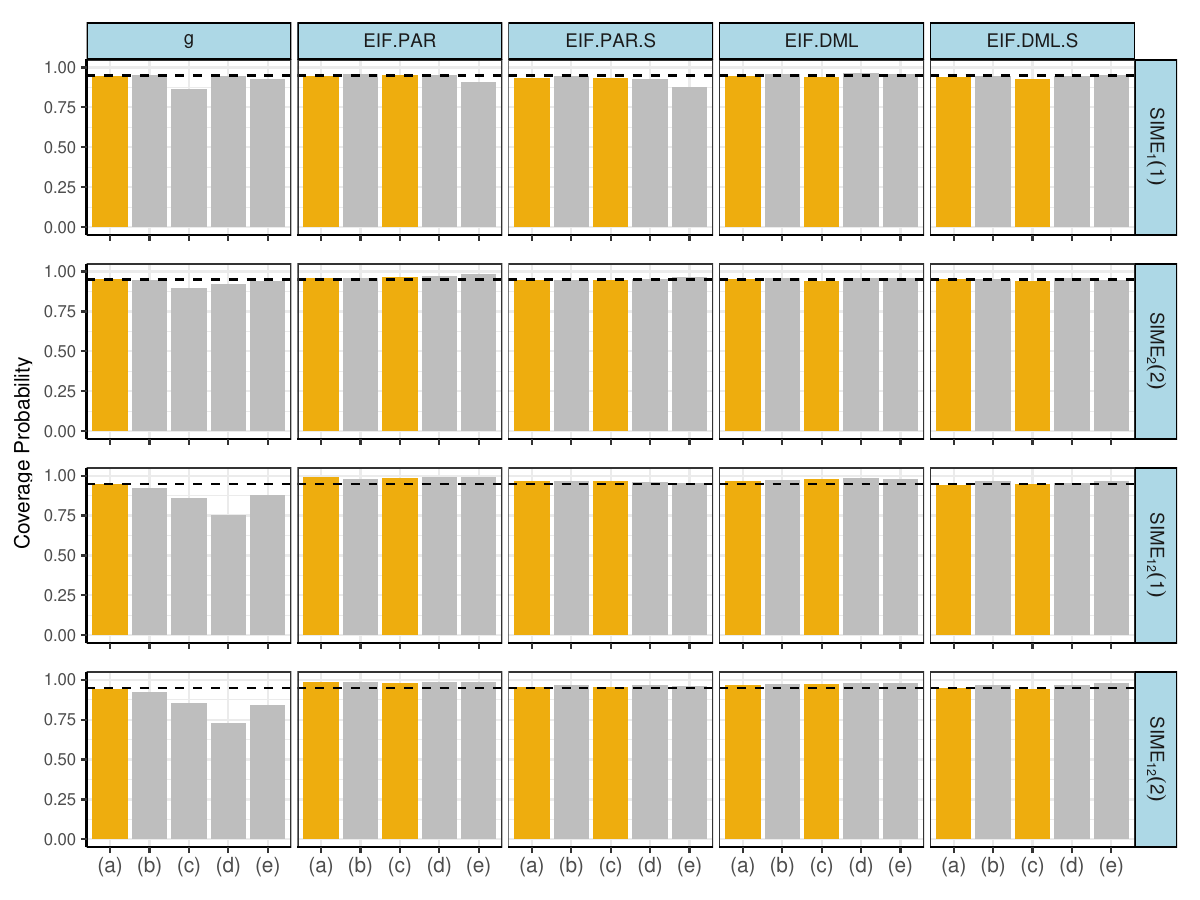}
    \caption{ Simulation results for the bias of the spillover interaction mediation effects $\text{SIME}_{\{1\}}(1),,\text{SIME}_{\{2\}}(2),\text{SIME}_{\{1,2\}}(1),\text{SIME}_{\{1,2\}}(2)$. We consider the following 5 scenarios: (a) assumes all components are correctly specified; (b) misspecifies the mediator conditional means $\{\underline{\eta}_{\buj}^1,\underline{\eta}_{\buj}^2\}$; (c) misspecifies the outcome conditional mean $\eta_{\buj}$; (d) misspecifies the copula generator $g$; and (e) misspecifies all three simultaneously. The estimators that are expected to be theocratically valid are highlighted in orange. }
    \label{fig:box-SIME-CP}
\end{figure}

\begin{figure}[ht!]
    \centering
    \includegraphics[scale=0.82]{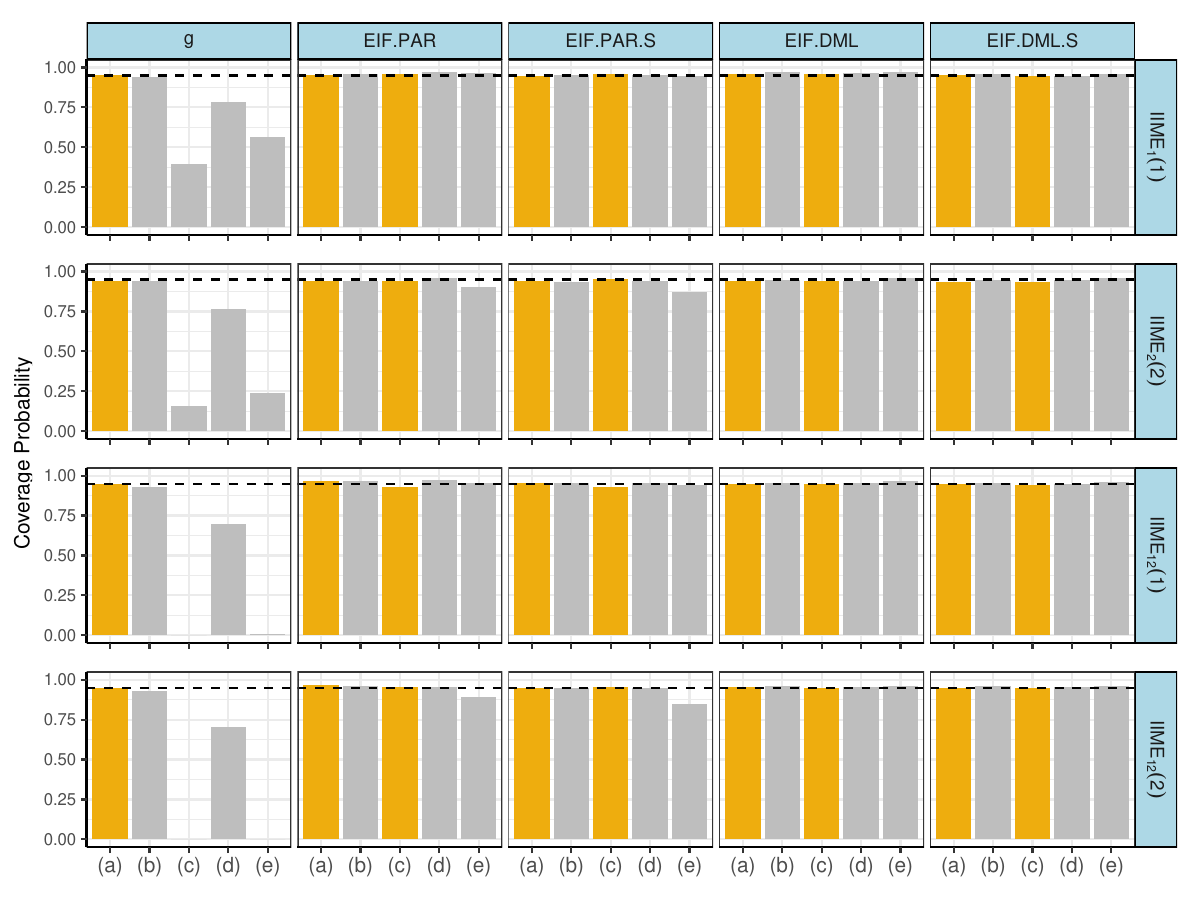}
    \caption{ Simulation results for the empirical coverage probabilities of the individual interaction mediation effects $\text{IIME}_{\{1\}}(1),,\text{IIME}_{\{2\}}(2),\text{IIME}_{\{1,2\}}(1),\text{IIME}_{\{1,2\}}(2)$. We consider the following 5 scenarios: (a) assumes all components are correctly specified; (b) misspecifies the mediator conditional means $\{\underline{\eta}_{\buj}^1,\underline{\eta}_{\buj}^2\}$; (c) misspecifies the outcome conditional mean $\eta_{\buj}$; (d) misspecifies the copula generator $g$; and (e) misspecifies all three simultaneously. The estimators that are expected to be theocratically valid are highlighted in orange. }
    \label{fig:box-IIME-CP}
\end{figure}

\clearpage

\bibliography{bibliography.bib}

\bibliographystyle{chicago}
\end{document}